%% file: disc-tech-report.tex
\documentclass[a4paper]{article}

\usepackage{fullpage}
\usepackage{amsthm}
\newtheorem{lemma}{Lemma}
\newcommand{\nat}{{\mathbb{N}}}
\newtheorem{theorem}{Theorem}
\newtheorem{proposition}{Proposition}
\newtheorem{claim}{Claim}
\newtheorem{remark}{Remark}
\newtheorem{corollary}{Corollary}
\theoremstyle{definition}
\newtheorem{example}{Example}

\input{macros}

\renewcommand{\succ}[1]{\mathbf{succ}(#1)}
\newcommand{\outdeg}[1]{\mathbf{deg}^+(#1)}
\newcommand{\proj}[2]{{[#1]}_{\mathbf{#2}}}
\newcommand{\learned}{E_{\forall}}
\newcommand{\supp}[1]{\mathbf{supp}(#1)}
\newcommand{\regret}[3]{\mathbf{reg}^{#1}_{#3}(#2)}
\newcommand{\Regret}[2]{\mathbf{Reg}_{#2}(#1)}
\newcommand{\Val}{\mathbf{Val}}
\newcommand{\aVal}{\mathbf{aVal}}
\newcommand{\cVal}{\mathbf{cVal}}
\newcommand{\mrp}{\mathbf{MRP}}
\newcommand{\mrs}{\mathbf{MRS}}
\newcommand{\PlayVal}[1]{\Val(#1)}
\newcommand{\StratVal}[4]{\Val_{#1}^{#2}(#3,#4)}
\newcommand{\out}[3]{\pi^{#1}_{#2#3}}
\newcommand{\VtcE}{V_\exists}
\newcommand{\VtcA}{V \setminus \VtcE}
\newcommand{\StrE}{\Sigma_\exists}
\newcommand{\StrA}{\Sigma_\forall}
\newcommand{\StrAllE}{\mathfrak{S}_\exists}
\newcommand{\StrAllA}{\mathfrak{S}_\forall}
\newcommand{\StrPosA}{\Sigma^1_\forall}
\newcommand{\StrWordA}{\mathfrak{W}_\forall}
\newcommand{\StrPosE}{\Sigma^1_\exists}
\newcommand{\locreg}{\mathbf{locreg}}
\newcommand{\wcopt}[1]{\mathbf{wOpt}(#1)}
\newcommand{\copt}[1]{\mathbf{cOpt}(#1)}
\newcommand{\switch}[3]{[#1\mathrel{\overset{\makebox[0pt]{\mbox{\normalfont\tiny #2}}}{\rightarrow}}#3]}

\title{Minimizing Regret in Discounted-Sum Games}

\author{Paul Hunter\thanks{Authors supported by the ERC inVEST (279499)
	project.}}
\author{Guillermo A. P\'{e}rez\thanks{Author supported by F.R.S.-FNRS
	fellowship.}}
\author{Jean-Fran\c{c}ois Raskin$^\ast$}
\affil{D\'{e}partement d'Informatique, Universit\'{e} Libre de Bruxelles (ULB)\\
\texttt{\{phunter,gperezme,jraskin\}@ulb.ac.be}}
\begin{document}

\maketitle

\begin{abstract}
	In this paper, we study the problem of minimizing regret in
	discounted-sum games played on weighted game graphs.  We give algorithms
	for the general problem of computing the minimal regret of the
	controller (Eve) as well as several variants depending on which
	strategies the environment (Adam) is permitted to use.  We also consider
	the problem of synthesizing regret-free strategies for Eve in each of
	these scenarios.
\end{abstract}

\section{Introduction}
\label{sec:intro}

Two-player games played by \eve\/ and \adam\/ on weighted graphs is a well
accepted mathematical formalism for modelling quantitative aspects of a controller (\eve) 
 interacting with its environment (\adam). The outcome of the interaction between the two
players is an infinite path in the weighted graph and a value is associated to
this infinite path using a {\em measure} such as \eg\ the mean-payoff of the
weights of edges traversed by the infinite path, or the discounted sum of those weights.
In the classical model, the game is considered to be zero sum: the two players
have antagonistic goals--one of the player want to maximize the value associated
to the outcome while the other want to minimize this value. The main solution
concept is then the notion of winning strategy and the main decision problem
asks, given a threshold $c$, whether \eve\/ has a strategy to ensure that, no matter
how \adam\/ plays, that the outcome has a value larger than or equal to $c$.

When the environment is not fully antagonistic, it is reasonable to study other
solution concepts. One interesting concept to explore is the
concept of {\em regret minimization}~\cite{bell82} which is as
follows.
%Because in many applications, there is unfortunately no optimal or winning
%strategy for the controller, yet we would like to synthesise one strategy that
%behaves \emph{as close as possible to the optimal}. Such a strategy might try
%to cooperate with the environment to obtain a good outcome instead of playing
%arbitrarily because there is no strategy that enforce a win against all the
%behaviours of the environment. To formalise this idea, we rely on the notions
%of \emph{best response} and \emph{regret}.  
When a strategy of \adam is fixed, we can identify the set of \eve's strategies
that allow her to secure the \emph{best possible outcome} against this strategy.
This constitutes \eve's \emph{best response}. Then we define the regret of a
strategy $\sigma$ of \eve as the difference between \eve's best response; 
and the payoff she secures thanks to her strategy $\sigma$. So, when
trying to minimize the regret associated to a strategy, we use best responses as
a {\em yardstick}. Let us now illustrate this with an example.

\begin{figure}
\begin{minipage}[b]{0.35\linewidth}
\begin{center}
\resizebox{0.7\textwidth}{!}{%
\begin{tikzpicture}
\node[ve,initial](A){$v_I$};
\node[va,right=of A](B){$x$};
\node[va,below=of A](C){$v$};
\node[va,right=of C](D){$y$};

\path
(A) edge node[el]{$1$} (B)
(A) edge[bend right] node[el,swap]{$0$} (C)
(C) edge[bend right] node[el,swap]{$0$} (A)
(C) edge node[el,swap]{$M$} (D)
(D) edge[loopright] node[el,swap]{$M$} (D)
(B) edge[loopright] node[el,swap]{$1$} (B)
;
\end{tikzpicture}
}
\caption{A game in which waiting is required to minimize regret.}\label{fig:bigmem}
\end{center}
\end{minipage}
\hfill
\begin{minipage}[b]{0.63\linewidth}
\begin{center}
\resizebox{0.9\textwidth}{!}{%
\begin{tikzpicture}[node distance=0.5cm]
\node[ve,initial above](A){};
\node[va,left=1cm of A,yshift=-0.5cm](B){};
\node[va,right=1cm of A,yshift=-0.5cm](C){};
\node[ve,below=of B](D){};
\node[ve,below=of C](E){};
\node[va,below=of D](G){};
\node[va,left=2cm of G](F){};
\node[va,below=of E](H){};
\node[ve,below=of F](I){};
\node[ve,below=of G](J){};
\node[ve,below=of H](K){};

\path
% from A
(A) edge[bend right] node[el,swap]{S} (B)
(A) edge[bend left] node[el]{B} (C)
% from B
(B) edge[bend right] node[el,swap]{H,$-4$} (D)
(B) edge[bend left] node[el]{L,$12$} (D)
% from C
(C) edge[bend right] node[el,swap]{H,$-2$} (E)
(C) edge[bend left] node[el]{H,$8$} (E)
% from D
(D) edge node[el,swap]{S} (F)
(D) edge node[el]{B} (G)
% from E
(E) edge node[el]{B} (H)
% from F
(F) edge[bend right] node[el,swap]{H,$-4$} (I)
(F) edge[bend left] node[el]{L,$12$} (I)
% from G
(G) edge[bend right] node[el,swap]{H,$-2$} (J)
(G) edge[bend left] node[el]{L,$8$} (J)
% from H
(H) edge[bend right] node[el,swap]{H,$-2$} (K)
(H) edge[bend left] node[el]{L,$8$} (K)
;
\end{tikzpicture}
}
  \caption{A game that models different investment strategies. \label{investment-ex}}
\end{center}
\end{minipage}
\end{figure}

\begin{table}
\begin{center}
\begin{tabular}{ c  c c c c c c}
 \hline
	 & {\sf HH} & {\sf HL} & {\sf LH} & {\sf LL} & Worst-case & Regret \\ 
 \hline
 \hline
 {\sf SS} & $-7.7616$ & $7.6048$ & $7.9784$  & $23.2848$ & $-7.7616$ & $3.8808$ \\ 
 \hline
 {\sf SB} & $-5.8408$ & $3.7632$ & $9.8392$ & $19.4432$ & $-5.8408$ & ${\bf 3.8416}$ \\ 
 \hline
 {\sf BB} & $-3.8808$ & $5.7232$ & $5.9192$ & $15.5232$ & ${\bf -3.8808}$ & $7.7616$\\ 
 \hline
 \hline
\end{tabular}
\caption{The possible rate configuration for the rate of interests are given as the first four columns, the follows the worst-case performance and the regret associated to each strategy of \eve\/ that are given in rows.
Entries in bold are the values that are maximizing the worst-case (strategy {\sf BB}) and
minimizing the regret (strategy {\sf SB}).
\label{table-regret} }
\end{center}
\end{table}

\begin{example}[Investment advice]
Consider the discounted sum game of Fig.~\ref{investment-ex}. It models the
rentability of different investment plans with a time horizon of two periods. In
the first period, it can be decided to invest in treasure bonds ({\sf B}) or to
invest in the stock market ({\sf S}). In the former case, treasure bonds ({\sf
B}) are chosen for two periods. In the latter case, after one period, there is
again a choice for either treasure bonds ({\sf B}) or stock market ({\sf S}).
The returns of the different investments depend on the fluctuation of the rate
of interests. When the rate of interests is low (L) then the return for the
stock market investments is equal to $12$ and for the treasure bonds it is equal
to $8$. When the interest rate is high (H) then the returns for the stock
market investments is equal to $-4$ and for the treasure bonds it is equal to
$-2$. To model time and take into account the inflation rate, say equal to $2$
percent, we consider a discount factor $\lambda=0.98$ for the returns. In this
example, we make the hypothesis that the fluctuation of the rate of interests is
{\em not} a function of the behavior of the investor. It means that this
fluctuation rate is either one of the following four possibilities: {\sf HH},
{\sf HL}, {\sf LH}, {\sf LL}. This corresponds to \adam\/ playing a {\em word
strategy} in our terminology. The discounted sum of returns obtained under the
$12$ different scenarios are given in Table~\ref{table-regret}.

Now, assume that you are a broker and you need to advise one of your customers regarding 
his next investment.  There are several ways to advise your customer.
First, if your customer is strongly {\em risk averse}, then you should be able
to convince him that he has to go for the treasure bonds ({\sf B}). Indeed, this
is the choice that {\em maximizes the worst-case}: if the interest rates stay
high for two periods ({\sf HH}) then the loss will be $-3.8808$ while it will be
higher for any other choices. Second, and maybe more interestingly, if
your customer tolerates some risks, then you may want to keep him happy so that
he will continue to ask for your advice in the future! Then you should propose the
following strategy: first invest in the stock market ({\sf S}) then in treasure
bonds ({\sf B}) as this strategy {\em minimizes regret}. Indeed, at the end of the two
investment periods, the actual interest rates will be known and so your
customer will evaluate your advices {\em ex-post}. So, after the two periods,
the value of the choices made {\em ex ante} can be compared to the best strategy
that could have been chosen knowing the evolution of the interest rates. The regret of
{\sf SB} is at most equal to $3.8416$ in all cases and it is minimal: the regret
of {\sf BB} can be as high as $7.7616$ if {\sf LL} is observed, and the regret
of {\sf SS} can be as high as $3.8808$.

Finally, let us remark that if the investments are done in financial markets that
are subject to different interest rates, then instead of considering the
minimization of regret against word strategies, then we could consider the
regret against all strategies. We also study this case in this paper.
\end{example}

\subparagraph{Previous works.}
In~\cite{fgr10}, we studied regret minimization in the context of reactive
synthesis for shortest path objectives. Recently in~\cite{hpr15}, we studied the
notion of regret minimization when we assume different sets of strategies from
which \adam\/ chooses. We have considered three cases: when the \adam\/
is allowed to play {\em any} strategy, when he is restricted to play a {\em
memoryless} strategy, and when he plays {\em word} strategies. We refer the
interested reader to~\cite{hpr15} for motivations behind each of these
definitions.  In that paper, we studied the regret minimization problem for the following classical
quantitative measures: $\inf$, $\sup$, $\lim \inf$, $\lim \sup$ and the
{\em mean-payoff} measure. In this paper, we complete this picture by studying the
regret minimization problem for the {\em discounted-sum} measure. Discounted-sum
is a central measure in quantitative games but we did not consider it
in~\cite{hpr15} because it requires specific techniques which are more involved
than the ones used for the other quantitative measures. For example, while for
mean-payoff objectives, strategies that minimize regret are memoryless when the
\adam\/ can play any strategy, we show in this paper that
pseudo-polynomial memory is necessary (and sufficient) 
to minimize regret in discounted-sum games.
The need for memory is illustrated by the following example.

\begin{example}
Consider the example in Figure~\ref{fig:bigmem} where $M \gg 1$.
%, which loosely 
%models the software company scenario above. 
\eve\/ can play the following strategies in this game: let $i \in \nat \cup \{
\infty \}$, and note $\sigma^i$  the strategy that first plays $i$ rounds the
edge $(v_I,v)$ and then switches to $(v_I,x)$. The regret values associated to
those strategies are as follows.

The regret of $\sigma^{\infty}$ is $\frac{1}{1-\lambda}$ and it is witnessed
when \adam\/ never plays the edge $(v,y)$. Indeed, the discounted sum of the
outcome in that case is $0$, while if \eve\/ had chosen to play $(v_I,x)$ at the
first step instead, then she would have gained $\frac{1}{1-\lambda}$.  The
regret of $\sigma^i$ is equal to the maximum between $\frac{1}{1-\lambda} -
\lambda^{2i}  \frac{1}{1-\lambda}$ and $\lambda^{2i+1} \frac{M}{1-\lambda} -
\lambda^{2i}  \frac{1}{1-\lambda}$. The maximum is either witnessed when \adam\/
never plays $(v,y)$ or plays $(v,y)$ if the edge $(v_I,x)$ has been chosen $i+1$
times (one more time compared to $\sigma^i$).

So the strategy that minimizes regret is the strategy $\sigma^N$ for $N >
\frac{-\log M}{2\log{\lambda}} - \frac{1}{2}$ (so that $\lambda^{2N+1} M <1$), {\it
i.e.} the strategy needs to count up to $N$.
%
%as this is the value obtained when choosing at the first step the edge $(v_I,x)$ while $sigma^{\omega}$ if \adam\/ never chooses the edge $(v,y)$. For all $i \in \nat$, the regret associated to  $sigma^{i}$ is equal to $\frac{1}{1-\lambda} - \lambda^i (\frac{1}{1-\lambda})$
%
%
%From the game starts, \eve can directly go to $x$ or repeat the $v_I$-$v$ cycle for a
%certain number rounds in the hope that \adam will help her reach $y$.  Note that
%trying to reach $x$ after the first round always yields a worse payoff. However,
%since $M \gg 1$, she must visit $v$ a certain number of times lest the best
%alternative strategy will be to visit $v$ once more.  	More formally, to
%minimize her regret, \eve must repeatedly visit $v$ sufficiently many times such
%that the best alternative strategy is to consider an initial deviation (to $x$).
%That is, with discount factor $\lambda$, \eve can achieve a regret of
%$\frac{1-\lambda^{2N}}{1-\lambda}$ by visiting $v$ $N$ times for $N =
%\frac{-\log(M)-1}{2\log{\lambda}}$ (\ie $\lambda^{2N+1} M = 1$).  With this
%strategy, \eve achieves a payoff of $\frac{\lambda^{2N}}{1-\lambda}$, whereas
%the best alternative strategy is either to consider deviating to $x$ initially
%(payoff: $\frac{1}{1-\lambda}$) or go to $v$ one last time, \ie $N+1$ in total,
%and subsequently to $y$ (payoff: $\lambda^{2N+1}\frac{M}{1-\lambda} =
%\frac{1}{1-\lambda}$).
\end{example}

\subparagraph{Contributions.} We describe algorithms to decide the regret
threshold problem for  games in three cases: when there is no
restriction on the strategies that \adam can play, when \adam\/ can only play memoryless strategies, and when \adam\/ can only play word strategies.  For this last case, our problem is closely
related to open problems in the field of discounted-sum automata, and we also consider variants given as $\epsilon$-gap promise problems.

We also study the complexity of the special case when the threshold is $0$, {\it i.e.} when we ask for the existence of regret free strategies. We show that that problem is sometimes easier to solve.
Our results on the complexity of both the regret threshold and the regret-free
problems are summarized in Table~\ref{tab:sum}. All our results are for fixed discount factor $\lambda$.

%We first observe that the regret value of game being $0$ does not necessarily
%imply the existence of a regret-free strategy, it only implies the existence of
%a sequence of strategies of arbitrarily small regret.  However, as a side
%product of the techniques used to obtain our results on the regret threshold
%problem, we obtain algorithms to synthesize strategies of \eve which ensure no
%more than the regret of the arena.  That is, in all the cases where we establish
%a complexity upper bound for the regret threshold problem we are able to extract
%at least one strategy of \eve which witnesses this minimal regret.

\begin{table}
\begin{center}
\small
\begin{tabular}{|l|c|c|c|}
	\hline
	& Any strategy & Memoryless strategies & Word strategies \\
	\hline\hline
	regret threshold &
		\NP & \PSPACE~(Thm.~\ref{thm:memlessadversary}),
		& \PSPACE-c ($\epsilon$-gap) \\
		& (Thm.~\ref{thm:anyadversary})  & \coNP-h
		(Thm.~\ref{thm:conp-hardness}) &
		(Thm.~\ref{thm:epsilon-gap-pspace},
		Thm.~\ref{thm:eloquent-pspace-hardness-epsilon}) \\
	\hline
	regret-free
		& \P &
		\PSPACE~(Thm.~\ref{thm:posadversary-zero}), &
		\NP-c~(Thm.~\ref{thm:eloquentadversary-zero}) \\
		& (Thm.~\ref{thm:anyadversary-zero})
		& \coNP-h (Thm.~\ref{thm:conp-hardness}) & \\
	\hline
\end{tabular}
\caption{Complexity of deciding the regret threshold and regret-free problems for fixed
$\lambda$.}\label{tab:sum}
\end{center}
\end{table}

\subparagraph{Other related works.}
%
%In~\cite{hp12} regret minimization is considered in the context of matrix games,
%whereas here we consider infinite games played in graphs. In the latter setting,
%iterated regret minimization has already been studied -- albeit only for the
%shortest path problem -- in~\cite{fgr10}.  Restrictions on how \adam is allowed
%to play were not considered there. Furthermore, as we do not consider an
%explicit objective for \adam, we do not consider iteration of regret
%minimization here.
%
A Boolean version of our \emph{regret-free strategies} has been described
in~\cite{df11}. In that paper, they are called \emph{remorse-free strategies}.
These correspond to strategies which minimize regret in games with
$\omega$-regular objectives. They do not establish lower bounds on the
complexity of realizability or synthesis of remorse-free strategies and they
only consider word strategies for \adam.

In~\cite{hpr15}, we established that regret minimization when \adam\/ plays word strategies only 
is a generalization of the notion of \emph{good-for-games
automata}~\cite{hp06} and \emph{determinization by pruning} (of a
refinement)~\cite{akl10}.
%In particular, we have shown the {\em word strategy} variant is
%closely related to the notions of determinism and non-determinism in automata
%theory. More precisely, regret minimization corresponds to a quantitative
%extension of the notion of good-for-games automata introduced in~\cite{hp06}. We
%also point out in~\cite{hpr15} that regret minimization is closely related to
%the formalization of online algorithm analysis using weighted automata
%introduced by Kupferman et al. in~\cite{akl11}.

The notion of regret is closely related to the notion of 
competitive ratios used for the analysis of online algorithms~\cite{SleatorT84}: the performance 
of an online
algorithm facing uncertainty (\eg about the future incoming requests or data)
is compared to the performance of an offline algorithm (where uncertainty is
resolved). According to this quality measure, an online algorithm is better if
its performance is closer to the performance of an optimal offline solution.

\subparagraph{Structure of the paper.}
In Sect.~2, we introduce the necessary definitions and notations.  In Sect.~3,
we study the minimization of regret when the second player plays any strategy.
Finally, in Sect.~4, we study the minimization of regret when the second player
plays a memoryless strategy and in Sect.~5 when he plays a word strategy.

\section{Preliminaries}\label{sec:prelim}
A \emph{weighted arena} is a tuple $G = (V,\VtcE,E,w,v_I)$ where $(V,E,w)$ is an
edge-weighted graph (with rational weights), $\VtcE \subseteq V$, and $v_I \in
V$ is the initial vertex.  For a given $v \in V$ we denote by $\succ{u}$ the set
of \emph{successors of $u$ in $G$}, that is the set $\{v \in V \st (u,v) \in E\}$.
We assume w.l.o.g. that no vertex is a sink, \ie~$\forall v \in V
: |\succ{v}| > 0$, and that every \eve vertex has more than one successor,
\ie~$\forall v \in \VtcE : |\succ{v}| > 1$.  In the sequel, we depict vertices
in $\VtcE$ with squares and vertices in $V \setminus \VtcE$ with circles. We
denote the maximum absolute value of a weight in a weighted arena by $W$.

A \emph{play} in a weighted arena is an infinite sequence of vertices $\pi =
v_0v_1 \dots$ where $(v_i,v_{i+1}) \in E$ for all $i$.  Given a
play $\pi = v_0v_1 \dots$ and integers $k,l$ we define $\pi[k..l] \defeq
v_k\dots v_l$, $\pi[..k] \defeq \pi[0..k]$, and $\pi[l..] \defeq
v_lv_{l+1}\dots$, all of which we refer to as \emph{play prefixes}. To improve
readability, we try to adhere to the following convention: use $\pi$ to denote
plays and $\rho$ for play prefixes. The \emph{length of a play} $\pi$, denoted
$|\pi|$, is $\infty$, and the \emph{length of a play prefix} $\rho = v_0 \dots
v_n$, \ie~$|\rho|$, is $n+1$.\todo{Added a note on when $\pi$ and when
$\rho$, also added the notation for length} 

A \emph{strategy for} \eve (\adam) is a function $\sigma$ that maps play
prefixes ending with a vertex $v$ from $\VtcE$ ($V \setminus \VtcE$) to a
successor of $v$.  A strategy has memory $m$ if it can be realized as the output
of a finite state machine with $m$ states (see \eg~\cite{hpr14} for a formal
definition). A \emph{memoryless (or positional) strategy} is a strategy with
memory $1$, that is, a function that only depends on the last element of the
given partial play. A play $\pi = v_0v_1 \dots$ is \emph{consistent with a
strategy} $\sigma$ for \eve (\adam) if whenever $v_i \in \VtcE$ ($v_i \in
V\setminus \VtcE$), then $\sigma(\pi[..i]) = v_{i+1}$. We denote by
$\StrAllE(G)$ ($\StrAllA(G)$) the set of all strategies for \eve (\adam) and by
$\StrE^{m}(G)$ ($\StrA^{m}(G)$) the set of all strategies for \eve (\adam) in
$G$ that require memory of size at most $m$, in particular $\StrPosE(G)$
($\StrPosA(G)$) is the set of all memoryless strategies for \eve (\adam) in $G$.
We omit $G$ if the context is clear.

Given strategies $\sigma, \tau$, for \eve and \adam respectively, and $v \in V$,
we denote by $\out{v}{\sigma}{\tau}$ the unique play starting from $v$ that is
consistent with $\sigma$ and $\tau$. If $v$ is omitted, it is assumed to be
$v_I$.

A \emph{weighted automaton} is a tuple $\Gamma=(Q, q_I, A, \Delta, w)$ where $A$
is a finite alphabet, $Q$ is a finite set of states, $q_I$ is the initial state,
$\Delta \subseteq Q \times A \times Q$ is the transition relation, $w : \Delta
\rightarrow \mathbb{Q}$ assigns weights to transitions. A \emph{run} of $\Gamma$
on a word $a_0 a_1 \dots \in A^\omega$ is a sequence $\rho = q_0 a_0 q_1 a_1
\dots \in (Q\times A)^\omega$ such that $(q_i,a_i,q_{i+1}) \in \Delta$, for all
$i \ge 0$, and has \emph{value} $\Val(\rho)$ determined by the sequence of
weights of the transitions of the run and the payoff function $\Val$. The value
$\Gamma$ assigns to a word $w$, $\Gamma(w)$, is the supremum of the values of
all runs on the word. We say the automaton is deterministic if $\Delta$ is
functional.

\subparagraph{Safety games.}
A \emph{safety game} is played on a non-weighted arena by \eve and \adam. The
goal of \eve is to perpetually avoid traversing edges from a set of \emph{bad
edges}, while \adam attempts to force the play through any unsafe edge.  More
formally, a safety game is a tuple $(G,B)$ where $G = (V,\VtcE,E,v_I)$ is a
non-weighted arena and $B \subseteq E$ is the set of bad edges. A play $\pi =
v_0 v_1 \dots$ is \emph{winning for \eve} if $(v_i,v_{i+1}) \not\in B$, for all
$i \ge 0$, and it is \emph{winning for \adam} otherwise. A strategy for \eve
(\adam) is winning for her (him) in the safety game if all plays consistent with
it are winning for her (him). A player wins the safety game if (s)he has a
winning strategy.

\begin{lemma}[from~\cite{ag11}]\label{lem:facts-safety}
	Safety games are \emph{positionally determined}: either \eve has a
	positional winning strategy or \adam has a positional strategy.
	Determining the winner in a safety game is decidable in linear time.
\end{lemma}

\subparagraph{Discounted-sum.}
A play in a weighted arena, or a run in a weighted automaton, induces an
infinite sequence of weights.  We define below the \emph{discounted-sum} payoff
function which maps finite and infinite sequences of rational weights to real
numbers. In the sequel we refer to a weighted arena together with a payoff
function as a \emph{game}. Formally, given a sequence of weights
$\chi=x_0 x_1 \dots$ of length $n \in \mathbb{N} \cup \{\infty\}$, the
\emph{discounted-sum} is defined by a rational discount factor $\lambda \in
(0,1)$:
\( \textstyle
	\discfun{\lambda}(\chi) \defeq \sum_{i=0}^n \lambda^i 
	x_i.
\)
For convenience, we apply payoff functions directly to plays, runs, and
prefixes. For instance, given a play or play prefix $\pi = v_0 v_1 \dots$ we
write $\discfun{\lambda}(\pi)$ instead $\discfun{\lambda}(w(v_0,v_1) w(v_1,v_2)
\dots)$.

Consider a fixed weighted arena $G$, and a discounted-sum payoff function
$\Val=\discfun{\lambda}$ for some $\lambda \in (0,1)$.  Given strategies
$\sigma, \tau$, for \eve and \adam respectively, and $v \in V$, we denote the
value of $\out{v}{\sigma}{\tau}$ by
\(
	\StratVal{G}{v}{\sigma}{\tau} \defeq
	\PlayVal{\out{v}{\sigma}{\tau}}.
\)
We omit $G$ if it is clear from the context. If $v$ is omitted, it is
assumed to be $v_I$.

\subparagraph*{Antagonistic \& co-operative values.}
Two values associated with a weighted arena that we will use
throughout are the \emph{antagonistic and co-operative values},
defined for plays from a vertex $v \in V$ as:
\[ \textstyle
	\aVal^v(G) \defeq \sup_{\sigma \in \StrAllE} \inf_{\tau \in \StrAllA}
		\StratVal{}{v}{\sigma}{\tau}
	\qquad
	\cVal^v(G) \defeq \sup_{\sigma \in \StrAllE} \sup_{\tau \in \StrAllA}
		\StratVal{}{v}{\sigma}{\tau}.
\]
Again, if $G$ is clear from the context it will be omitted, and if $v$ is
omitted it is assumed to be $v_I$.  We note that, as memoryless strategies are
sufficient in discounted-sum games~\cite{zp96}, $\aVal$ can be computed in time
polynomial (in $\frac{1}{1-\lambda}$, $|V|$, and $\log_2 W$). If $\lambda$ is
given as part of the input, this becomes exponential (in the size of the input).
Regardless of whether $\lambda$ is part of the input, $\cVal$ is computable in
polynomial time, determining if $\aVal$ is bigger (or smaller) than a given
threshold is decidable and in $\NP \cap \coNP$, and the values $\cVal$ and
$\aVal$ are representable using a polynomial number of bits.

A useful observation used by Zwick and Paterson in~\cite{zp96}, and which is
implicitly used throughout this work, is the following.
\begin{remark}
	For all $u \in V$, $\cVal^u(G) = \max\{w(u,v) + \lambda\cVal^v(G)
	\st (u,v) \in E\}$. For all $u \in \VtcE$, $\aVal^u(G) = \max\{w(u,v) +
	\lambda\aVal^v(G) \st (u,v) \in E\}$. For
	all $u \in \VtcA$, $\aVal^u(G) = \min\{w(u,v) +
	\lambda\aVal^v(G) \st (u,v) \in E\}$.
\end{remark}

We say a strategy $\sigma$ for \eve is \emph{worst-case optimal (maximizing)}
from $v \in V$ if it holds that $\inf_{\tau \in \StrAllA} \StratVal{
}{v}{\sigma}{\tau} = \aVal^v(G)$. Similarly, a strategy $\tau$ for \adam is
\emph{worst-case optimal (minimizing)} from $v \in V$ if it holds that
$\sup_{\sigma \in \StrAllE} \StratVal{ }{v}{\sigma}{\tau} = \aVal^v(G)$. Also, a
pair of strategies $\sigma,\tau$ for \eve and \adam, respectively, is said to be
\emph{co-operative optimal} from $v \in V$ if $\StratVal{ }{v}{\sigma}{\tau} =
\cVal^v(G)$.

\begin{lemma}[from~\cite{zp96}]\label{lem:exist-strats}
	The following hold:
	\begin{itemize}
		\item there exists $\sigma \in \StrAllE$ which is worst-case
			optimal maximizing from all $v \in V$,
		\item there exists $\tau \in \StrAllA$ which is worst-case
			optimal minimizing from all $v \in V$,
		\item there are $\sigma \in \StrAllE$ and $\tau \in \StrAllA$
			which are co-operative optimal from all $v \in V$.
	\end{itemize}
\end{lemma}

We now recall the definition of a \emph{strongly co-operative optimal} strategy
$\sigma$ for \eve. Formally, for any play prefix $\rho = v_0 \dots v_n$
consistent with $\sigma$, and such that $v_n \in \VtcE$ if $\sigma(\rho) = v'$,
then $v' \in \copt{v_n}$; where $\copt{u} \defeq \{v \in V \st (u,v) \in E
\text{ and } \cVal^{u}(G) = w(u,v) + \lambda \cVal^{v}(G)\}$. Finally,
we define a new type of strategy for \eve: \emph{co-operative worst-case
optimal} strategies. A strategy is of this type if, for any play prefix $\rho =
v_0 \dots v_n$ consistent with $\sigma$, and such that $v_n \in \VtcE$, if
$\sigma(\rho) = v'$ then $v' \in \wcopt{v_n}$ and 
\[
	w(v_n,v') + \lambda  \cVal^{v'}(G) = \max\{w(v_n,v'') +
	\lambda  \cVal^{v''}(G) \st v'' \in \wcopt{v_n}\},
\]
where $\wcopt{u} \defeq \{v \in V \st (u,v) \in E \text{ and } \aVal^u(G) =
w(u,v) + \lambda  \aVal^v(G)\}$.

It is not hard to verify that strategies of the above types always exist for \eve.
\begin{lemma}\label{lem:exist-strange-strats}
	There exist strongly co-operative optimal strategies and co-operative
	worst-case optimal strategies for \eve.
\end{lemma}
%\begin{proof}
%	We actually argue that positional strategies of both kinds exist.
%
%	The first claim follows from the remark we made earlier. That is, since
%	for all $u \in V$, $\cVal^u(G) = \max\{w(u,v) + \lambda\cVal^v(G) \st
%	(u,v) \in E\}$ then we can easily construct a strategy which chooses,
%	for every $u \in \VtcE$, a vertex from $\copt{v}$.
%
%	The argument for the second claim is only slightly more complicated. As
%	above, we can easily construct a strategy which chooses, from every $u
%	\in \VtcE$, a successor from $\wcopt{u}$. Further, since the arena is
%	finite, there must be at least one vertex $v' \in \wcopt{u}$ which
%	witnesses the required maximum; we choose said vertex.
%\end{proof}

\subparagraph{Regret.}
Let $\StrE \subseteq \StrAllE$ and $\StrA \subseteq \StrAllA$ be sets of
strategies for \eve and \adam respectively.  Given $\sigma \in \StrE$ we define
the \emph{regret of $\sigma$ in $G$ w.r.t. $\StrE$ and $\StrA$} as:
\[ \textstyle
	\regret{\sigma}{G}{\StrE,\StrA} \defeq \sup_{\tau \in \StrA}
	(\sup_{\sigma' \in \StrE} \StratVal{}{}{\sigma'}{\tau} -
	\StratVal{}{}{\sigma}{\tau}).
\]
A strategy $\sigma$ for \eve is then said to be \emph{regret-free} w.r.t.
$\StrE$ and $\StrA$ if $\regret{ \sigma}{G}{\StrE,\StrA} = 0$.  We define the
\emph{regret of $G$ w.r.t. $\StrE$ and $\StrA$} as: 
\[ \textstyle
	\Regret{G}{\StrE,\StrA} \defeq \inf_{\sigma \in \StrE}
		\regret{\sigma}{G}{\StrE,\StrA}.
\]
When $\StrE$ or $\StrA$ are omitted from $\regret{{}}{\cdot}{{}}$ and
$\Regret{\cdot}{{}}$ they are assumed to be the set of all strategies for \eve
and \adam. 

In the unfolded definition of the regret of a game, \ie
\[ \textstyle
	\Regret{G}{\StrE,\StrA} \defeq \inf_{\sigma \in \StrE} \sup_{\tau \in
	\StrA} (\sup_{\sigma' \in \StrE} \StratVal{}{}{\sigma'}{\tau} -
	\StratVal{}{}{\sigma}{\tau}),
\]
let us refer to the witnesses $\sigma$ and $\sigma'$ as the \emph{primary
strategy} and the \emph{alternative strategy} respectively. Observe that for any
primary strategy for \eve and any one strategy for \adam, we can assume \adam
plays to maximize the payoff (\ie~co-operates) against the alternative strategy
once it deviates (necessarily at an \eve vertex) or to minimize against the
primary strategy---again, once it deviates. Indeed, since the deviation
yields different histories, the two strategies for \adam can be combined without
conflict. More formally,\todo{Rewrote the line about deviation and maximizing}

\begin{lemma}\label{lem:combine-behaviors1}
	Consider any $\sigma \in \StrAllE$, $\tau \in \StrAllA$, and
	corresponding play $\out{}{\sigma}{\tau} = v_0 v_1 \dots$. For all $i
	\ge 0$ such that $v_i \in \VtcE$, for all $v' \in
	\succ{v_i} \setminus \{v_{i+1}\}$ there exist $\sigma' \in \StrAllE$,
	$\tau' \in \StrAllA$ for which
	\begin{inparaenum}[$(i)$]
		\item $\out{}{\sigma'}{\tau}[..i+1] = \out{}{\sigma}{\tau}[..i]
			\cdot v'$,
		\item $\PlayVal{\out{}{\sigma'}{\tau'}[i+1..]} = \cVal^{v'}(G)$,
			and
		\item $\out{}{\sigma}{\tau} = \out{}{\sigma}{\tau'}$.
	\end{inparaenum}
\end{lemma}
%
%\noindent
%Lemma~\ref{lem:combine-behaviors1} in fact follows from results in~\cite{zp96}.
%More specifically, it is known that there are strategies for both players which
%ensure $\cVal$. Furthermore, from any vertex $v \in V$, \eve has a strategy to
%ensure a payoff of at least $\aVal^v(G)$ and \adam has a strategy to ensure a
%payoff of at most $\aVal^v(G)$.
%Thus, one could further assume that \adam plays
%to minimize against the primary strategy while maximizing against the
%alternative one.
%
\begin{lemma}\label{lem:combine-behaviors2}
	Consider any $\sigma \in \StrAllE$, $\tau \in \StrAllA$, and
	corresponding play $\out{}{\sigma}{\tau} = v_0 v_1 \dots$. For all $i
	\ge 0$ such that $v_i \in \VtcE$, for all $v' \in
	\succ{v_i} \setminus \{v_{i+1}\}$ there exist $\sigma' \in \StrAllE$,
	$\tau' \in \StrAllA$ for which
	\begin{inparaenum}[$(i)$]
		\item $\out{}{\sigma'}{\tau}[..i+1] = \out{}{\sigma}{\tau}[..i]
			\cdot v' = \out{}{\sigma}{\tau'}[..i] \cdot v'$,
		\item $\PlayVal{\out{}{\sigma'}{\tau'}[i+1..]} = \cVal^{v'}(G)$,
			and
		\item $\PlayVal{\out{}{\sigma}{\tau'}[i+1..]} \le
			\aVal^{v_{i+1}}(G)$.
	\end{inparaenum}
\end{lemma}\todo{Fixed the last claim $(iii)$ which had to be $\le$ instead of
equality}

Both claims follow from the definitions of strategies for \eve and \adam and
from Lemma~\ref{lem:exist-strats}.

In the remaining of this work, we will assume that $\lambda$ is \textbf{not} given as
part of the input.\todo{Something else we want to assume?}

\section{Regret against all strategies of \adam}\label{sec:any-adversary}
In this section we describe an algorithm to compute the (minimal) regret of a
discounted-sum game when there are no restrictions placed on the strategies of
\adam. The algorithm can be implemented by an alternating machine guaranteed to
halt in polynomial time. We show that the regret \emph{value} of any game is
achieved by a strategy for \eve which consists of two strategies, the first
choosing edges which lead to the optimal co-operative value, the second choosing
edges which ensure the antagonistic value. The switch from the former to the
latter is done based on the ``local regret'' of the vertex (this is formalized
in the sequel). The latter allows us to claim \NP-membership of the regret
threshold problem.
%We also show that computing the regret is at least as hard as
%computing the antagonistic value of the associated discounted-sum game.
The
following theorem summarizes the bounds we obtain: 
\begin{theorem}\label{thm:anyadversary}
	Deciding if the regret value is less than a given threshold (strictly or
	non-strictly), playing against all strategies of \adam, is in \NP.
	%, and the (antagonistic value) threshold
	%problem for discounted-sum games reduces in polynomial time to it when
	%$\lambda$ is given as part of the input.  For fixed $\lambda$ the problem 
	%is in \PSPACE.
\end{theorem}

Let us start by formalizing the concept of
\emph{local regret}. Given a play or play prefix $\pi = v_0 \dots$ and integer
$0 \le i < |\pi|$ such that $v_i \in \VtcE$, define $\locreg(\pi,i)$ as follows:
\[
	\begin{cases}
		\lambda^i \left(\cVal^{v_i}_{\lnot v_{i+1}}(G) -
		\PlayVal{\pi[i..]} \right) & \text{if } \pi \text{ is a play,}\\
		\lambda^i \left(\cVal^{v_i}_{\lnot v_{i+1}}(G) -
		\PlayVal{\pi[i..j]}\right) -
		\lambda^j \aVal^{v_j}(G) & \text{if } \pi \text{ is a prefix
		of length } j+1 > i+1,\\
		\lambda^i \left(\cVal^{v_i}(G) -
		\aVal^{v_i}(G) \right) & \text{if } \pi \text{ is a prefix
		of length } i+1,
	\end{cases}
\]
where 
\(
	\cVal^{v_i}_{\lnot v_{i+1}}(G) = \max \{ w(v_i,v) + \lambda 
		\cVal^{v}(G) \st (v_i,v) \in E\text{ and } v \neq v_{i+1}\}.
\)
Intuitively, for $\pi$ a play, $\locreg(\pi,i)$ corresponds to the difference
between the value of the best \emph{deviation} from position $i$ and the value of
$\pi$. For $\pi$ a play prefix, $\locreg(\pi,i)$ assumes that after position $j
= |\pi| - 1$ \eve will play a worst-case optimal strategy.

\subparagraph*{Deciding 0-regret.}
%We will now make a small detour and focus on the problem of deciding if
%regret-free strategies for \eve exists. 
We will now argue that the problem of
determining whether \eve has a regret-free strategy can be
decided in polynomial time. Furthermore, if no such strategy for \eve exists, we
will extract a strategy for \adam which, against any strategy of \eve, ensures
non-zero regret. To do so, we will reduce the problem to that of deciding
whether \eve wins a safety game. The unsafe edges are determined by a function
of the antagonistic and co-operative values of the original game. %(Recall that
%the values are computable in polynomial time and, if $\lambda$ is not fixed,
%$cVal$ is still computable in polynomial time while the threshold problem for
%$\aVal$ is in $\NP \cap \coNP$.)
%-- the decision problem is in $\NP \cap \coNP$).
Critically, the game is played on the same arena as the original regret game. 

\begin{theorem}\label{thm:anyadversary-zero}
	Deciding if the regret value is $0$, playing against all strategies of
	\adam, is in \P.
	%~for fixed $\lambda$, and in 
	%$\NP \cap \coNP$ when $\lambda$ is not fixed.
\end{theorem}
\begin{proof}
	We define a partition of the edges leaving vertices from $\VtcE$ into
	good and bad for \eve. A bad edge is one which witnesses non-zero local
	regret.
	%Intuitively, a bad edge is such that the
	%best value \eve can ensure after taking it is strictly smaller than the
	%best alternative option.
	We then show that \eve can ensure a regret
	value of $0$ if and only if she has a strategy to avoid ever traversing
	bad edges. More formally, let us assume a given weighted arena $G = (V,
	\VtcE, v_I, E, w)$ and a discount factor $\lambda \in (0,1)$.  We define
	the set of bad edges $\mathcal{B} \defeq \{ (u,v) \in E \st u \in \VtcE$
	and $w(u,v) + \lambda  \aVal^v(G) < \cVal^{u}_{\lnot
	v}(G)\}$.

	Note that strategies for either player in the newly defined safety game
	are also strategies for them in the original game (and vice versa as
	well).  We now claim that winning strategies for \adam in the
	safety game $\hat{G} = (V,\VtcE,v_I,E,\mathcal{B})$ ensure that,
	regardless of the strategy of \eve, its regret will be strictly
	positive. The idea behind the claim is that, \adam can force to traverse
	a bad edge and from there, play adversarially against the primary
	strategy and co-operatively with an alternative strategy.\todo{Added the
	intuition recommended by JF}
	\begin{claim}\label{cla:strat-transfer-adam}
		If $\tau \in \StrAllA$ is a winning strategy for \adam in
		$\hat{G}$, then there exist $\tau' \in \StrAllA$ and $\sigma'
		\in \StrAllE$ such that
		\(
			\forall \sigma \in \StrAllE:
			\StratVal{}{}{\sigma'}{\tau'} -
			\StratVal{}{}{\sigma}{\tau'} \ge \lambda^{|V|} 
			\min\{ \cVal^{u}_{\lnot v}(G) - w(u,v) - \lambda 
			\aVal^v(G) \st (u,v) \in \mathcal{B} \text{ and } u \in \VtcE \} >
			0.
		\)
	\end{claim}
	The claim follows from the definitions and
	Lemma~\ref{lem:combine-behaviors2}. Conversely, winning strategies for
	\eve in $\hat{G}$ are actually regret-free.
	\begin{claim}\label{cla:strat-transfer-eve}
		If $\sigma \in \StrAllE$ is a winning strategy for \eve in
		$\hat{G}$, then $\regret{\sigma}{G}{} = 0$.
	\end{claim}
	Our argument to prove this claim requires we
	first show that a winning strategy for \eve ensures the antagonistic
	value of $G$ from $v_I$. For completeness, a proof for this claim is
	included in appendix. 

	The desired result then follows from Lemma~\ref{lem:facts-safety} and
	from the fact that membership of an edge in $\mathcal{B}$ can be decided
	by computing $cVal$ and a threshold query regarding $\aVal$, thus in
	polynomial time.
	%, if $\lambda$ is fixed; in $\NP \cap \coNP$ otherwise.
\end{proof}
We observe the proof of Theorem~\ref{thm:anyadversary-zero}---more
precisely, Claim~\ref{cla:strat-transfer-adam}---implies that, if there is no
regret-free strategy for \eve in a game, then the regret of the game is at least
$\lambda^{|V|}$ times the smallest local regret labelling the bad edge from
$\mathcal{B}$ which \adam can force. More formally:
\begin{corollary}\label{cor:lower-bound}
	If no regret-free strategy for \eve exists in $G$, then 
	\(
		\Regret{G}{} \ge a_G
	\)
	where
	\(
		a_G \defeq \lambda^{|V|}  \min\{
			\locreg(uv,0)
			%\cVal^{u}_{\lnot v}(G) -
			%w(u,v) - \lambda  \cdot \aVal^v(G)
			\st u \in \VtcE \text{ and } (u,v) \in \mathcal{B}\}.
	\)
\end{corollary}

\subparagraph{Deciding r-regret.}
It will be useful in the sequel to define the \emph{regret of a play} and the
\emph{regret of a play prefix}.  Given a play $\pi = v_0 v_1 \dots$, we define
the regret of $\pi$ as:
\[ \textstyle
	%\regret{}{\pi}{} \defeq \sup \{ \lambda^i  (\cVal^{v_i}_{\lnot
	%	v_{i+1}}(G) - \PlayVal{\pi[i..]}) \st v_i \in V_\exists\} \cup
	%	\{0\},
	\regret{}{\pi}{} \defeq \left( \sup \{ \locreg(\pi,i) \st v_i \in
	V_\exists\} \cup \{0\} \right).
\]
%where 
%\(
%	\cVal^{v_i}_{\lnot v_{i+1}}(G) = \max \{ w(v_i,v) + \lambda 
%		\cVal^{v}(G) \st (v_i,v) \in E\text{ and } v \neq v_{i+1}\}.
%\)
%We call the difference $\cVal^{v_i}_{\lnot v_{i+1}} - \PlayVal{\pi[i..]}$ the
%\emph{local regret} at $v_i$. 
Intuitively, the %discounted
local regrets give lower bounds for the overall regret of a play. We will also
let the \emph{regret of a play prefix} $\rho = v_0 \dots v_j$ be equal to
\[ \textstyle
	\max \left(\{ \lambda^i 
		(\cVal^{v_i}_{\lnot v_{i+1}}(G) - \PlayVal{\rho[i..j]}) \st 0 \le
		i < j \text{ and } v_i \in V_\exists\} \cup
		\{0\}\right).
\]

\begin{figure}
\begin{minipage}[t]{.48\linewidth}
\begin{center}
\resizebox{0.8\textwidth}{!}{%
\begin{tikzpicture}
	\node (root) at (0,5) {$v_I$};
	\node (leftcorner) at (-4,0) {};
	\node (rightcorner) at (4,0) {};

	\path[-]
	(root) edge (leftcorner)
	(root) edge (rightcorner);

	\node (bottom) at (0,0) {$\out{}{\sigma}{\tau}$};
	\path [->,decoration={zigzag,segment length=4,amplitude=.9,
		post=lineto,post length=2pt}]
	(root) edge[decorate] (bottom);

	\node[ve,fill=white] (alt) at (0,3) {$v_i$};

	\node (altbottom) at (3.5,0) {$\out{}{\sigma'}{\tau}$};
	\path [->,green,decoration={zigzag,segment length=4,amplitude=.9,
		post=lineto,post length=2pt}]
	(alt) edge[decorate] (altbottom);
\end{tikzpicture}
}
\end{center}
\caption{Depiction of a play and a ``better alternative play''.}
\label{fig:deviation}
\end{minipage}
\hfill
\begin{minipage}[t]{.48\linewidth}
\begin{center}
\resizebox{0.8\textwidth}{!}{%
\begin{tikzpicture}
	\node (root) at (0,5) {$v_I$};
	\node (leftcorner) at (-4,0) {};
	\node (rightcorner) at (4,0) {};

	\path[-]
	(root) edge (leftcorner)
	(root) edge (rightcorner);

	\node (bottom) at (0,0) {$\out{}{\sigma}{\tau}$};
	\path [->,decoration={zigzag,segment length=4,amplitude=.9,
		post=lineto,post length=2pt}]
	(root) edge[decorate] (bottom);

	\node[ve,fill=white] (alt) at (0,3) {$v_i$};
	\node (altbottom) at (3.5,0) {$\out{}{\sigma'}{\tau}$};
	\path [->,green,decoration={zigzag,segment length=4,amplitude=.9,
		post=lineto,post length=2pt}]
	(alt) edge[decorate] (altbottom);

	\draw[dashed,-,gray] (-3.8,2.5) -- (4,2.5);
	\node[gray] at (-4,2.5) {$j$};

	\node[ve,fill=white] (alt2) at (0,2) {$v_k$};
	\node (altbottom2) at (-2.5,0) {$\out{}{\sigma''}{\tau}$};
	\path [->,red,decoration={zigzag,segment length=4,amplitude=.9,
		post=lineto,post length=2pt}]
	(alt2) edge[decorate] (altbottom2);
\end{tikzpicture}
}
\end{center}
\caption{A deviation from $v_k$ cannot be a best alternative to $\out{
}{\sigma}{\tau}$ if $j \ge N(\StratVal{ }{ }{\sigma'}{\tau} -
\StratVal{ }{ }{\sigma}{\tau})$.}
\label{fig:deviation-bound}
\end{minipage}
\end{figure}

Let us give some more intuition regarding the regret of a play. Consider a pair
of strategies $\sigma$ and $\tau$ for \eve and \adam, respectively. Suppose there is
an alternative strategy $\sigma'$ for \eve, such that, against $\tau$, the
obtained payoff is greater than that of $\out{ }{\sigma}{\tau}$. It should be
clear that this implies there is some position $i$ such that, from vertex $v_i
\in \VtcE$ $\sigma'$ and $\tau$ result in a different play from $\out{
}{\sigma}{\tau}$ (see Figure~\ref{fig:deviation}).
%Notice that the difference
%between the payoffs of these two plays is equivalent to the difference of their
%payoffs, had we started the game from $v_i$---discounted by $\lambda^i$, that
%is.
We will sometimes refer to this deviation, \ie the play $\out{
}{\sigma'}{\tau}$, as a \emph{better alternative} to $\out{ }{\sigma}{\tau}$.

We can now show the regret of a strategy for \eve in fact corresponds to the
supremum of the regret of plays consistent with the strategy.
%In other words,
%for any play consistent with a strategy of \eve, the difference between the
%payoff of its best alternative and its own payoff.
%
\begin{lemma}\label{lem:play-regret}
	For any strategy $\sigma$ of Eve,
	\(
		\regret{\sigma}{G}{} = \sup \{ \regret{}{\pi}{}
		\st \pi\text{ is consistent with }\sigma\}.
	\)
\end{lemma}

We note that for any play $\pi$, the sequence $\langle\lambda^i 
(\cVal^{v_i}_{\lnot v_{i+1}}(G) - \PlayVal{\pi[i..]})\rangle_{i \ge 0}$
converges to $0$ because $(\cVal^{v_i}_{\lnot v_{i+1}}(G) - \PlayVal{\pi[i..]})$
is bounded by $\frac{2W}{(1 - \lambda)}$.  It follows that if we have a non-zero
lower bound for the regret of $\pi$, then there is some index $N$ such that the
witness for the regret occurs before $N$.  Moreover, we can place a
polynomial upper bound on $N$.
%when $\lambda$ is fixed.
More precisely:
\begin{lemma}\label{lem:expBnd}
	Let $\pi$ be a play in $G$ and suppose $0 < r \leq \regret{}{\pi}{}$.
	Let
	\[
		N(r) \defeq \left \lfloor (\log r + \log (1-\lambda) -
		\log(2W))/\log \lambda \right \rfloor + 1.
	\]
	Then $\regret{}{\pi}{} = \regret{}{\pi[..{N(r)}]}{} -
	\lambda^{N(r)}  \PlayVal{\pi[{N(r)}..]}$.
\end{lemma}
The above result gives us a bound on how far we have to unfold a game after
having witnessed a non-zero lower bound, $r$, for the regret. If we consider the
example from Figure~\ref{fig:deviation}, this translates into a bound on how
many turns after $v_i$ a deviation can still yield bigger local regret (see
Figure~\ref{fig:deviation-bound}).

Corollary~\ref{cor:lower-bound} then gives us the required lower bound to be
able to use Lemma~\ref{lem:expBnd}.
\begin{lemma}\label{lem:regret-of-tree}
	If $\Regret{G}{} \ge a_G$ then $\Regret{G}{}$ is equal to
	\[
		\inf_{\sigma \in \StrAllE} \sup\{
			\regret{}{\pi[..N({a_G})]}{} - \lambda^{N({a_G})} 
			\aVal^{v_{N({a_G})}}(G) \st \pi = v_0 v_1 \dots \text{
			is consistent with } \sigma\}.
	\]
\end{lemma}

This already implies we can compute the regret value in alternating polynomial
time (or equivalently, deterministic polynomial space~\cite{cks81}).
%if $\lambda$ is fixed.
%
%\AEXPTIME\ (or, equivalently, \EXPSPACE~\cite{cks81}). 
%Indeed, we could simulate $G$ using an alternating Turing machine which halts in
%at most a pseudo-polynomial number of steps. Note that one needs to check
%whether the regret of the game is $0$. Thus we first need to label the arena
%with the antagonistic and co-operative values. However, this can be done with an
%alternating machine halting in time $\mathcal{O}(|V|)$ since discounted-sum and
%safety games are positionally determined (and therefore fall under the framework
%of \emph{first cycle games}~\cite{ar14}).
\begin{proposition}\label{pro:aexptime-solution}
	The regret value is computable using only polynomial space.
	%decidable in \PSPACE.
	%~for fixed $\lambda$,
	%and in \EXPSPACE~when $\lambda$ is not fixed.
	%Computing the regret value of a game, playing against any
	%adversary, can be done in time $\mathcal{O}(\max\{|V|,N(a_G)\})$ with an
	%alternating Turing machine.
\end{proposition}
\begin{proof}
	We first label the arena with the antagonistic and co-operative values
	and solve the safety game described for
	Theorem~\ref{thm:anyadversary-zero}. The latter can be done in
	polynomial time.
	%with an
	%alternating machine halting in time $\mathcal{O}(|V|)$ since
	%discounted-sum and safety games are positionally determined (and
	%therefore fall under the framework of \emph{first cycle
	%games}~\cite{ar14}). 
	If the \eve wins the safety game, the regret value  is $0$. Otherwise,
	we know $a_G > 0$ is a lower bound for the regret value.
	We now simulate $G$ using an alternating Turing machine which halts in
	at most $N(a_G)$ steps. That is, a polynomial number of steps. The
	simulated play prefix is then assigned a regret value as per
	Lemma~\ref{lem:regret-of-tree} (recall we have already
	pre-computed the antagonistic value of every vertex).
\end{proof}

As a side-product of the algorithm described in the above proof we get that
finite memory strategies suffice for \eve to minimize her regret in a
discounted-sum game.
\begin{corollary}\label{cor:finite-mem-regret}
	Let $\mu \defeq |\Delta|^{N(a_G)}$, with $N(0) = 0$. It holds that
	\[
		\Regret{G}{\StrE^{\mu},\StrAllA} = \Regret{G}{\StrAllE,\StrAllA}.
	\]
\end{corollary}

\subparagraph*{Simple regret-minimizing behaviours.}
We will now argue that \eve has a simple strategy which ensures regret of
at most $\Regret{G}{}$. Her strategy will consist in ``playing co-operatively''
(\ie, a strategy that attempts to maximize the co-operative payoff) for 
some turns (until a high local regret has already been witnessed) and
then switch to a co-operative worst-case optimal strategy (\ie, a strategy
attempting to maximize the co-operative payoff while achieving at least the
antagonistic payoff).
%We believe this somewhat surprising observation to
%be of independent interest, especially to other disciplines working in decision
%theory.

We will now define a family of strategies which switch from co-operative
behaviour to antagonistic, after a specific number of turns have elapsed (in
fact, enough for the discounted local regret to be less than the desired
regret).  Denote by $\sigma^{\mathsf{co}}$ a strongly co-operative strategy for
\eve in $G$ and
%Formally, for any play prefix $\pi = v_0 \dots v_n$ such that
%$v_n \in \VtcE$ if $\sigma^{\mathsf{co}}(\pi) = v'$ then $v' \in \copt{v_n}$,
%where $\copt{u} \defeq \{v \in V \st (u,v) \in E \text{ and } \cVal^{u}(G) =
%w(u,v) + \lambda \cVal^{v}(G)\}$.
by $\sigma^{\mathsf{cw}}$ % we denote a
a co-operative worst-case optimal strategy for \eve in $G$. Recall that, by
Lemma~\ref{lem:exist-strange-strats}, such strategies for her always exist.
%That is, for
%any play prefix $\pi = v_0 \dots v_n$ such that $v_n \in \VtcE$ if
%$\sigma^{\mathsf{cw}}(\pi) = v'$ then $v' \in \wcopt{v_n}$ and $w(v_n,v') +
%\lambda  \cVal^{v'}(G) = \max\{w(v_n,v'') + \lambda  \cVal^{v''}(G) \st v'' \in
%\wcopt{v_n}\}$, where $\wcopt{u} \defeq \{v \in V \st (u,v) \in E \text{ and }
%\aVal^u(G) = w(u,v) + \lambda  \aVal^v(G)\}$. It is well-known that, for
%discounted-sum games, co-operative strategies always exist. It is
%straightforward to show that the existence of co-operative worst-case optimal
%strategies is also guaranteed.  Indeed, it follows from the existence of
%worst-case optimal strategies---another classical result for discounted-sum
%games.
Finally, given a co-operative strategy $\sigma^{\mathsf{co}}$, a
co-operative worst-case optimal strategy $\sigma^{\mathsf{cw}}$, and $t \in
\mathbb{Q}$ let us define an \emph{optimistic-then-pessimistic strategy for
\eve} $\switch{\sigma^{\mathsf{co}}}{t}{\sigma^{\mathsf{cw}}}$. The strategy is
such that, for any play prefix $\rho = v_0 \dots v_n$ such that $v_n \in \VtcE$
\[
	\switch{\sigma^{\mathsf{co}}}{t}{\sigma^{\mathsf{cw}}}(\rho) =
	\begin{cases}
		\sigma^{\mathsf{co}}(\rho) & \text{if } |\copt{v_n}| = 1
			\text{ and }
			\locreg({\rho\cdot\sigma^{\mathsf{cw}}(\rho)},{n+1}) > t\\
		\sigma^{\mathsf{cw}}(\rho) & \text{otherwise.}
	\end{cases}
\]
%where $\dlocreg{\rho}{i} = \lambda^{i}(\cVal^{v_{i}}(\rho) - \aVal^{v_{i}}(\rho))$.

We claim that, when we set $t = \Regret{G}{}$, an optimistic-then-pessimistic
strategy for \eve ensures minimal regret. That is
\begin{proposition}\label{pro:simple-behaviour}
	%There is a strongly co-operative strategy $\sigma^{\mathsf{co}}$ for
	%\eve and a co-operative worst-case optimal strategy
	%$\sigma^{\mathsf{cw}}$ for \eve such that, for $r = \Regret{G}{}$, the
	Let $\sigma^{\mathsf{co}}$ be a strongly co-operative strategy
	for \eve, $\sigma^{\mathsf{cw}}$ be a \eve and a co-operative worst-case
	optimal strategy for \eve, and $t = \Regret{G}{}$. The strategy $\sigma
	= \switch{\sigma^{\mathsf{co}}}{t}{\sigma^{\mathsf{cw}}}$ has the
	property that $\regret{\sigma}{G}{} = \Regret{G}{}$.
\end{proposition}
This is a refinement of the strategy one can obtain from applying
the algorithm used to prove
Proposition~\ref{pro:aexptime-solution}.\footnote{In fact, our proof of
Prop.~\ref{pro:simple-behaviour} relies in \eve requiring finite memory,
to minimize her regret.} The latter tells us that a regret-minimizing
strategy of \eve eventually switches to a worst-case optimal behaviour. For
vertices where, before this switch, another edge was chosen by \eve, we argue
that she must have been playing a co-operative strategy. Otherwise, she could
have switched sooner. A full proof is provided in
Appendix~\ref{sec:simple-behaviour}.

%To conclude this section, let us recall some of the obtained results.
We have
shown the regret value can be computed using an algorithm which requires
polynomial space only. This algorithm is based on a polynomial-length unfolding
of the game and from it we can deduce that the regret value is representable
using a polynomial number of bits. (Indeed, all exponents ocurring in the
formula from Lemma~\ref{lem:regret-of-tree} will be polynomial according
to Lemma~\ref{lem:expBnd}.) Also, we have argued that \eve has a ``simple''
strategy $\sigma$ to ensure minimal regret. Such a strategy is defined by two
polynomial-time constructible sub-strategies and the regret value of the game.
Hence, it can be encoded into a polynomial number of bits itself. Furthermore,
$\sigma$ is guaranteed to be playing as its co-operative worst-case optimal
component after $N(\Regret{G}{})$ turns (see, again, Lemma~\ref{lem:expBnd}),
which is a polynomial number of turns.  
Given a regret threshold $r$, we claim we can
verify that $\sigma$ ensures regret at most $r$ in polynomial time. This can be
achieved by allowing \adam to play in $G$, and against $\sigma$, with the
objective of reaching an edge with high local regret before $N(\Regret{G}{})$
turns. An possible formalization of this idea follows. Consider the product of
$G$ with a counter ranging from $1$ to $N(\Regret{G}{})$ where we
%have some special marking for the edges chosen by $\sigma$ and we 
make all vertices belong to \adam. In this game $H$, we make edges leaving vertices
previously belonging to \eve go to a sink and define a new weight function $w'$
which assigns to these edges their negative non-discounted local regret: going
from $u$ to $v$ when $\sigma$ dictates to go to $v'$ yields $w(u,v') + \lambda
\aVal^{v'}(H \times \sigma) - w(u,v) + \lambda \cVal^v(H)$.
Lemma~\ref{lem:regret-of-tree} allows us to show that $\sigma$ ensures regret at
most $r$ in $G$ if and only if the antagonistic value of a discounted-sum game
played on $H$ with weight function $w'$ is at most $-r$.
%Intuitively, we
%construct $H$ and $w'$ so that \adam chooses the best turn to deviate, maximize
%against an alternative strategy, and minimize against $\sigma$.

It follows that the regret threshold problem is in \NP, as stated in
Theorem~\ref{thm:anyadversary}.

\begin{example}
	We revisit the discounted-sum game from Figure~\ref{fig:bigmem}. Let us
	instantiate the values $M = 100$ and $\lambda = \frac{9}{10}$. According
	to our previous remarks on this arena, after $i$ visits to $v$ without
	\adam choosing $(v,y)$, \eve could achieve $(\frac{9}{10})^{2i}10$ by
	going to $x$ or hope for $(\frac{9}{10})^{2i + 1}1000$ by going to $v$
	again. Her best regret minimizing strategy corresponds to $\sigma^{22}$
	which ensures regret of at most $9.9030 = 10 - (\frac{9}{10})^{44}10$.
	It is easy to see that \eve cannot win the safety game $\hat{G}$
	constructed from this arena and that the lower bound $a_G$ one can
	obtain from $\hat{G}$ is equal to $1.2466 = (\frac{9}{10})^{4}(10 -
	(\frac{9}{10})^2 10)$. As expected, when \eve plays
	her optimal regret-minimizing (optimistic-then-pessimistic)
	strategy any better alternative must deviate before $N(a_G) = 71$
	turns. Indeed, we have already argued that the regret $9.9030$ is
	witnessed by \adam choosing the edge $(v,y)$ for any strategy of \eve
	going to $v$ more than $22$ times.
\end{example}

\section{Regret against positional strategies of \adam}\label{sec:pos-adversary}
In this section we consider the problem of computing the (minimal) regret when
Adam is restricted to playing positional strategies.

\begin{theorem}\label{thm:memlessadversary}
	Deciding if the regret value is less than a given threshold (strictly or
	non-strictly), playing against positional strategies of \adam, is
	in \PSPACE.
\end{theorem}

Playing against an \adam, when he is restricted to playing
memoryless strategies gives \eve the opportunity to learn some of \adam's
strategic choices.  However, due to its decaying nature, with the discounted-sum
payoff function \eve must find a balance between exploring too quickly, thereby
presenting lightly discounted alternatives; and learning too slowly, thereby
heavily discounting her eventual payoff.

A similar approach to the one we have adopted in Section~\ref{sec:any-adversary}
can be used to obtain an algorithm for this setting. For reasons of space we
defer its presentation to the appendix. The claimed lower bound follows from
Theorem~\ref{thm:conp-hardness}. 
%Here, we only state the % complexity bound
%implies by the algorithm and 
%memory requirements for \eve which follow from the
%construction.

\subparagraph*{Deciding 0-regret.}
As in the previous section, we will reduce the problem of deciding if the game
has regret value $0$ to that of determining the winner of a safety game. It will
be obvious that if no regret-free strategy for \eve exists in the original game,
then we can construct, for any strategy of hers, a positional strategy of \adam
which ensures non-zero regret. Hence, we will also obtain a lower bound on the
regret of the game in the case \adam wins the safety game.

Let us fix some notation. For a set of edges $D \subseteq E$, we denote by $G
\restriction D$ the weighted arena $(V,\VtcE,v_I, D,w)$. Also, for a
positional strategy $\tau : (V\setminus\VtcE) \to E$ for \adam in $G$, we denote
by $G \times \tau$ the weighted arena resulting from removing all edges not
consistent with $\tau$. Next, for an edge $(s,t) \in E$ we define
$\learned(st) \defeq \{(u,v) \in E \st \text{if } u=s \text{
then } v = t \text{ or } u \in \VtcE\}$. We extend the latter to play prefixes
$\rho = v_0 \dots v_n$ by (recursively) defining $\learned(\rho) \defeq
\learned(\rho[..n-1]) \cap \learned(v_{n-1} v_n)$.  If $\pi$ is a play, then $E
\supseteq \learned(\pi[..i]) \supseteq \learned(\pi[..j])$ for all $0 \le i \le
j$. Hence, since $E$ is finite, the value $\learned(\pi) \defeq \lim_{i \ge 0}
\learned(\pi[..i])$ is well-defined. Remark that $\learned(\pi)$ does not
restrict edges leaving vertices of \eve.
%Furthermore, $\learned(\pi)$ might not
%include an outgoing edge for all vertices in $G$---\eg if $\pi$ has as factors
%$v \cdot v'$ and $v \cdot v''$, where $v \in \VtcA$ and $v' \neq v''$. However,
The following properties directly follow from our definitions.
\begin{lemma}\label{lem:learned-properties}
	Let $\pi$ be a play or play prefix consistent with a positional strategy for
	\adam. It then holds that:
	\begin{inparaenum}[$(i)$]
		\item for every $v \in \VtcA$ there is some edge $(v,\cdot) \in
			\learned(\pi)$,
		\item $\pi$ is consistent with a strategy $\tau \in \StrPosA(G)$
			if and only if $\tau \in \StrPosA(G \restriction
			\learned(\pi))$,
			and
		\item every strategy $\tau \in \StrPosA(G \restriction
			\learned(\pi))$ is also an element from
			$\StrPosA(G)$.
	\end{inparaenum}
\end{lemma}

To be able to decide whether regret-free strategies for \eve exist, we define a
new safety game. The arena we consider is $\hat{G} \defeq (\hat{V}, \hat{\VtcE},
\hat{v_I}, \hat{E})$ where $\hat{V} \defeq V \times \pow(E)$, $\hat{\VtcE}
\defeq \VtcE \times \pow(E)$, $\hat{v_I} \defeq (v_I,E)$, and $\hat{E}$ contains
the edge $\left( (u,C),(v,D) \right)$ if and only if $(u,v) \in E$ and $D = C
\cap \learned(uv)$.

\begin{theorem}\label{thm:posadversary-zero}
	Deciding if the regret value is $0$, playing against
	positional strategies of \adam, is in \PSPACE.
\end{theorem}
\begin{proof}
	A safety game is constructed as in the proof of
	Theorem~\ref{thm:anyadversary-zero}. Here, we consider $\tilde{G}$ and
	the set of bad edges $\tilde{\mathcal{B}} \defeq \{\left( (u,C),(v,D)
	\right) \in \hat{E} \st u \in \VtcE \text{ and } \exists \tau \in
	\StrPosA(G\restriction C), w(u,v) + \lambda \cVal^v(G \times
	\tau) < \cVal^u_{\lnot v}(G \times \tau)\}$. We then have the safety
	game $\tilde{G} = (\hat{V}, \hat{\VtcE}, \hat{v_I}, \hat{E},
	\tilde{\mathcal{B}})$. Note that there is an obvious bijective mapping
	from plays (and play prefixes) in $\tilde{G}$ to plays (prefixes) in $G$
	which are consistent with a positional strategy for \adam. One can then
	show the following properties hold:
	\begin{claim}\label{cla:strat-transfer-adam2}
		If $\tau \in \StrAllA(\tilde{G})$ is a winning strategy for \adam in
		$\tilde{G}$, then for all $\sigma \in \StrAllE(G)$, 
		there exist $t_{\tau\sigma} \in \StrPosA(G)$ and $s_{\tau\sigma}
		\in \StrAllE(G)$ such that
		\(
			\StratVal{}{}{s_{\tau\sigma}}{t_{\tau\sigma}} -
			\StratVal{}{}{\sigma}{t_{\tau\sigma}} \ge
			\lambda^{|V|(|E| + 1)}
		\)\\
		\(
			\min\{
				\cVal^u_{\lnot v}(G \times \tau) 
				-w(u,v) - \lambda
				\cVal^v(G \times \tau)\st
				\left( (u,C),(v,D) \right) \in
				\tilde{\mathcal{B}}, \tau
				\in \StrPosA(G\restriction C)\}.
		\)
	\end{claim}
	The claim follows from \emph{positional determinacy} of safety games and
	Lemma~\ref{lem:learned-properties} (see
	Appendix~\ref{sec:proof-strat-transfer-adam2}).

	\begin{claim}\label{cla:strat-transfer-eve2}
		If $\sigma \in \StrAllE(\tilde{G})$ is a winning strategy for \eve in
		$\tilde{G}$, then there is $s_\sigma \in \StrAllE(G)$ such that 
		$\regret{s_\sigma}{G}{\StrAllE,\StrPosA} = 0$.
	\end{claim}
	It then follows from the determinacy of safety games that \eve wins the
	safety game $\tilde{G}$ if and only if she has a regret-free strategy.
	We provide full proofs for these claims in appendix.

	We observe that simple cycles in $\tilde{G}$ have length at most
	$|V|(|E|+1)$. Thus, we can simulate the safety game until we complete a
	cycle and check that all traversed edges are good, all in alternating
	polynomial time. Indeed, an alternating Turing machine can simulate the
	cycle and then (universally) check that for all edges, for all positional
	strategies of the \adam, the inequality holds.
\end{proof}
\begin{corollary}\label{cor:lower-bound-positional}
	If no regret-free strategy for \eve exists in $G$, then
	$\Regret{G}{\StrAllE,\StrPosA}
	\ge b_G$ where $b_G \defeq \lambda^{|V|(|E|+1)} \min\{
	\cVal^u_{\lnot v}(G \times \tau) 
	-w(u,v) - \lambda
	\cVal^v(G \times \tau) \st
	\left( (u,C),(v,D) \right) \in \tilde{\mathcal{B}} \text{ and } \tau
	\in \StrPosA(G\restriction C)\}$.
\end{corollary}

%\subparagraph*{Deciding r-regret.}
%A similar approach to the one we have adopted in Section~\ref{sec:any-adversary}
%can be used to obtain an algorithm for this setting. For reasons of space we
%defer its presentation to the appendix. The claimed lower bound follows from
%Theorem~\ref{thm:conp-hardness}. Here, we only state the % complexity bound
%%implies by the algorithm and 
%memory requirements for \eve which follow from the
%construction.
%
%\begin{proposition}%\label{pro:horrible-result}
%	%The regret threshold problem, playing against a positional
%	%adversary, is decidable in \PSPACE\ for fixed $\lambda$, and in
%	%\EXPSPACE\ when $\lambda$ is not fixed.
%	Exponential memory
%	suffices for \eve to minimize her regret against a positional versary.
%\end{proposition}

\subparagraph{Lower bounds.}
%For lower bounds we consider two cases depending on whether $\lambda$ and $r$
%are given as part of the input or not. For the first case we claim
%\PSPACE-hardness; for the second case, \coNP-hardness.
We claim that both $0$-regret and $r$-regret are \coNP-hard.
This can be shown by
%These can be shown by
adapting the %reductions from QSAT (and 
reduction from $2$-disjoint-paths %, respectively)
given in~\cite{hpr15}
to
the regret threshold problem against memoryless adversaries.
%given
%in~\cite{hpr15}.
For completeness, we provide the reductions here in appendix.

%\begin{lemma}
%	\label{lem:pspace-hardness}
%	For a discount factor $\lambda \in (0,1)$, regret threshold $r \in
%	\mathbb{Q}$, and weighted arena $G$, determining whether
%	$\Regret{G}{\StrAllE,\StrPosA} \lhd r$, for $\lhd \in \{<,\le\}$, is
%	\PSPACE-hard.
%\end{lemma}

\begin{theorem}
	\label{thm:conp-hardness}
	Let $\lambda \in (0,1)$ and $r \in \mathbb{Q}$ be fixed. Deciding if the
	regret value is less than $r$ (strictly or non-strictly), playing
	against positional strategies of \adam, is
	\coNP-hard.
\end{theorem}

\section{Playing against word strategies of \adam}\label{sec:eloquent-adversary}
In this section, we consider the case where \adam is restricted to playing
\emph{word strategies}.  First, we show that the regret threshold problem can be
solved whenever the discounted sum automata associated to the game structure can
be made deterministic. As the determinization problem for discounted sum
automata has been solved in the literature for only sub-classes of discount
factors, and left open in the general case, we complement this result by two
other results. First, we show how to solve an {\em $\epsilon$-gap promise}
variant of the regret threshold problem, and second, we give an algorithm to
solve the $0$ regret problem. In the two cases, we obtain completeness results
on the computational complexities of the problems.

\subparagraph{Preliminaries.}
The formal definition of the $\epsilon$-gap promise problem is given below. We
first define here the necessary vocabulary.  We say that a strategy of \adam is
a \emph{word strategy} if his strategy can be expressed as a function $\tau :
\mathbb{N} \to [\max\{\outdeg{v} \st v \in V\}]$, where $[n] = \{i \st 1 \le i
\le n\}$. Intuitively, we consider an order on the successors of each \adam
vertex. On every turn, the strategy $\tau$ of \adam will tell him to move to the
$i$-th successor of the vertex according to the fixed order. We denote by
$\StrWordA$ the set of all such strategies for \adam. A game in which \adam
plays word strategies can be reformulated as a game played on a weighted
automaton $\Gamma=(Q, q_I, A, \Delta, w)$ and strategies of \adam---of the form
$\tau : \mathbb{N} \to A$---determine a sequence of input symbols, i.e. an omega
word, to which \eve has to react by choosing $\Delta$-successor states starting
from $q_I$. In this setting a strategy of \eve which minimizes regret defines a
run by resolving the non-determinism of $\Delta$ in $\Gamma$, and ensures the
difference of value given by the constructed run is minimal w.r.t.  to the value
of the best run on the word spelled out by \adam.

\subparagraph*{Deciding 0-regret.}
We will now show that if the regret of an arena (or automaton) is $0$, then we
can construct a memoryless strategy for \eve which ensures no regret is
incurred. More specifically, assuming the regret is $0$, we have the existence
of a family of strategies of \eve which ensure decreasing regret (with limit
$0$). We use this fact to choose a small enough $\epsilon$ and the corresponding
strategy of hers from the aforementioned family to construct a memoryless
strategy for \eve with nice properties which allow us to conclude that its
regret is $0$. Hence, it follows that an automaton has zero regret if and only 
if
a memoryless strategy of \eve ensures regret $0$. As we can guess such a
strategy and easily check if it is indeed regret-free (using the obvious
reduction to non-emptiness of discounted-sum automata or one-player
discounted-sum games), the problem is in \NP.
A matching lower bound follows from a reduction from SAT which was first
described in~\cite{akl10}. We sketch it, for completeness, in the appendix.

\begin{theorem}\label{thm:eloquentadversary-zero}
	Deciding if the regret value is $0$, playing against word
	strategies of \adam, is \NP-complete.
\end{theorem}

\subparagraph{Deciding r-regret: determinizable cases.}
When the weighted automaton  $\Gamma$ associated to the game structure can be
made deterministic, we can solve the regret threshold problem with the following
algorithm. In~\cite{hpr15} we established that, against eloquent adversaries,
computing the regret reduced to computing the value of a quantitative simulation
game as defined in~\cite{cdh10}. The game is obtained by taking the product of
the original automaton and a deterministic version of it. The new weight
function is the difference of the weights of both components (for each pair of
transitions). In~\cite{bh14}, it is shown how to determinize discounted-sum
automata when the discount factor is of the form $\frac{1}{n}$, for $n \in
\nat$. So, for this class of discount factor, we can state the following
theorem:

\begin{theorem}\label{thm:eloquentadversary}
	Deciding if the regret value is less than a given threshold (strictly or
	non-strictly), playing against word strategies of \adam, is in \EXP~for
	$\lambda$ of the form $\frac{1}{n}$.
\end{theorem}

%Let us start with the algorithm to solve the problem
%when the discount factor is the reciprocal of an integer. 
%
%
%
%
%The construction was inspired by the quantitative
%simulation game given in~\cite{cdh10} and, together with the determinization
%procedure from~\cite{bh14} implies that:
%%
%\begin{lemma}\label{lem:exp-memb-eloquent}
%	The regret value, playing against a positional
%	adversary, can be computed in exponential time if $\lambda =
%	\frac{1}{n}$, where $n \in \mathbb{N} \setminus \{0,1\}$.
%\end{lemma}
%%
%If $\lambda$ is fixed and if the response is only required to be
%$\epsilon$-accurate, in the sense formalized below, then we are able to solve
%the problem in polynomial space.
%
%
%!!! ADD SUCCINCT DESCRIPTION OF THE ALGO !!!
%In~\cite{},  

%\subparagraph{Deciding r-regret (approximately).}
%
\subparagraph*{The $\epsilon$-gap promise problem.}
Given a discounted-sum automaton $\calA$, $r \in \mathbb{Q}$, and $\epsilon >
0$, the $\epsilon$-gap promise problem adds to the regret threshold problem the
hypothesis that $\calA$ will either have regret $\leq r$ or $> r + \epsilon$. We
observe that an algorithm which gives:
%
%An algorithm that solves the \textsc{$\epsilon$-gap regret problem} is one which
%satisfies the following. Given a discounted-sum automaton $\mathcal{A}$,
%positive regret threshold $r \in \mathbb{Q}$ and $\epsilon > 0$,
\begin{itemize}
	\item a YES answer implies that
\(
	\Regret{\mathcal{A}}{\StrE,\StrWordA} \le r+\epsilon,
\)
	\item whereas a NO answer implies
\(
	\Regret{\mathcal{A}}{\StrE,\StrWordA} > r.
\)
\end{itemize}
will decide the $\epsilon$-gap promise problem.

In~\cite{bh14}, it is shown that there are discounted-sum automata which 
define functions
that cannot be realized with deterministic-sum automata. Nevertheless, it is also
shown in that paper that given a discounted-sum automaton it is always possible
to construct a deterministic one that is $\epsilon$-close in the following
formal sense. A \emph{discounted-sum automaton $\mathcal{A}$ is $\epsilon$-close
to another discounted sum automaton $\mathcal{B}$}, if for all words $x$ the absolute
value of the difference between the values assign by $\mathcal{A}$ and
$\mathcal{B}$ to $x$ is at most $\epsilon$. So, it should be clear that we can
apply the algorithm underlying Theorem~\ref{thm:eloquentadversary} to $\Gamma$
and a determinized version $\mathcal{D}_\Gamma$ of it (which is $\epsilon$-close
to $\Gamma$) and solve the $\epsilon$-gap promise problem. We can then prove the
following result.

%It follows from the proof of Lemma~\ref{lem:exp-memb-eloquent} and
%the \emph{approximate determinization procedure} from~\cite{bh14} that the
%\textsc{$\epsilon$-gap regret problem} is in \EXP.
%\begin{corollary}\label{cor:epsilon-gap-exp}
%	For $\lambda \in (0,1)$, $r \in \mathbb{Q}$, and weighted arena $\calA$, 
%	the \textsc{$\epsilon$-gap regret problem} is solvable in exponential
%	time.
%\end{corollary}
%%
%If, however, $\lambda$ is fixed, then a finer analysis of a more na\"ive
%approach towards determinization of $\discfun{\lambda}$ automata gives a
%\PSPACE~algorithm.
%
\begin{theorem}\label{thm:epsilon-gap-pspace}
	Deciding the $\epsilon$-gap regret problem
	is in \PSPACE.
\end{theorem}

The complexity of the algorithm follows from the fact that the value of a
(quantitative simulation) game, played on the product of $\Gamma$ and
$\mathcal{D}_\Gamma$ we described above, can be determined by simulating the
game for a polynomial number of turns. Thus, although the automaton constructed
using the techniques of Boker and Henzinger~\cite{bh14} is of size exponential,
we can construct it ``on-the-fly'' for the required number of steps and then
stop.
%\begin{proof}[Sketch]
%We reduce the problem to determining the winner of a reachability game on an
%exponentially larger arena.  Although the arena is exponentially larger, all
%paths are only polynomial in length, so the winner can be determined in
%alternating polynomial time, or equivalently, polynomial space.  
%
%The idea of the construction is as follows.  Given a discounted-sum automaton
%$\mathcal{A}$, we determinize its transitions via a subset construction, to
%obtain a deterministic, multi-valued discounted-sum automaton $D_{\mathcal{A}}$.
%Then we decide if Eve is able to simulate, within the regret bound, the
%$D_{\mathcal{A}}$ on $\mathcal{A}$ for all \emph{finite} words up to a length
%(polynomially) dependent on $\epsilon$.  If we simulate the automaton for a
%sufficient number of steps, then any significant gap between the automata will
%be unrecoverable regardless of future inputs, and we can give a satisfactory
%answer for the \textsc{$\epsilon$-gap regret problem}. More specifically, we
%only have to simulate this determinization process for $N$ steps, where
%\(
%	N \defeq \left\lfloor
%	\log_{\lambda}\left(\epsilon(1-\lambda) / 4W\right)\right\rfloor +
%	1.
%\)
%\end{proof}

\subparagraph{Lower bounds.}
We claim the $\epsilon$-gap promise problem is \PSPACE-hard even if both
$\lambda$ and $\epsilon$ are not part of the input. To establish the result, we
give a reduction from QSAT which uses the gadgets depicted in
Figures~\ref{fig:initial-gadget} and~\ref{fig:clause-gadgets}. For space
reasons we defer the reduction to Appendix~\ref{app:eloquent-adversary}.
%We now give a reduction from QBF to the \emph{promise} \textsc{$\epsilon$-gap
%regret problem} in order to establish \PSPACE-hardness.

\begin{theorem}
	\label{thm:eloquent-pspace-hardness-epsilon}
	Let $\lambda \in (0,1)$ and $\epsilon \in (0,1)$ be fixed. As input,
	assume we are given $r \in \mathbb{Q}$ and weighted arena $\calA$
	such that $\Regret{\mathcal{A}}{\StrE,\StrWordA} \le r$ or
	$\Regret{\mathcal{A}}{\StrE,\StrWordA} > r + \epsilon$. Deciding if the
	regret value is less than a given threshold, playing against word
	strategies of \adam, is \PSPACE-hard.
\end{theorem}
It follows that the general problem is also \PSPACE-hard (even if $\epsilon$ is
set to $0$).
\begin{corollary}\label{cor:eloquent-pspace-hard}
	Let $\lambda \in (0,1)$. For $r \in \mathbb{Q}$, weighted arena $G$, 
	determining whether $\Regret{G}{\StrAllE,\StrWordA} \lhd r$, for $\lhd
	\in \{<,\le\}$, is \PSPACE-hard.
\end{corollary}

\bibliographystyle{plain}
\bibliography{refs}

\newpage
\appendix

\section{Missing Proofs From Section~\ref{sec:any-adversary}}
\subsection{Proof of Lemma~\ref{lem:play-regret}}
	Consider any $\sigma, \sigma' \in \StrAllE$ and $\tau \in \StrAllA$ such
	that $\out{}{\sigma}{\tau} \neq \out{}{\sigma'}{\tau}$.  Let us write
	$\out{}{\sigma}{\tau} = v_0 v_1 \dots$ and $\out{}{\sigma'}{\tau} = v'_0
	v'_1 \dots$ and denote by $\ell$ the length of the longest common
	prefix of $\out{}{\sigma}{\tau}$ and $\out{}{\sigma'}{\tau}$. We claim
	that
	\begin{equation}\label{equ:inverse-combine-behavior1}
		\lambda^{\ell}  \bigl(
		\cVal^{v_\ell}_{\lnot
		v_{\ell + 1}}(G) -
		\PlayVal{\out{}{\sigma}{\tau}} \bigr)
		\ge
		\lambda^{\ell}  \bigl(
		\PlayVal{\out{}{\sigma'}{\tau}[\ell..]}-
		\PlayVal{\out{}{\sigma}{\tau}[\ell..]}
		\bigr).
	\end{equation}
	Indeed, if we assume it is not the case, we then get that
	$\cVal^{v'_{\ell+1}}(G) < \PlayVal{\out{}{\sigma'}{\tau}[\ell + 1..]}$,
	which contradicts the definition of $\cVal$. Note that
	Lemma~\ref{lem:combine-behaviors1} actually tells us that there is
	another strategy $\tau'$ for \adam and a second alternative strategy
	$\sigma''$ for \eve which give us equality in the above equation. More
	formally, from Equation~\ref{equ:inverse-combine-behavior1} and
	Lemma~\ref{lem:combine-behaviors1} we get that for all $\sigma \in
	\StrAllE$, if there are $\tau \in \StrAllA$ and $\sigma' \in \StrAllE$
	such that $\out{}{\sigma}{\tau} \neq \out{}{\sigma'}{\tau}$ then
	\begin{equation}\label{equ:full-combine-behavior1}
		\sup_{\tau,\sigma' \text{ s.t. }
		\out{}{\sigma}{\tau} \neq \out{}{\sigma'}{\tau}}
		\lambda^{\ell}  \bigl(
		\PlayVal{\out{}{\sigma'}{\tau}[\ell..]}-
		\PlayVal{\out{}{\sigma}{\tau}[\ell..]}
		\bigr) =
		\lambda^{\ell}  \bigl(
		\cVal^{v_\ell}_{\lnot
		v_{\ell + 1}}(G) -
		\PlayVal{\out{}{\sigma}{\tau}} \bigr).
	\end{equation}
	We are now able to prove
	the result. That is, for any strategy $\sigma$ for \eve:
	\begin{align*}
		& \sup\{ \regret{}{\pi}{} \st \pi \text{ is
		consistent with } \sigma\} & \\
		= & \sup_{\tau \in \StrAllA} \regret{}{\out{}{\sigma}{\tau} =
		v_0 v_1 \dots }{}
		& \text{def. of } \out{}{\sigma}{\tau}\\
		= & \sup_{\tau \in \StrAllA} \max\left\{0,
		\sup_{\substack{i \ge 0\\
		v_i \in \VtcE}} \lambda^i  \left(
		\cVal^{v_i}_{\lnot v_{i+1}}(G) -
		\PlayVal{\out{}{\sigma}{\tau}[i..]} \right) \right\}
		& \text{def. of } \regret{}{\out{}{\sigma}{\tau}}{}\\
		= & \sup_{\tau \in \StrAllA} \max\left\{0,
		\sup_{\sigma' \text{s.t.}\out{}{\sigma}{\tau} \neq
		\out{}{\sigma'}{\tau}}
		\lambda^\ell  \left(
		\PlayVal{\out{}{\sigma'}{\tau}[\ell..]} -
		\PlayVal{\out{}{\sigma}{\tau}[\ell..]} \right) \right\}
		& \text{by Eq.~\eqref{equ:full-combine-behavior1}}\\
		= & \sup_{\tau \in \StrAllA} \max\left\{0,
		\sup_{\sigma' \text{s.t.}\out{}{\sigma}{\tau} \neq
		\out{}{\sigma'}{\tau}}
		\left(
		\StratVal{}{}{\sigma'}{\tau} -
		\StratVal{}{}{\sigma}{\tau} \right) \right\}
		& \text{def. of } \PlayVal{\cdot},\ell\\
		= & \sup_{\tau \in \StrAllA}
		\sup_{\sigma' \in \StrAllE}
		\left(
		\StratVal{}{}{\sigma'}{\tau} -
		\StratVal{}{}{\sigma}{\tau} \right)
		& 0 \text{ when } \out{}{\sigma}{\tau} = \out{}{\sigma'}{\tau}
	\end{align*}
	as required.
	\qed

\subsection{Proof of Lemma~\ref{lem:expBnd}}
	Observe that $N(r)$ is such that $\frac{2W\lambda^{N(r)}}{1-\lambda}<r$.
	Hence, we have that for all $i \ge N(r)$ such that $v_i \in \VtcE$ it
	holds that $\lambda^i (\cVal^{v_i}_{\lnot v_{i+1}}(G) -
	\PlayVal{\pi[i..]}) \leq \frac{2W\lambda^{N(r)}}{1-\lambda} < r$. It
	follows that
	\begin{align*}
		\regret{}{\pi}{} &=
		\sup\{
		\lambda^i  (\cVal^{v_i}_{\lnot v_{i+1}}(G) -
		\PlayVal{\pi[i..]})\st i \ge 0 \text{ and } v_i \in \VtcE\}\\
		&= \max_{\substack{0 \leq i < {N(r)}\\
		v_i \in \VtcE}} 
		\lambda^i  \left(\cVal^{v_i}_{\lnot v_{i+1}}(G) -
		\PlayVal{\pi[i..{N(r)}]}\right) - \lambda^{N(r)} 
		\PlayVal{\pi[{N(r)}..]}
	\end{align*}
	as required.\qed

\subsection{Proof of Lemma~\ref{lem:regret-of-tree}}
	First, note that if $\Regret{G}{} > 0$ then there cannot be any
	regret-free strategies for \eve in $G$. It then follows from
	Corollary~\ref{cor:lower-bound} that $\Regret{G}{} \ge a_G$. Next, using
	Lemma~\ref{lem:expBnd} and the definition of the regret of a play we
	have that $\Regret{G}{}$ is equal to
	\[
		\inf_{\sigma \in \StrAllE} \sup\{ \regret{}{\pi[..N(a_G)]}{} -
		\lambda^{N(a_G)}  \PlayVal{\pi[N(a_G)..]} \st \pi \text{ is
		consistent with } \sigma\}.
	\]
	Finally, note that it is in the interest of \eve to maximize the value
	$\lambda^{N(a_G)}  \PlayVal{\pi[N(a_G)..]}$ in order to minimize
	regret. Conversely, \adam tries to minimize the same value.  Thus, we
	can replace it by the antagonistic value from $\pi[N(a_G)..]$ discounted
	accordingly. More formally, we have
	\begin{align*}
		&\inf_{\sigma \in \StrAllE} \sup\{ \regret{}{\pi[..N(a_G)]}{} -
		\lambda^{N(a_G)}  \PlayVal{\pi[N(a_G)..]} \st \pi \text{ is
		consistent with } \sigma\} \\
		=&\inf_{\sigma \in \StrAllE}
		\sup_{\tau \in \StrAllA}
		\regret{}{\out{}{\sigma}{\tau}[..N(a_G)]}{} -
		\lambda^{N(a_G)} 
		\PlayVal{\out{}{\sigma}{\tau}[N(a_G)..]}\\
		=&\inf_{\substack{\sigma \in \StrAllE\\\sigma' \in \StrAllE}}
		\sup_{\substack{\tau \in \StrAllA\\\tau' \in \StrAllA}}
		\regret{}{\out{}{\sigma}{\tau}[..N(a_G)] = \dots v}{} -
		\lambda^{N(a_G)}  \StratVal{}{v}{\sigma'}{\tau'}\\
		=&\inf_{\sigma \in \StrAllE}
		\sup_{\tau \in \StrAllA}
		\regret{}{\out{}{\sigma}{\tau}[..N(a_G)] = \dots v}{} +
		\inf_{\sigma' \in \StrAllE} \sup_{\tau' \in \StrAllA}
		\left(
		-\lambda^{N(a_G)} 
		\StratVal{}{v}{\sigma'}{\tau'}
		\right)\\
		=&\inf_{\sigma \in \StrAllE}
		\sup_{\tau \in \StrAllA}
		\regret{}{\out{}{\sigma}{\tau}[..N(a_G)] = \dots v}{} -
		\lambda^{N(a_G)} 
		\left(
		\sup_{\sigma' \in \StrAllE} \inf_{\tau' \in \StrAllA}
		\StratVal{}{v}{\sigma'}{\tau'}
		\right)\\
		=&\inf_{\sigma \in \StrAllE}
		\sup_{\tau \in \StrAllA}
		\regret{}{\out{}{\sigma}{\tau}[..N(a_G)] = \dots v}{} -
		\lambda^{N(a_G)} 
		\aVal^v(G)
	\end{align*}
	as required.\qed

\subsection{Proof of Claim~\ref{cla:strat-transfer-eve}}
	As a first step towards proving the result, we first make the
	observation that any winning strategy of \eve in $\hat{G}$ also ensures
	a value of at least $\aVal(G)$ in the discounted-sum game played on $G$.
	More formally,
	\begin{claim}\label{cla:win-is-wco}
		If $\sigma \in \StrAllE$ is a winning strategy for \eve in
		$\hat{G}$, then
		\begin{equation}\label{eqn:ensure-aval}
			\forall \tau \in \StrAllA, \forall i \ge 0 :
			\PlayVal{\out{ }{\sigma}{\tau}[i..] = v_i\dots}
			\ge \aVal^{v_i}(G).
		\end{equation}
	\end{claim}
	\begin{proof}
		Consider a winning strategy $\sigma \in \StrAllE$ for \eve in
		$\hat{G}$. Since safety games are positionally determined (see,
		\eg~\cite{ag11}) we can assume w.l.o.g. that $\sigma$ is
		memoryless.

		To convince the reader that $\sigma$ has the property from
		Equation~\eqref{eqn:ensure-aval}, we consider the synchronized
		product of $G$ and $\sigma$---that is, the synchronized product
		of $G$ and the finite Moore machine realizing $\sigma$.  As
		$\sigma$ is memoryless, then this product, which
		we denote in the sequel by $G \times \sigma$, is finite.
		Now, towards a contradiction, suppose that
		Equation~\eqref{eqn:ensure-aval} does not hold for $\sigma$.
		Further, let us consider an alternative (memoryless) strategy
		$\sigma'$ of \eve which ensures $\aVal^v(G)$ from all $v \in V$.
		The latter exists by definition of $\aVal(G)$ and memoryless
		determinacy of discounted-sum games (see, \eg~\cite{zp96}).
		
		Let $H$ denote a copy of $G \times \sigma$ where all edges
		induced by $E$ from $G$ are added---not just the ones allowed
		by $\sigma$---and $H \arestriction \sigma'$ denote the
		sub-graph of $H$ where only edges allowed by $\sigma'$ are left.
		Since, by assumption, $\sigma$ does not have the property of
		Equation~\eqref{eqn:ensure-aval} then the edges present in at
		least one vertex from $H \arestriction \sigma'$ and $G \times
		\sigma$ differ. Note that such a vertex $u$ is necessarily such
		that $u \in \VtcE$. Furthermore, from our definition of a
		strategy, we know that there is a single outgoing edge from it
		in both structures. Let us write $(u,v)$ for the
		edge in $G \times \sigma$ and $(u,v')$ for the edge in $H
		\arestriction \sigma'$. Recall that $\sigma$ is winning for \eve
		in $\hat{G}$. Thus, we have that $(u,v) \not\in 
		\mathcal{B} = \{ (u,v) \in E \st u \in \VtcE$ and $w(u,v) +
		\lambda  \aVal^v(G) < \cVal^{u}_{\lnot v}(G)\}$. It follows
		that
		\begin{align*}
			w(u,v) + \lambda  \aVal^v(H) & \ge
				\max_{x \neq v}\{ w(u,x) + \lambda 
					\cVal^{x}(H) \} & \\
			&\ge \max_{x \neq v}\{ w(u,x) + \lambda 
					\aVal^{x}(H) \} &
					\text{as }
					\cVal^{x}(H) \ge \aVal^{x}(H)\\
			&=\aVal^u(H) & \text{because } u \in
					V_\exists.
		\end{align*}
		Thus, the strategy $\sigma''$ of \eve which takes $(u,v)$
		instead of $(u,v')$ and follows $\sigma'$ otherwise---indeed,
		this might mean $\sigma''$ is not memoryless---also achieves at
		least $\aVal^u(H)$ from $u$ onwards and is therefore
		an worst-case optimal antagonistic strategy in $G$ (\ie~it has
		the property of Equation~\eqref{eqn:ensure-aval}).  Notice that
		this process can be repeated for all vertices in which the two
		structures differ. Further, since both are finite, it will
		eventually terminate and yield a strategy of \eve which plays
		exactly as $\sigma$ and for which Equation~\eqref{eqn:ensure-aval}
		holds, which is absurd.
	\end{proof}

	Once more, consider a winning strategy $\sigma \in \StrAllE$ for \eve in
	$\hat{G}$. We will now show that
	\[
		\forall \tau \in \StrAllA, \forall \sigma' \in \StrAllE
		\setminus \{\sigma\}: \StratVal{}{}{\sigma}{\tau}
		\ge \StratVal{}{}{\sigma'}{\tau}.
	\]
	The desired result will then directly follow.
	
	Consider arbitrary strategies $\tau \in \StrAllA$ and $\sigma' \in
	\StrAllE \setminus \{\sigma\}$. Suppose that $\out{ }{\sigma}{\tau} \neq
	\out{ }{\sigma'}{\tau}$, as our claim trivially holds otherwise. Let
	$\iota$ be the maximal index $i \ge 0$ such that, if we write $\out{
	}{\sigma}{\tau} = v_0 v_1 \dots$ and $\out{ }{\sigma'}{\tau} = v'_0 v'_1
	\dots$, then $v_i = v'_i$. That is, $\iota$ is the maximal index for
	which the outcomes of $\sigma$ and $\tau$, and $\sigma'$ and $\tau$
	coincide. Note that $v_\iota$ is necessarily an \eve vertex,
	\ie~$v_\iota \in \VtcE$. We observe that, by definition of $\cVal$,
	it holds that
	\begin{equation}\label{equ:cval}
		\PlayVal{\out{ }{\sigma'}{\tau}[\iota + 1..]} \le
		\cVal^{v'_{\iota + 1}}(G).
	\end{equation}
	Furthermore, we know from the fact that $\sigma$ is winning for \eve in
	$\hat{G}$ that the edge $(v_\iota,v_{\iota + 1})$ is such that
	\begin{equation}\label{equ:winning}
		w(v_\iota,v_{\iota + 1}) + \lambda  \aVal^{v_{\iota +
		1}}(G) \ge \max_{t \neq v_{\iota + 1}}\{ w(v_\iota,t) +
		\lambda  \cVal^{t}(G) \}.
	\end{equation}
	In particular, this implies that $w(v_\iota,v_{\iota + 1}) + \lambda
	 \aVal^{v_{\iota + 1}}(G) \ge w(v_\iota,v'_{\iota+1}) + \lambda
	 \cVal^{v'_{\iota + 1}}(G)$. It is then easy to verify that
	$w(v_\iota,v_{\iota + 1}) + \lambda  \aVal^{v_{\iota + 1}}(G) =
	\aVal^{v_\iota}(G)$ using the observation that $v_\iota \in V_\exists$.
	From Claim~\ref{cla:win-is-wco} we also get that
	\begin{equation}\label{equ:from-claim}
		\PlayVal{\out{ }{\sigma}{\tau}[\iota..]} \ge \aVal^{v_\iota}(G).
	\end{equation}
	Putting all the above inequalities together, we have
	\begin{align*}
		\PlayVal{\out{ }{\sigma}{\tau}[\iota..]} & \ge
		\aVal^{v_\iota}(G) = w(v_\iota,v_{\iota + 1}) + \lambda 
		\aVal^{v_{\iota + 1}}(G)
		& \text{by Eqn.~\eqref{equ:from-claim}} \\
		& \ge w(v_\iota,v'_{\iota+1}) + \lambda  \cVal^{v'_{\iota +
		1}}(G) & \text{by Eqn.~\eqref{equ:winning}} \\
		& \ge \PlayVal{\out{}{\sigma'}{\tau}[\iota..]} & \text{by
			Eqn.~\eqref{equ:cval}}
	\end{align*}
	which, in turn, implies $\StratVal{}{}{\sigma}{\tau} \ge
	\StratVal{}{}{\sigma'}{\tau}$ since $\out{}{\sigma}{\tau}[..\iota] =
	\out{}{\sigma'}{\tau}[..\iota]$.
	\qed

\subsection{Proof of Proposition~\ref{pro:simple-behaviour}}
\label{sec:simple-behaviour}
Let us start by showing that the regret of a play $\pi$ is
bounded (from above) by the discounted local regret from any index $i$, where
from the $i$-th turn onwards \eve plays a worst-case optimal strategy. More
formally:
\begin{lemma}\label{lem:bubble-up-local-regret}
	Let $\pi = v_0 v_1 \dots$ be a play. Assume there is some $i \in
	\mathbb{N}$ such that
	\begin{enumerate}[$(i)$]
	\item $v_i \in \VtcE$;
	\item $\regret{}{\pi}{} \le \lambda^i  \regret{}{\pi[i..]}{}$; and
	\item $\aVal^{v_j}(G) = w(v_j,v_{j+1}) +
		\lambda  \aVal^{v_{j+1}}(G)$, for all $j \ge i$.
	\end{enumerate}
	It then holds that
	$\regret{}{\pi}{} \le \lambda^i  \left( \cVal^{v_i}(G) -
	\aVal^{v_i}(G) \right)$.
\end{lemma}
\begin{proof}
	If $\regret{}{\pi}{} = 0$ then the claim holds trivially. Hence, let us
	assume $\regret{}{\pi}{} > 0$. It follows from
	Lemma~\ref{lem:regret-of-tree} and Assumption $(ii)$ that there
	exists $k \ge i$ such that $v_k \in \VtcE$ and 
	\[
		\regret{}{\pi}{} = \lambda^k  \left(\cVal^{v_k}_{\lnot
			v_{k+1}}(G) - w(v_k,v_{k+1}) - \lambda 
			\aVal^{v_{k+1}}(G)\right).
	\]
	Observe that $\cVal^{v_k}(G) \ge \cVal^{v_k}_{\lnot v_{k+1}}(G)$, by
	definition, and that from Assumption $(iii)$ we have that 
	$\aVal^{v_{k}}(G) \le w(v_k,v_{k+1}) + \lambda
	 \aVal^{v_{k+1}}(G)$. Thus, we get that
	$\regret{}{\pi}{} \le \lambda^k  \left( \cVal^{v_k}(G) -
	\aVal^{v_k}(G) \right)$. Also, note that by definition of $\cVal$ we
	have that
	\[
		\cVal^{v_j}(G) \ge w(v_j,v_{j+1}) + \lambda  \cVal^{v_{j+1}}(G)
	\]
	for all $j \ge 0$. It thus follows from Assumption $(iii)$ and the
	previous arguments that $\regret{}{\pi}{} \le \lambda^i  \left(
	\cVal^{v_i}(G) - \aVal^{v_i}(G) \right)$ as required.
\end{proof}

We are now ready to prove the Proposition holds.

	\subparagraph*{The zero case.} If $\Regret{G}{} = 0$, then it follows from
	our reduction to safety games that \eve has a co-operative worst-case
	optimal strategy which minimizes regret. Indeed, it is straightforward
	to show that the strategy for \eve obtained from the safety game does
	not only ensure at least the antagonistic value, but it is also
	co-operative worst-case optimal. Thus, since
	$\switch{\sigma^{\mathsf{co}}}{0}{\sigma^{\mathsf{cw}}}$ is
	clearly equivalent to $\sigma^{\mathsf{cw}}$ in this case, the result
	follows.

	\subparagraph*{Non-zero regret.} Let us assume that $\Regret{G}{} > 0$. It
	then follows from Lemma~\ref{lem:regret-of-tree} that \eve has a
	finite memory strategy $\sigma$ which ensures regret of at most
	$\Regret{G}{}$ (see Corollary~\ref{cor:finite-mem-regret}) and which,
	furthermore, can be assumed to switch after turn $N(a_G)$ to a
	co-operative worst-case optimal strategy $\sigma^{\mathsf{cw}}$ for \eve
	(since such a strategy ensures at least the antagonistic value of the
	vertex from which \eve starts playing it). We will further assume,
	w.l.o.g., that
	for all play prefixes $\pi = v_0 \dots v_n$ with $n \le N(a_G)$, $v_n
	\in \VtcE$ and having $\sigma^{\mathsf{cw}}(\pi) \neq
	\sigma^{\mathsf{co}}(\pi) = \sigma(\pi)$, if $\sigma$ switches to
	$\sigma^{\mathsf{cw}}$ from $\pi$ onwards---that is, for all prefixes
	extending $\pi$---then the regret of the resulting strategy is strictly
	greater than $\Regret{G}{}$. Otherwise, one can consider the strategy
	resulting from the previously described switch instead of $\sigma$.

	We will now argue that for all play prefixes $\pi = v_0 \dots v_n$ with
	$n \le N(a_G)$ and $v_n \in \VtcE$, if $\sigma(\pi) \neq
	\sigma^{\mathsf{cw}}$ then $\copt{v_n}$ is a singleton and
	$\locreg{\pi[..n]\cdot \sigma^{\mathsf{cw}}(\pi[..n])}{n+1} >
	\Regret{G}{}$. The desired result will follow since in order for our
	assumption of $\regret{}{\sigma}{} = \Regret{G}{}$ to be true \eve must
	then choose the unique edge leading to the single element in
	$\copt{v_n}$.
	
	Let us consider two cases. 
	
	First, if $\locreg{\pi[..n]\cdot
	\sigma^{\mathsf{cw}}(\pi[..n])}{n+1} \le \Regret{G}{}$, we can switch to
	$\sigma^{\mathsf{cw}}$ fron $\pi[..n]$ onwards. Contradicting our
	initial assumption.
	
	Second, if $|\copt{v_n}| > 1$ and $\locreg{\pi[..n]\cdot
	\sigma^{\mathsf{cw}}(\pi[..n])}{n+1} > \Regret{G}{}$, then by
	Lemma~\ref{lem:bubble-up-local-regret} we get that the regret of
	the play (if we switched to $\sigma^{\mathsf{cw}}$) is bounded
	above by $\lambda^n  \left(\cVal^{v_n}(G) -
	\aVal^{v_n}(G)\right)$. Also, since $\copt{v_n}$ is not a
	singleton, if \eve does not switch, then she cannot ensure a
	local regret of less than $\lambda^n  \left(\cVal^{v_n}(G) -
	\aVal^{v_n}(G)\right)$---particularly, not even by taking an
	edge leading to a vertex in $\copt{v_n}$. This contradicts the
	assumption that that switching to $\sigma^{\mathsf{cw}}$ yields
	strictly more regret.
	\qed

\subsection{Lower bound}
We now establish a lower bound for computing the
minimal regret against any strategy by reducing from the problem of determining
the antagonistic value of a discounted-sum game.  More precisely, from a
weighted arena $G$ we construct, in logarithmic space, a weighted arena $G'$
such that the antagonistic value of $G$ is equal to the regret value of $G'$.
This gives us:

\begin{lemma}\label{lem:AlltoRegret}
	Computing the regret of a discounted-sum game is at least as hard as
	computing the antagonistic value of a (polynomial-size) game with the
	same payoff function.
\end{lemma}

\begin{figure}
\begin{center}
\resizebox{0.4\textwidth}{!}{%
\begin{tikzpicture}[inner sep=2mm, ve/.style={rectangle,
	draw},va/.style={circle, draw}, node distance=1cm]
\node[ve,initial above](A){{$v_I'$}};
\node[ve,dotted,left=of A](B){{$v_I$}};
\node[va,right=of A](C){};
\node[va,right=of C, yshift=.5cm](D){};
\node[va,right=of C, yshift=-.5cm](E){};

\path
(A) edge node[el]{$0$} (B)
(A) edge node[el]{$0$} (C)
(C) edge node[el]{$K+1$}(D)
(C) edge node[el,swap]{$-3K-2$}(E)
(D) edge[loopright, looseness=6, in=135, out=45] node[el,swap]{$0$} (D)
(E) edge[loopright, looseness=6, in=-135, out=-45] node[el]{$0$} (E);
\end{tikzpicture}
}
\caption{Gadget to reduce a game to its regret game.}\label{fig:AlltoRegret}
\end{center}
\end{figure}
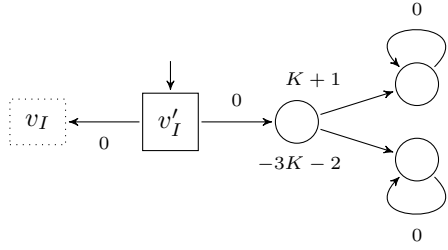

\begin{proof}[Proof of Lemma~\ref{lem:AlltoRegret}]
	Suppose $G$ is a weighted arena with initial vertex $v_I$. Consider the
	weighted arena $G'$ obtained by adding to $G$ the gadget of
	Figure~\ref{fig:AlltoRegret} with $K \defeq \frac{W}{1 - \lambda}$.  The
	initial vertex of $G'$ is set to be $v'_I$. We will show that
	\(
		\aVal(G) = K+1-{\Regret{G'}{}}/{\lambda}.
	\)

	At $v_I'$ \eve has a choice: she can choose to remain in the gadget or
	she can move to the original game $G$.  If \eve remains in the gadget
	her payoff will be $\lambda (-3K-2)$ while \adam could choose to enter
	the game and achieve a payoff of $\lambda \cdot \cVal(G)$. In this case
	her regret is $\lambda (\cVal(G)+3K+2) \geq \lambda (2K+2)$. Otherwise,
	if she chooses to play into $G$ she can achieve at most $\lambda \cdot
	\aVal(G)$. The strategy of \adam which maximizes regret against this
	choice of \eve is the one which remains in the gadget. The payoff for
	\adam is $\lambda(K+1)$ in this case. Hence, the regret of the game in
	this scenario is $\lambda(K+1 - \aVal(G)) \leq \lambda(2K + 1)$.
	Clearly she will choose to enter the game and $\Regret{G'}{} =
	\lambda(K+1-\aVal(G))$.
\end{proof}

\section{Missing Proofs from Section~\ref{sec:pos-adversary}}

\subsection{Proof of Claim~\ref{cla:strat-transfer-adam2}}
\label{sec:proof-strat-transfer-adam2}
We will now argue that if $\tau \in \StrAllA(\tilde{G})$ is a winning strategy
for \adam in $\tilde{G}$, then for all $\sigma \in \StrAllE(G)$, there exist
$t_{\tau\sigma} \in \StrPosA(G)$ and $s_{\tau\sigma} \in \StrAllE(G)$ such
that
\(
	\StratVal{}{}{s_{\tau\sigma}}{t_{\tau \sigma}} -
	\StratVal{}{}{\sigma}{t_{\tau \sigma}}
\) is at least
\begin{equation}\label{eqn:min-unsafe}
	\lambda^{|V|(|E| + 1)} 
	\min_{\substack{
		\left( (u,C),(v,D) \right) \in
		\tilde{\mathcal{B}}\\
		\tau
		\in \StrPosA(G\restriction C)
	}}\{
		\cVal^u_{\lnot v}(G \times \tau) 
		-w(u,v) - \lambda
		\cVal^v(G \times \tau)
		\}.
\end{equation}

The argument is straightforward and based on the bijection between plays from
$G$, which are consistent with positional strategies of \adam, and plays in
$\tilde{G}$. Recall that safety games are positionally determined. That is,
either \eve has a positional strategy which allows her to perpetually avoid the
unsafe edges against any strategy for \adam, or \adam has a positional strategy
which ensures that---regardless of the behaviour of \eve---the play eventually
traverses some unsafe edge. Thus, since we assume $\tau \in \StrAllA(\tilde{G})$
is winning for \adam in $\tilde{G}$ we can assume that $\tau$ is in fact a
positional strategy for \adam in $\tilde{G}$. Now consider an arbitrary strategy
$\sigma$ for \eve in $G$. We note, once more, that $\tau$ is a strategy for
\adam in $G$, not only in $\tilde{G}$.  Furthermore, $\tau$ is a positional
strategy for \adam in $G$. Conversely, $\sigma$ is a valid strategy for \eve in
$\tilde{G}$. These facts follow from the definition of $\learned(\cdot)$ and
construction $\tilde{G}$. Since $\tau$ is winning for \adam in $\tilde{G}$, the
play $\tilde{\out{ }{\sigma}{\tau}}$ traverses an unsafe edge. In fact, since
$\tau$ is positional, the unsafe edge is necessarily traversed in at most
$|V|(|E| + 1)$ steps---that is, at most the length of the longest simple path in
$\tilde{G}$. Let us write $(\tilde{v}_i, \tilde{v}_{i+1}) = \left(
(v_i,C_i),(v_{i+1},C_{i+1}) \right)$ for the traversed unsafe edge at step $i
\le |V|(|E| + 1)$. By definition of $\tilde{\mathcal{B}}$ we have that there
exists $t_{\tau \sigma} \in \StrPosA(G\restriction C_i)$ such that
\[
         	\cVal^{v_i}_{\lnot v_{i+1}}(G \times t_{\tau \sigma}) 
		-w(v_i,v_{i+1}) - \lambda
		\cVal^{v_i}(G \times t_{\tau \sigma}).
\]
We now move from the game $\tilde{G}$ back to the original game $G$. Henceforth,
we consider the play $\out{ }{\sigma}{\tau} = v_0 v_1 \dots$ in $G$ which corresponds to
$\tilde{\out{ }{\sigma}{\tau}} = (v_0,C_0) (v_1,C_1) \dots$ in $\tilde{G}$.
It is easy to see that $\out{ }{\sigma}{\tau}[..i]$ is consistent with
$t_{\tau \sigma}$. Hence, $\out{ }{\sigma}{t_{\tau \sigma}}$ traverses edge
$(v_i, v_{i+1})$ corresponding to bad edge $(\tilde{v}_i, \tilde{v}_{i+1})$ in
$\tilde{G}$. Finally, by determinacy of discounted-sum games and by virtue of $G
\times t_{\tau \sigma}$ being a finite weighted arena, we have that there is a
strategy $s_{\tau \sigma} \in \StrAllE(G \times t_{\tau \sigma})$ such that
$\StratVal{G}{v_i}{s_{\tau \sigma}}{t_{\tau \sigma}} = \cVal^{v_i}(G \times t_{\tau
\sigma})$. It then follows from the definition of $\cVal$ and $G \times s_{\tau
\sigma}$ that
\(
	\StratVal{G}{v_I}{s_{\tau\sigma}}{t_{\tau \sigma}} -
	\StratVal{G}{v_I}{\sigma}{t_{\tau \sigma}}
\) is at least the value from Equation~\eqref{eqn:min-unsafe},
just as required.\qed

\subsection{Proof of Claim~\ref{cla:strat-transfer-eve2}}
Let us show that if $\sigma \in \StrAllE(\tilde{G})$ is a winning strategy for
\eve in $\tilde{G}$, then there is $s_\sigma \in \StrAllE(G)$ such that
$\regret{s_\sigma}{G}{\StrAllE,\StrPosA} = 0$. The intuition behind the argument
is the same as for the proof of Claim~\ref{cla:strat-transfer-eve}. However, in
this case we first need to describe how to construct the strategy for \eve in
$G$ from a strategy for her in $\tilde{G}$.

\paragraph*{A regret-free strategy from $\tilde{G}$.}
Observe that, by construction of $\tilde{G}$, for any vertex $(u,C) \in
\hat{\VtcE}$ and any edge $(u,v) \in E$ there is exactly one corresponding edge
in $\tilde{G}$: $\left( (u,C), (v,C) \right)$. Given a vertex $(u,C)$ from
$\tilde{G}$, denote by $\proj{(u,C)}{1}$ the vertex $u$. Now, given a strategy
$\sigma \in \StrAllE(\tilde{G})$ we define $s_\sigma \in \StrAllE(G)$ as follows
\[
	s_\sigma(v_0 v_1 v_2 \dots) = \proj{\sigma( (v_0,C_0) (v_1, C_1 = C_0 \cap
	\learned(v_0 v_1)) (v_2, C_1 \cap \learned(v_1 v_2)) \dots)}{1}
\]
where $C_0 = E$. It follows from the fact that we have a bijective mapping from
plays in $\tilde{G}$ to plays in $G$ which are consistent with positional
strategies for \adam, that $s_\sigma$ is a valid strategy for \eve in $G$ when
playing against a positional adversary. Additionally, it is easy to see that
$s_\sigma$ can be realized using finite memory only. The memory required
corresponds to the subsets of $E$. The current memory element is determined by
the applying the operator $\learned(\cdot)$ to the current play prefix.

Now that we have our strategy $s_\sigma$ for \eve in $G$, we proceed by proving
the analogue of Claim~\ref{cla:win-is-wco} in this setting.

\begin{claim}\label{cla:win-is-wco2}
	If $\sigma \in \StrAllE(\tilde{G})$ is a winning strategy for \eve in
	$\tilde{G}$, then
	\begin{equation}\label{eqn:ensure-aval2}
		\forall \tau \in \StrPosA(G), \forall i \ge 0 :
		\PlayVal{\out{ }{s_\sigma}{\tau}[i..] = v_i\dots}
		\ge \cVal^{v_i}(G \times \tau).
	\end{equation}
\end{claim}
\begin{proof}
	To convince the reader that $s_\sigma$ has the property from
	Equation~\eqref{eqn:ensure-aval2}, we consider the synchronized
	product of $G$ and $s_\sigma$---that is, the synchronized product
	of $G$ and the finite Moore machine realizing $s_\sigma$.  As
	$s_\sigma$ is a finite memory strategy, then this product, which
	we denote in the sequel by $G \times s_\sigma$, is finite.
	Now, towards a contradiction, suppose that
	Equation~\eqref{eqn:ensure-aval2} does not hold for $s_\sigma$. That is,
	there is some $\tau \in \StrPosA(G)$ for which the property fails.
	Further, let us consider an alternative (memoryless) strategy
	$\sigma'$ of \eve which ensures $\cVal^v(G \times \tau)$ from all $v \in V$.
	The latter exists by definition of $\cVal(G \times \tau)$ and memoryless
	determinacy of discounted-sum games (see, \eg~\cite{zp96}).
	
	Let $H$ denote a copy of $G \times s_\sigma$ where all edges
	induced by $E$ from $G$ are added---not just the ones allowed
	by $s_\sigma$---and $H \arestriction \sigma'$ denote the
	sub-graph of $H$ where only edges allowed by $\sigma'$ are left.
	Intuitively, both $G \times s_\sigma$ and $H \arestriction \sigma'$
	are sub-structures of $\tilde{G}$ with a weight function $\tilde{w}$
	lifted from $w$ to the blown-up vertex set $\tilde{V}$. This is due to
	the way in which we constructed $s_\sigma$.

	Since, by assumption, $s_\sigma$ does not have the property of
	Equation~\eqref{eqn:ensure-aval2} then the edges present in at
	least one vertex from $H \arestriction \sigma'$ and $G \times
	\sigma$ differ. Note that such a vertex $(u,C)$ is necessarily such
	that $u \in \VtcE$---and $C$ is a ``memory element'' from the machine
	realizing $s_\sigma$ corresponding to a subset of $E$ obtained via
	$\learned(\cdot)$. Furthermore, from our definition of a
	strategy, we know that there is a single outgoing edge from it
	in both structures. Let us write $(u,v)$---instead of $\left(
	(u,C),(v,D) \right)$---for the
	edge in $G \times s_\sigma$ and $(u,v')$ for the edge in $H
	\arestriction \sigma'$. Recall that $s_\sigma$ is winning for \eve
	in $\tilde{G}$. Thus, we have that $(u,v) \not\in 
	\tilde{\mathcal{B}} = \{\left( (u,C),(v,D)
	\right) \in \hat{E} \st u \in \VtcE \text{ and } \exists \tau' \in
	\StrPosA(G\restriction C), w(u,v) + \lambda \cVal^v(G \times
	\tau') < \cVal^u_{\lnot v}(G \times \tau')\}$. It follows
	that
	\[
		w(u,v) + \lambda  \cVal^v(H \times \tau) \ge \cVal^{v'}(H \times
		\tau).
	\]
	Thus, the strategy $\sigma''$ of \eve which takes $(u,v)$
	instead of $(u,v')$ and follows $\sigma'$ otherwise---indeed,
	this might mean $\sigma''$ is no longer memoryless---also achieves at
	least $\cVal^u(H \times \tau)$ from $u$ onwards. Notice that
	this process can be repeated for all vertices in which the two
	structures differ. Further, since both are finite, it will
	eventually terminate and yield a strategy of \eve which plays
	exactly as $s_\sigma$ and for which, since $\tau$ was chosen
	arbitrarily, Equation~\eqref{eqn:ensure-aval2} holds. Contradiction.
\end{proof}

It follows immediately that $\regret{s_\sigma}{G}{\StrAllE,\StrPosA} = 0$.
Indeed, if we suppose that this is not the case, then there exists a strategy
$\sigma' \in \StrAllE(G)$ such that
\[
	\exists \tau \in \StrPosA(G) : \StratVal{}{}{s_\sigma}{\tau} <
	\StratVal{}{}{\sigma'}{\tau}.
\]
The above directly contradicts Claim~\ref{cla:win-is-wco2}.  \qed

\subsection{Proof of Theorem~\ref{thm:memlessadversary}}
In this section we present sufficient modifications to our definitions from
Section~\ref{sec:any-adversary} in order for the techniques used therein to be
adapted for this case. Particularly, our notion of regret of a play and the
safety game used to decide the existence of regret-free strategies need to take
into account the fact that witnessing edges taken by \adam affects previously
observed local regrets. That is, we formalize the intuition that alternative
plays must also be consistent with the behaviour of \adam that we have witnessed
in the current play.

We are now ready to define the regret of a play in a game against a positional
adversary. Given a play $\pi = v_0 v_1 \dots$, we let
\[
	\regret{}{\pi}{} \defeq \sup\{\lambda^i(\cVal^{v_i}_{\lnot v_{i+1}}(G
	\restriction \learned(\pi)) - \PlayVal{\pi[i..]} \st v_i \in \VtcE \}
	\cup \{0\}.
\]
Consider now a play prefix $\rho = v_0 \dots v_j$. We let the regret of $\rho$ be
\[
	\max\{\lambda^i(\cVal^{v_i}_{\lnot v_{i+1}}(G
	\restriction \learned(\rho[i..j])) - \PlayVal{\rho[i..j]} \st 0 \le i < j \text{
	and } v_i \in \VtcE \}
	\cup \{0\}.
\]

We will now re-prove Lemma~\ref{lem:play-regret} in the current
setting.
%For completeness, we state it here with the necessary notation changes,
%and provide a proof in appendix.
\begin{lemma}\label{lem:play-regret-memless}
	For any strategy $\sigma$ of Eve,
	\[
		\regret{\sigma}{G}{\StrAllE,\StrPosA} = \sup \{ \regret{}{\pi}{}
		\st \pi\text{ is consistent with } \sigma \text{ and some } \tau
		\in \StrPosA\}.
	\]
\end{lemma}
\begin{proof}
	%What follows is essentially the same proof as was provided for
	%Lemma~\ref{lem:play-regret}. We adapt the argument to account for the
	%new definition of the regret of a play.
	%
	Consider any $\sigma, \sigma' \in \StrAllE$ and $\tau \in \StrPosA$ such
	that $\out{}{\sigma}{\tau} \neq \out{}{\sigma'}{\tau}$.  Let us write
	$\out{}{\sigma}{\tau} = v_0 v_1 \dots$ and $\out{}{\sigma'}{\tau} = v'_0
	v'_1 \dots$ and denote by $\ell$ the length of the longest common
	prefix of $\out{}{\sigma}{\tau}$ and $\out{}{\sigma'}{\tau}$. We claim
	that
	\begin{equation}\label{equ:inverse-combine-behavior2}
		\lambda^{\ell}  \bigl(
		\cVal^{v_\ell}_{\lnot
			v_{\ell + 1}}(G \restriction
			\learned(\out{}{\sigma}{\tau})) -
		\PlayVal{\out{}{\sigma}{\tau}}[\ell..] \bigr)
		\ge
		\lambda^{\ell}  \bigl(
		\PlayVal{\out{}{\sigma'}{\tau}[\ell..]}-
		\PlayVal{\out{}{\sigma}{\tau}[\ell..]}
		\bigr).
	\end{equation}
	Indeed, if we assume it is not the case, we then get that
	\[
		\cVal^{v'_{\ell+1}}(G \restriction
		\learned(\out{}{\sigma}{\tau})) <
		\PlayVal{\out{}{\sigma'}{\tau}[\ell + 1..]}.
	\]
	However, recall that $G \times \tau$ is a sub-arena of
	$G \restriction \learned(\out{}{\sigma}{\tau})$. Thus, the co-operative
	value \eve can obtain in the former, say by playing $\sigma'$, must be
	at most that which she can obtain in the latter. Contradiction. 

	Note that there is another positional strategy $\tau'$ for \adam and a
	second alternative strategy $\sigma''$ for \eve which do give us
	equality for Equation~\eqref{equ:inverse-combine-behavior2}. For this
	purpose, we choose $\tau'$ so that $\tau' \in \StrPosA(G \restriction
	\learned(\out{}{\sigma}{\tau}))$---so that $\out{ }{\sigma}{\tau}$ is
	also consistent with $\tau'$, thus $\learned(\out{ }{\sigma}{\tau}) =
	\learned(\out{ }{\sigma}{\tau'})$ (see
	Lemma~\ref{lem:learned-properties})---and also such that
	\[
		\cVal^{v'_{\ell+1}}(G \times \tau') = \cVal^{v'_{\ell+1}}(G
		\restriction \learned(\out{ }{\sigma}{\tau})).
	\]
	We choose $\sigma''$ so that it follows $\sigma$ for $\ell$ turns, goes
	to $v'$, and then plays co-operatively with $\tau'$ from $v'$. More
	formally, let $\sigma''$ be a strategy for \eve such that
	$\out{}{\sigma}{\tau}[..\ell] = \out{}{\sigma''}{\tau}[..\ell]$ and
	therefore, by choice of $\tau'$, such that
	$\out{}{\sigma}{\tau'}[..\ell] = \out{}{\sigma''}{\tau'}[..\ell]$
	%.
	%In other words, the strategy $\sigma''$ is then chosen so that $\ell$ is
	%the length of the longest common prefix of $\out{}{\sigma}{\tau'}$
	%and $\out{}{\sigma''}{\tau'}$ and so that
	and so that
	\[
		\PlayVal{\out{ }{\sigma''}{\tau'}[\ell..]} =
		\cVal^{v'_{\ell+1}}(G \times \tau').
	\]

	It follows from Equation~\eqref{equ:inverse-combine-behavior2} and the
	above arguments that for all $\sigma \in \StrAllE$, if there are $\tau
	\in \StrPosA$ and $\sigma' \in \StrAllE$ such that $\out{}{\sigma}{\tau}
	\neq \out{}{\sigma'}{\tau}$ then
	\begin{equation}\label{equ:full-combine-behavior2}
		\sup_{\tau,\sigma' \st
		\out{}{\sigma}{\tau} \neq \out{}{\sigma'}{\tau}}
		\lambda^{\ell}  \bigl(
		\PlayVal{\out{}{\sigma'}{\tau}[\ell..]}-
		\PlayVal{\out{}{\sigma}{\tau}[\ell..]}
		\bigr) =
		\lambda^{\ell}  \bigl(
		\cVal^{v_\ell}_{\lnot
			v_{\ell + 1}}(G \restriction \learned(\out{
			}{\sigma}{\tau})) -
		\PlayVal{\out{}{\sigma}{\tau}} \bigr).
	\end{equation}

	We are now able to prove
	the result. That is, for any strategy $\sigma$ for \eve:
	\begin{align*}
		& \sup\{ \regret{}{\pi}{} \st \pi \text{ is
		consistent with } \sigma \text{ and some } \tau \in \StrPosA\} & \\
		= & \sup_{\tau \in \StrPosA} \regret{}{\out{}{\sigma}{\tau} =
		v_0 v_1 \dots }{}
		& \text{def. of } \out{}{\sigma}{\tau}\\
		= & \sup_{\tau \in \StrPosA} \max\left\{0,
		\sup_{\substack{i \ge 0\\
		v_i \in \VtcE}} \lambda^i  \left(
		\cVal^{v_i}_{\lnot v_{i+1}}(G \restriction \learned(\out{
		}{\sigma}{\tau})) -
		\PlayVal{\out{}{\sigma}{\tau}[i..]} \right) \right\}
		& \text{def. of } \regret{}{\out{}{\sigma}{\tau}}{}\\
		= & \sup_{\tau \in \StrPosA} \max\left\{0,
		\sup_{\sigma' \st \out{}{\sigma}{\tau} \neq
		\out{}{\sigma'}{\tau}}
		\lambda^\ell  \left(
		\PlayVal{\out{}{\sigma'}{\tau}[\ell..]} -
		\PlayVal{\out{}{\sigma}{\tau}[\ell..]} \right) \right\}
		& \text{by Eq.~\eqref{equ:full-combine-behavior2}}\\
		= & \sup_{\tau \in \StrPosA} \max\left\{0,
		\sup_{\sigma' \st \out{}{\sigma}{\tau} \neq
		\out{}{\sigma'}{\tau}}
		\left(
		\StratVal{}{}{\sigma'}{\tau} -
		\StratVal{}{}{\sigma}{\tau} \right) \right\}
		& \text{def. of } \PlayVal{\cdot},\ell\\
		= & \sup_{\tau \in \StrPosA}
		\sup_{\sigma' \in \StrAllE}
		\left(
		\StratVal{}{}{\sigma'}{\tau} -
		\StratVal{}{}{\sigma}{\tau} \right)
		& 0 \text{ when } \out{}{\sigma}{\tau} = \out{}{\sigma'}{\tau}
	\end{align*}
	as required.
\end{proof}

We will now state and prove a restricted version of
Lemma~\ref{lem:expBnd}. Intuitively,
for a play $\pi$, we will not be able to consider a deviation with respect to a
prefix of $\pi$. Rather, we are forced to take the co-operative value with
respect to the set $\learned(\pi)$---that is, the edges consistent with any
positional strategy \adam might be playing---even after the bound on where the
best deviation occurs.
\begin{lemma}\label{lem:expBnd-memless}
	Let $\pi$ be a play in $G$ and suppose $0 < r \leq \regret{}{\pi}{}$.
	Let
	\[
		N(r) \defeq \left \lfloor (\log r + \log (1-\lambda) -
		\log(2W))/\log \lambda \right \rfloor + 1.
	\]
	Then $\regret{}{\pi}{}$ is equal to
	\[
		\max_{\substack{0 \le i < N(r)\\ v_i \in
		\VtcE}}\{\lambda^i(\cVal^{v_i}_{\lnot v_{i+1}}(G \restriction
		\learned(\pi)) - \PlayVal{\pi[i..N(r)]}\} - \lambda^{N(r)}
		\PlayVal{\pi[{N(r)}..]}.
	\]
\end{lemma}
\begin{proof}
	Observe that $N(r)$ is such that $\frac{2W\lambda^{N(r)}}{1-\lambda}<r$.
	Hence, we have that for all $i \ge N(r)$ such that $v_i \in \VtcE$ it
	holds that $\lambda^i (\cVal^{v_i}_{\lnot v_{i+1}}(G) -
	\PlayVal{\pi[i..]}) \leq \frac{2W\lambda^{N(r)}}{1-\lambda} < r$.
	Clearly, since $\cVal^{v_i}_{\lnot v_{i+1}}(H) \le \cVal^{v_i}_{\lnot
	v_{i+1}}(G)$ holds for any sub-arena $H$ of $G$, we have that
	\[
		\lambda^i (\cVal^{v_i}_{\lnot v_{i+1}}(G \restriction
		\learned(\pi)) - \PlayVal{\pi[i..]})
		\leq \frac{2W\lambda^{N(r)}}{1-\lambda} < r.
	\]
	It thus follows that
	\begin{align*}
		\regret{}{\pi}{} &=
		\sup\{
		\lambda^i  (\cVal^{v_i}_{\lnot v_{i+1}}(G \restriction \learned(\pi)) -
		\PlayVal{\pi[i..]})\st i \ge 0 \text{ and } v_i \in \VtcE\}\\
		&= \max_{\substack{0 \leq i < {N(r)}\\
		v_i \in \VtcE}} 
		\lambda^i  \left(\cVal^{v_i}_{\lnot v_{i+1}}(G \restriction
		\learned(\pi)) -
		\PlayVal{\pi[i..{N(r)}]}\right) - \lambda^{N(r)} 
		\PlayVal{\pi[{N(r)}..]}
	\end{align*}
	as required.
\end{proof}

%The above result once more allows us to argue that
%Proposition~\ref{pro:regret-of-tree} holds in this setting as well. Hence,
%computing the regret of a game can be done in alternating pseudo-polynomial
%time. That is,
%\begin{proposition}\label{pro:aexptime-solution2}
%	Computing the regret value of a game, playing against a positional
%	adversary, can be done in time $\mathcal{O}(\max\{|V|(|E|+1), N(b_G)\}$.
%\end{proposition}
%\begin{corollary}
%	Let $\mu \defeq |\Delta|^{b_G}$. It holds that
%	$\Regret{G}{\StrE^\mu,\StrPosA} = \Regret{G}{\StrAllE,\StrPosA}$.
%\end{corollary}
%
\begin{figure}
\begin{center}
\begin{tikzpicture}
	\node (root) at (0,5) {$v_I$};
	\node (leftcorner) at (-4,0) {};
	\node (rightcorner) at (4,0) {};

	\path[-]
	(root) edge (leftcorner)
	(root) edge (rightcorner);

	\node[va] (adam-choice) at (0,1) {$v_j$};
	\node (bottom2) at (-1.5,0) {$\pi$};
	\node (bottom1) at (1.5,0) {$\pi'$};
	\path [->,decoration={zigzag,segment length=4,amplitude=.9,
		post=lineto,post length=2pt}]
	(root) edge[decorate] (adam-choice)
	(adam-choice) edge[decorate,green] (bottom2)
	(adam-choice) edge[decorate,red] (bottom1)
	;

	\node[ve,fill=white] (alt1) at (0,3) {$v_{i'}$};
	\node[label={right:$\cVal^{v_{i'}}_{\lnot v_{i'+1}}(G \restriction
		\learned(\pi))$}] (alttarget1) at (2,2.6) {};
	\path
	(alt1) edge[green] (alttarget1)
	;
	\node[ve,fill=white] (alt2) at (0,4) {$v_{i}$};
	\node[label={left:$\cVal^{v_i}_{\lnot v_{i+1}}(G \restriction
		\learned(\pi'))$}] (alttarget2) at (-1.3,3.5) {};
	\path
	(alt2) edge[red] (alttarget2)
	;

	\draw[dashed,-,gray] (-3.4,2.3) -- (4,2.3);
	\node[gray] at (-4,2.3) {$N(b_G)$};
\end{tikzpicture}
\end{center}
\caption{Let $\rho$ denote the play prefix $v_0 \dots v_j$. The alternative play from
	$v_{i'}$ is better than the one from
	$v_i$ w.r.t $\rho$. However, for play
	$\pi'$ extending $\rho$, the alternative play from $v_i$ becomes better
	than the one from $v_{i'}$
	if $\lambda^{i' -i}\cVal^{v_{i'}}_{\lnot v_{i' + 1}}(G \restriction
	\learned(\pi'))$ is smaller than $\cVal^{v_i}_{\lnot v_i + 1}(G
	\restriction \learned(\pi')) - \PlayVal{\rho[i..i']}$.}
\label{fig:two-deviations}
\end{figure}

\begin{figure}
\begin{center}
\begin{tikzpicture}
	\node (root) at (0,5) {$v_I$};
	\node (leftcorner) at (-4,0) {};
	\node (rightcorner) at (4,0) {};

	\path[-]
	(root) edge (leftcorner)
	(root) edge (rightcorner);

	\node[va] (adam-choice) at (0,1) {$v_j$};
	\node (bottom2) at (-1.5,0) {$\pi$};
	\node (bottom1) at (1.5,0) {$\pi'$};
	\path [->,decoration={zigzag,segment length=4,amplitude=.9,
		post=lineto,post length=2pt}]
	(root) edge[decorate] (adam-choice)
	(adam-choice) edge[decorate,green] (bottom2)
	(adam-choice) edge[decorate,red] (bottom1)
	;

	\node[ve,fill=white] (alt1) at (0,3) {$v_{i'}$};
	\node[label={right:$\cVal^{v_{i'}}_{\lnot v_{i'+1}}(G \restriction
		\learned(\rho))$}] (alttarget1) at (2,2.6) {};
	\path
	(alt1) edge[green] (alttarget1)
	;
	\node[ve,fill=white] (alt2) at (0,4) {$v_{i}$};
	\node[label={left:$\cVal^{v_i}_{\lnot v_{i+1}}(G \restriction
		\learned(\pi'))$}] (alttarget2) at (-1.3,3.5) {};
	\path
	(alt2) edge[red] (alttarget2)
	;

	\draw[dashed,-,gray] (-3.4,2.3) -- (4,2.3);
	\node[gray] at (-4,2.3) {$N(b_G)$};
	\draw[dashed,-,gray] (-3.4,1.8) -- (4,1.8);
	\node[gray] at (-4,1.8) {$\nu(b_G)$};
\end{tikzpicture}
\end{center}
\caption{A play $\pi'$ extending $\rho$ in a way such that $\StrPosA(G
	\restriction \learned(\pi')) \cap \mrs(\rho) = \emptyset$ cannot have
	more regret than a play $\pi$ extending $\rho$ for which $\StrPosA(G
	\restriction \learned(\pi)) \cap \mrs(\rho) \neq \emptyset$---for $\rho$
	longer than $\nu(b_G)$.}
\label{fig:two-deviations-bound}
\end{figure}

The main difference between the problem at hand and the one we solved in
Section~\ref{sec:any-adversary} is that, when playing against a positional adversary,
information revealed to \eve in the present can affect the best alternatives to
her current behaviour. Some definitions are in order. Let $\rho = v_0 \dots v_j$
be a play prefix.
The \emph{maximal-regret points} of $\rho$, denoted by $\mrp(\rho)$, is the set
\[
	\{0 \le i < j \st v_i \in \VtcE \text{ and }
	\lambda^i\left( \cVal^{v_i}_{\lnot v_{i+1}}(G\restriction
	\learned(\rho[..j])) - \PlayVal{\rho[i..j]}\right) = \regret{}{\rho}{} \};
\]
and the \emph{maximal-regret strategies} of $\rho$, written $\mrs(\rho)$, is
equal to
\[
	\left\{ \tau \in \StrPosA(G\restriction \learned(\rho[..j]))
	\st
	\bigvee_{i \in
	\mrp(\rho)} \cVal^{v_i}_{\lnot v_{i+1}}(G\restriction \learned(\rho[..j]))
	= \cVal^{v_i}_{\lnot v_{i+1}}(G \times \tau)\right\}.
\]
The above definitions are meant to capture the intuition that, upon witnessing a
new choice of \adam, we can reduce the size of the set of possible positional
strategies he could be using. Consider a play prefix $\rho$.  The maximal-regret
points of $\rho$ correspond to the positions at which best alternatives to
$\rho$ occur. The maximal-regret strategies of $\rho$ is the set of
positional strategies of \adam, $\rho$ consistent
with them, such that at least one of the best alternatives to $\rho$ is
consistent with them. Recall from Lemma~\ref{lem:learned-properties} $(ii)$ that
a play prefix $\rho$ is consistent with a positional strategy $\tau \in
\StrPosA(G)$ if and only if $\tau \in \StrPosA(G \restriction \learned(\rho))$.
We can, therefore, think of the set of edges $\learned(\rho)$ as representing
the set of all positional strategies for \adam in $G$ that $\rho$ is consistent
with, \ie~$\{\tau \in \StrPosA(G) \st \rho \text{ is consistent with } \tau\}$.
Let us write $\StrPosA(G,\rho)$ for the set we just described.  Let $\beta$ be
the value of one of the best alternatives to $\rho$. If $\beta' < \beta$ is
the value of one of the best alternatives to $\rho'$, then we know the
best alternatives to $\rho$ are not consistent with any strategy from
$\StrPosA(G,\rho')$. Then, according to our definition of maximal-regret
strategies, this also means that $\mrs(\rho) \cap \StrPosA(G,\rho') =
\emptyset$. The converse is also true.

%The main difference between the problem at hand and the one we solved in
%Section~\ref{sec:any-adversary} is that, when playing against a positional adversary,
%information revealed to \eve in the present can affect the best alternatives to
%her current behaviour. Indeed, since \adam is playing positionally,
%if he reveals a ``minimizing choice''---that is, from play prefix $\rho$ \adam
%goes to $v'$ and from $\rho' = \rho \cdot v'$ the value of a best alternative to
%$\rho'$ is strictly smaller than the value of a best alternative to
%$\rho$---against \eve, then some alternative are no longer better.
As an example, consider the situation depicted in
Figure~\ref{fig:two-deviations}.  If, from $v_j$, the play $\pi'$ is obtained
and we have that $\StrPosA(G \restriction \learned(\pi') \cap \mrs(\rho)$ is
empty, then the deviation from $v_{i'}$ might no longer be a best alternative.
Indeed, there is no positional strategy of \adam which allows the deviation from
$v_{i'}$ to obtain the value we assumed (from just looking at the prefix $\rho$)
and which is also consistent with $\pi'$. In order to deal with this, we need
some more definitions.

Assume that $\Regret{G}{\StrAllE,\StrPosA} \ge b_G$. For a play prefix $\rho =
v_0 \dots v_n$ with $n \ge N(b_G)$, let us define the value $\delta_\rho$
($\delta$ for \emph{drop}) as
\[
	\min_{\substack{0 \le i \le j < N(b_G)\\
	\tau,\tau' \in \StrPosA(G \restriction \learned(\rho))}}
	\left| \lambda^i \left( 
		\cVal^{v_i}_{\lnot v_{i+1}}(G \times \tau) -
		\PlayVal{\rho[i..j]} \right) - \lambda^{j} \cVal^{v_j}_{\lnot
			v_{j+1}}(G \times \tau'))
	\right|.
\]
Intuitively $\delta_\rho$ is the minimal drop of the regret achievable by a
better alternative (given the information we can extract from $\rho$).

\paragraph*{The smallest possible drop.}
Let us derive a universal lower bound on $\delta_\rho$ for all $\rho$ of length
at least $N(b_G)$. In order to do so we will recall ``the shape'' of the co-operative
value of $G$. Recall the $\cVal$ in a discounted-sum game can be obtained by
supposing \eve controls all vertices and computing $\aVal$ instead. It then
follows from positional determinacy of discounted-sum games that the $\cVal$ is
achieved by a \emph{lasso} in the arena $G$. More formally, we know that there
is a play $\pi$ in $G$ of the form
\[
	\pi = v_0 \dots v_{k-1} (v_k \dots v_\ell)^\omega
\]
where $0 \le k < \ell \le |V|$, and such that $\PlayVal{\pi} = \cVal^{v_0}(G)$.
Let us write $\lambda = \frac{\alpha}{\beta}$ with $\alpha,\beta \in
\mathbb{Z}$. One can then verify that
\begin{lemma}\label{lem:shape-cval}
	For all sub-arenas $H$ of $G$, for all vertices $v \in V$, there exists
	$N \in \mathbb{Z}$ such that $\cVal^v(H) = \frac{N}{D}$ where
	\(
		D \defeq \beta^{|V|} (\beta^{|V|} - \alpha^{|V|}).
	\)
\end{lemma}
It then follows from the definition of $\delta_\rho$ that:
\begin{lemma}\label{lem:universal-drop}
	For all play prefixes $\rho = v_0 \dots v_n$
	such that $n \ge N(b_G)$ we have that
	\[
		\delta_\rho > \frac{1}{\beta^{N(b_G)} D}.
	\]
\end{lemma}

\paragraph*{Formalizing our claims.}
We can now prove a replacement for Lemma~\ref{lem:expBnd} holds in this context.
%by choosing $i \in
%\mathbb{N}$ big enough so that $i \ge N(b_G)$ and
%\[
%	\lambda^i \left( \frac{W}{1-\lambda} \right)
%	< \frac{1}{\beta^{N(b_G) + |V|}D}.
%\]
%to hold. More formally:
\begin{lemma}\label{lem:expBnd2}
	Let $\pi$ be a play in $G$ and assume $\Regret{G}{\StrAllE,\StrPosA} > 0$.
	Let
	$\nu(b_G)$ denote the value
	\[
	%\begin{cases}
	%	|V|(|E| + 1) & \text{if } b_G = 0\\
		N(b_G) + \left\lfloor \frac{\log(1 - \lambda) - \log W - (N(b_G)
			+ |V|)\log \beta - \log(\beta^{|V|} -
			\alpha^{|V|})}{\log \lambda} \right\rfloor + 1.% &
	%		\text{otherwise.}
	%\end{cases}
	\]
	Then for all $\sigma \in \StrAllE$,
	\[
		\sup_{\tau \in \StrPosA} \regret{}{\out{}{\sigma}{\tau}}{} =
		\sup_{\tau \in \StrPosA}
		\regret{}{\out{}{\sigma}{\tau}[..{\nu(b_G)}]}{} - \lambda^{\nu(b_G)}
		\PlayVal{\out{}{\sigma}{\tau}[{\nu(b_G)}..]}.
	\]
\end{lemma}
\begin{proof}
	Let us consider throughout this argument an arbitrary $\sigma \in
	\StrAllE$.
	From Lemma~\ref{lem:expBnd-memless} and the fact that
	$\nu(b_G)$ is such that $N(b_G)$, we know that $\sup_{\tau \in \StrPosA}
	\regret{}{\out{}{\sigma}{\tau} = v_0 \dots}{}$ equals
	\[
		\sup_{\tau \in \StrPosA}
		\max_{\substack{0 \le i < \nu(b_G)\\ v_i \in
		\VtcE}}\{\lambda^i(\cVal^{v_i}_{\lnot v_{i+1}}(G \restriction
		\learned(\out{}{\sigma}{\tau})) -
		\PlayVal{\out{}{\sigma}{\tau}[i..\nu(b_G)]}\} - \lambda^{\nu(b_G)}
		\PlayVal{\out{}{\sigma}{\tau}[{\nu(b_G)}..]}.
	\]
	Now, also note that $\nu(b_G)$ was chosen so that
	\[
		\frac{W\lambda^{\nu(b_G)}}{1-\lambda} < \frac{1}{\beta^{N(b_G)
		+ |V|} D}.
	\]
	Hence, for all $\tau' \in \StrPosA$ if we write
	$\out{ }{\sigma}{\tau'} = v'_0 \dots$, then
	for all $j \ge \nu(b_G)$ such that $v'_j \in \VtcE$ it holds that 
	\[
		-\frac{1}{\beta^{N(b_G) + |V|} D} <
		\lambda^i 
		\PlayVal{\out{ }{\sigma}{\tau'}[i..]}) <
		\frac{1}{\beta^{N(b_G) + |V|} D}.
	\]
	It then follows from Lemma~\ref{lem:universal-drop} and the definition
	of $\delta_{\out{ }{\sigma}{\tau'}[..\nu(b_G)]}$ that, if there exists
	$\ell \ge \nu(b_G)$ such that for all $0 \le k \le \nu(b_G)$ with $v'_k
	\in \VtcE$
	\[
		\cVal^{v'_k}_{\lnot v'_{k+1}}(G \restriction \learned(\pi[..\ell])) <
		\cVal^{v'_k}_{\lnot v'_{k+1}}(G \restriction
		\learned(\pi[..\nu(b_G)]))
	\]
	then $\regret{}{\out{ }{\sigma}{\tau'}}{} < \regret{}{\out{
	}{\sigma}{\tau''}}{}$ for all $\tau'' \in \mrs(\pi'[..\nu(b_G)])$. This
	is due to the fact that
	that $\out{ }{\sigma}{\tau''}[..\nu(b_G)] = \out{
	}{\sigma}{\tau'}[..\nu(b_G)]$ and
	\[
		\cVal^{v'_k}_{\lnot v'_{k+1}}(G \times \tau'')
		= \cVal^{v'_k}_{\lnot v'_{k+1}}(G \restriction
		\learned(\out{ }{\sigma}{\tau''}[..\nu(b_G)])).
	\]
	The above implies
	that for all $\sigma \in \StrAllA$ the value $\sup_{\tau
	\in \StrPosA} \regret{}{\out{}{\sigma}{\tau} = v_0 \dots}{}$ equals
	\[
		\max \{
		\cVal^{v_i}_{\lnot v_{i+1}}(G\restriction
		\learned(\out{}{\sigma}{\tau}[..\nu(b_G)])) - 
		\lambda^{\nu(b_G)}
		\PlayVal{\out{
		}{\sigma}{\tau}[..\nu(b_G)]} \st 
		0 \le i \le N(b_G) \text{ and }
		v_i \in \VtcE
		\}
	\]
	and therefore (by definition of regret of a prefix) we have that
	\[
		\sup_{\tau \in \StrPosA} \regret{}{\out{}{\sigma}{\tau}}{} =
		\sup_{\tau \in \StrPosA}
		\regret{}{\out{}{\sigma}{\tau}[..{\nu(b_G)}]}{} - \lambda^{\nu(b_G)}
		\PlayVal{\out{}{\sigma}{\tau}[{\nu(b_G)}..]}.
	\]
	as required.
\end{proof}

\paragraph*{Putting everything together.}
Let us go back to our example to illustrate how to use $\nu(b_G)$ and the drop
of a prefix.  Consider now the situation from
Figure~\ref{fig:two-deviations-bound}. Recall we have assumed $\pi'$ is a play
extending $\rho$ with $\StrPosA(G \restriction \learned(\pi')) \cap \mrs(\rho)
=\emptyset$. It follows that all
best alternatives to $\pi'$ achieve a payoff strictly smaller than
$\cVal^{v'}_{\lnot v_{i' + 1}}(G \restriction \learned(\rho))$. Thus, the regret
of $\pi'$ can only be bigger than the regret of a play $\pi$ with $\StrPosA(G
\restriction \learned(\pi)) \cap
\mrs(\rho) \neq \emptyset$ if the minimal
index $k > j$ such that $\StrPosA(G \restriction \learned(\pi'[..j])) \cap
\mrs(\rho) = \emptyset$---\ie~the turn at which \adam revealed he was not
playing a strategy from $\mrs(\rho)$---is small enough. In other words, the drop
in the value of the best alternative has to be compensated by a similar drop in
the value obtained by \eve, and the discount factor makes this impossible after
some number of turns.

\begin{proposition}
	If $\Regret{G}{\StrAllE,\StrPosA} \ge b_G$ then
	$\Regret{G}{\StrAllE,\StrPosA}$ is equal to
	\[
		\inf_{\sigma \in \StrAllE}
		\sup\{ \regret{}{\pi[..\nu(b_G)]}{} - \lambda^{\nu(b_G)}
		\aVal^{\hat{u}}(\hat{H}) \st \pi = v_0 v_1 \dots \text{ cons.
		with } \sigma \text{ and some }\tau \in \StrPosA\} 
	\]
	where
	\begin{itemize}
		\item $\hat{u} \defeq (v_{\nu(b_G)},\learned(\pi[..\nu(b_G)]))$ and
		\item $\hat{H} \defeq \hat{G} \restriction \{\left( (C,u),(D,v)
			\right) \st \StrPosA(G\restriction D) \cap
			\mrs(\pi[..\nu(b_G)]) \neq \emptyset \}$.
	\end{itemize}
\end{proposition}
\begin{proof}
	First, note that if $\Regret{G}{\StrAllE,\StrPosA} > 0$ then there
	cannot be any regret-free strategies for \eve in $G$ when playing
	against a positional adversary. It then follows from
	Corollary~\ref{cor:lower-bound-positional} that
	$\Regret{G}{\StrAllE,\StrPosA} \ge b_G$.

	Now using Lemma~\ref{lem:expBnd2} together with the definition of the
	regret of a play we get that $\Regret{G}{\StrAllE,\StrPosA}$ is equal
	to
	\[
		\inf_{\sigma \in \StrAllE} \sup\{ \regret{}{\pi[..\nu(b_G)]}{} -
		\lambda^{\nu(b_G)}  \PlayVal{\pi[\nu(b_G)..]} \st \pi \text{
		cons. } \sigma \text{ and some } \tau \in \StrPosA \}.
	\]
	Finally, note that it is in the interest of \eve to maximize the value
	$\lambda^{\nu(b_G)}  \PlayVal{\pi[\nu(b_G)..]}$ in order to minimize
	regret. Conversely, \adam tries to minimize the same value with a
	strategy from $\mrs(\pi[..\nu(b_G)])$: critically, the strategy is such
	that the prefix $\pi[..\nu(b_G)]$ is consistent with it.
	Thus, we
	can replace it by the antagonistic value from $\pi[\nu(b_G)..]$ discounted
	accordingly. In this setting we also want to force \adam to play a
	positional strategy which is consistent with deviations before $N(b_G)$
	which achieve the assumed regret of the prefix $\pi[..\nu(b_G)]$.
	More formally, we have
	\begin{align*}
		 &\inf_{\sigma \in \StrAllE}
		\sup_{\tau \in \StrPosA}
		\regret{}{\out{}{\sigma}{\tau}[..\nu(b_G)]}{} -
		\lambda^{\nu(b_G)} 
		\PlayVal{\out{}{\sigma}{\tau}[\nu(b_G)..]}\\
		=&\inf_{\substack{\sigma \in \StrAllE\\\sigma' \in \StrAllE}}
		\sup_{\substack{\tau \in \StrPosA\\
		\tau' \in \mrs(\out{}{\sigma}{\tau}[..\nu(b_G)])
		}}
		\regret{}{\out{}{\sigma}{\tau}[..\nu(b_G)]}{} -
		\lambda^{\nu(b_G)}  \StratVal{}{}{\sigma'}{\tau'}\\
		=&\inf_{\sigma \in \StrAllE}
		\sup_{\tau \in \StrAllA}
		\regret{}{\out{}{\sigma}{\tau}[..\nu(b_G)]}{} +
		\inf_{\sigma' \in \StrAllE} \sup_{
		\substack{
		\tau' \in \mrs(\out{}{\sigma}{\tau}[..\nu(b_G)])
		}}
		\left(
		-\lambda^{\nu(b_G)} 
		\StratVal{}{}{\sigma'}{\tau'}
		\right).
	\end{align*}
	It should be clear that the RHS term of the sum is equivalent to
	\[
		-\lambda^{\nu(b_G)}\aVal^{\hat{u}}(\hat{H})
	\]
	as required.
\end{proof}

The above result allows us to claim an \EXPSPACE~algorithm (when $\lambda$ is
not fixed) to compute the regret of a game. As in
Section~\ref{sec:any-adversary}, we simulate the game using an alternating
machine which halts in at most a pseudo-polynomial number of steps which depends
on $\nu(b_G)$ and, in turn, on $b_G$. After that, we must compute the
antagonistic value of $\hat{G}$. As a first step, however, we compute the safety
game $\tilde{G}$ and determine its winner.
\begin{proposition}\label{pro:aexptime-solution2}
	Computing the regret value of a game, playing against a positional
	adversary, can be done in time $\mathcal{O}(\max\{|V|(|E| + 1),
	\nu(b_G)\})$ with an
	alternating Turing machine.
\end{proposition}
The memory requirements for \eve are as follows:
\begin{corollary}\label{cor:finite-mem-regret2}
	Let $\eta \defeq |\Delta|^{d}$ where $d = \max\{|V|(|E| + 1),
	\nu(b_G)\}$. It then holds that $\Regret{G}{\StrE^{\eta},\StrPosA} =
	\Regret{G}{\StrAllE,\StrPosA}$.
\end{corollary}

\subsection{Lower bounds}

In the main body of the paper, namely in Section~\ref{sec:pos-adversary}, we
have claimed that the regret threshold problem is \coNP-hard when $\lambda$ is
fixed. The proof of this claim is provided in
Appendix~\ref{sec:proof-conp-hardness}. In the next section we shall prove the
following result which applies for when $\lambda$ is not fixed.
\begin{lemma}
	\label{lem:pspace-hardness}
	For a discount factor $\lambda \in (0,1)$, regret threshold $r \in
	\mathbb{Q}$, and weighted arena $G$, determining whether
	$\Regret{G}{\StrAllE,\StrPosA} \lhd r$, for $\lhd \in \{<,\le\}$, is
	\PSPACE-hard.
\end{lemma}

\subsubsection{Proof of Lemma~\ref{lem:pspace-hardness}}
\begin{figure}
\begin{center}
\begin{tikzpicture}[clause/.style={isosceles triangle, shape border
	rotate=90,inner sep=3pt, draw,dotted},node distance=0.5cm]
\node[ve,initial](A){};
\node[ve, right=of A, yshift=-0.5cm](B){$x_0$};
\node[ve, right=of A, yshift=0.5cm](B'){$\overline{x_0}$};
\node[va, right=of B, yshift=0.5cm](C){};
\node[ve, right=of C, yshift=-0.5cm](D){$x_1$};
\node[ve, right=of C, yshift=0.5cm](D'){$\overline{x_1}$};
\node[right=of D, yshift=0.5cm](E){$\dots$};
\node[ve,right=of E, yshift=-0.5cm](F){$x_m$};
\node[ve,right=of E, yshift=0.5cm](F'){$\overline{x_m}$};
\node[ve,right=of F, yshift=0.5cm](G){$\Phi$};

\node[below=0.5cm of B, clause,xshift=-0.8cm](X){$C_i$};
\node[below=0.5cm of B](Y){$\dots$};
\node[below=0.5cm of B, clause,xshift= 0.8cm](Z){$C_j$};

\path
(A) edge (B)
(B) edge (C)
(A) edge (B')
(B') edge (C)
(C) edge (D)
(C) edge (D')
(D) edge (E)
(D') edge (E)
(E) edge (F)
(E) edge (F')
(F) edge (G)
(F') edge (G)
(G) edge[loopright] node[el,swap]{$A$} (G)
(B) edge[bend left] (X.apex)
(B) edge[bend left] (Z.apex)
(X.apex) edge[bend left] (B)
(Z.apex) edge[bend left] (B)
;

\end{tikzpicture}
\caption{Depiction of the reduction from QBF.}
\label{fig:spine-qbf}
\end{center}
\end{figure}
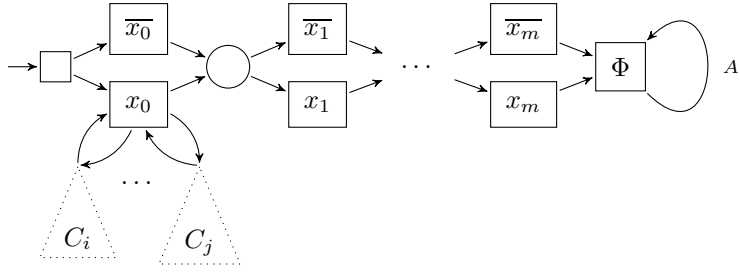

\begin{figure}
\begin{center}
\begin{tikzpicture}
\node[ve](A) at (2,6) {\small{$x_i$}};
\node(rA) at (4,7) {};
\node(lA) at (0,7) {};
\node[va](B) at (2,4) {};
\node[ve](C) at (4,3) {};
\node[va](D) at (4,1) {};
\node[ve](E) at (6,0) {\small{$\overline{x_j}$}};
\node[ve](F) at (2,0) {};
\node[va](G) at (0,1) {};
\node[ve](H) at (0,3) {};
\node[ve](I) at (-2,0) {\small{$x_k$}};
\node[ve](P) at (2,2) {};

\path
(A) edge[dotted] node[el]{} (rA)
(lA) edge[dotted] node[el]{} (A)
(A) edge[bend left] node[el]{} (B)
(B) edge[bend left] node[el]{} (A)
(B) edge node[el]{$C$} (C)
(C) edge node[el,swap]{} (P)
(C) edge node[el]{$C$} (D)
(D) edge[bend left] node[el]{} (E)
(E) edge[bend left] node[el]{} (D)
(D) edge node[el]{$C$} (F)
(F) edge node[el,swap]{} (P)
(F) edge node[el]{$C$} (G)
(G) edge node[el]{$C$} (H)
(H) edge node[el]{$C$} (B)
(G) edge[bend left] node[el]{} (I)
(I) edge[bend left] node[el,swap]{} (G)
(H) edge node[el,swap]{} (P)
(P) edge[out=45,in=135,loop,looseness=8] node[el,swap]{$B$} (P);
\end{tikzpicture}
\caption{Clause gadget for the QBF reduction for clause $x_i \lor \lnot x_j
\lor x_k$.}
\label{fig:clause-gadget}
\end{center}
\end{figure}

	The \textsc{QSAT Problem} asks whether a given fully quantified boolean
	formula (QBF) is satisfiable. The problem is known to be
	\PSPACE-complete~\cite{gj79}. It is known the result holds even if the
	formula is assumed to be in conjunctive normal form with three literals
	per clause (also known as $3$-CNF). Therefore, w.l.o.g., we consider an
	instance of the \textsc{QSAT Problem} to be given in the following form:
	\[
		\exists x_0 \forall x_1 \exists x_2 \dots \Phi(x_0, x_1,
		\dots, x_m)
	\]
	where $\Phi$ is in $3$-CNF. Let $n$ be the number of clauses from
	$\Phi$.

	In the sequel we describe how to construct, in polynomial time, a
	weighted arena in which \eve ensures regret of at most $r$ if and only
	if the QBF is true.

	We first describe the value-choosing part of the game (see
	Figure~\ref{fig:spine-qbf}). $\VtcE$ contains vertices for every
	existentially quantified variable from the QBF and $\VtcA$ contains
	vertices for every universally quantified variable. At each of this
	vertices, there are two outgoing edges with weight $0$ corresponding to
	a choice of truth value for the variable. For the variable $x_i$ vertex,
	the \textbf{true} edge leads to a vertex from which \eve can choose to
	move to any of the clause gadgets corresponding to clauses where the
	literal $x_i$ occurs (see dotted incoming edge in
	Figure~\ref{fig:clause-gadget}) or to advance to $x_{i+1}$. The
	\textbf{false} edge construction is similar. From the vertices encoding
	the choice of truth value for $x_m$ \eve can either visit the clause
	gadgets for it or move to a ``final'' vertex $\Phi \in \VtcE$. This
	final vertex has a self-loop with weight $A$.

	Our reduction works for values of $\lambda$, $r$, $A$, $B$, and $C$ such
	that the following constraints are met:
	\begin{enumerate}[$(i)$]
		\item $A < B < C$,
		\item $\lambda^2\left(\frac{C}{1 - \lambda}\right) -
			\lambda^{2nm - 2}\left(C +
			\lambda^2 \frac{B}{1-\lambda}\right) < r$,
		\item $\lambda^{2nm-2}\left(C + \lambda^2
			\frac{B}{1-\lambda}\right) >
			\lambda^2\left(\frac{C \frac{1-\lambda^4}{1-\lambda}}{1 -
			\lambda^8}\right)$,
		\item $\lambda^2\left(C + \lambda^2 \frac{B}{1-\lambda}\right) -
			\lambda^{2nm}\left(\frac{A}{1-\lambda}\right) < r$, and
		\item $\lambda^{2nm-2}\left(\frac{C}{1-\lambda}\right) -
			\lambda^{2nm}\left(\frac{A}{1-\lambda}\right) \ge r$.
	\end{enumerate}
	(See below for a sample concrete assignment.)

	\paragraph*{Value-choosing strategies.}
	To conclude the proof, we describe the strategy of \eve which ensures
	the desired property if the QBF is satisfiable and a strategy of \adam
	which ensures the property is falsified otherwise.
	
	Assume the QBF is true. It follows that there is a strategy of
	the existential player in the QBF game such that for any strategy of the
	universal player the QBF will be true after they both choose values for
	the variables. \eve now follows this strategy while visiting all clause
	gadgets corresponding to occurrences of chosen literals. At every gadget
	clause she visits she chooses to enter the gadget. If \adam
	now decides to take the weight $C$ edge, \eve can go to the center-most
	vertex and obtain a payoff of at least
	\[
		\lambda^{2nm -2}\left(C + \lambda^2 \frac{B}{1-\lambda}\right),
	\]
	with equality holding if \adam helps her at the very last clause visit
	of the very last variable gadget. In this case,
	the claim holds by $(i)$. We therefore focus in
	the case where \adam chooses to take \eve back to the vertex from which
	she entered the gadget. She can now go to the next clause gadget and
	repeat. Thus, when the play reaches vertex $\Phi$, \eve must have
	visited every clause gadget and \adam has chosen to disallow a weight
	$C$ edge in every gadget. Now \eve can ensure a payoff value of
	$\lambda^{2nm}(\frac{A}{1-\lambda})$ by
	going to $\Phi$. As she has witnessed that in every clause gadget there
	is at least one vertex in which \adam is not helping her, alternative
	strategies might have ensured a payoff of at most $\lambda^{2}
	(C + \lambda^2 \frac{B}{1-\lambda})$, by playing to the center of some
	clause gadget, or
	\[
		\lambda^2\left(\frac{C
			\frac{1-\lambda^4}{1-\lambda}}{1-\lambda^8}\right)
	\]
	by playing in and out of some adjacent clause gadgets. By $(iii)$, we
	know it suffices to show that the former is still not enough to make the
	regret of \eve at least $r$. Thus, from $(iv)$, we get that her regret
	is less than $r$.

	Conversely, if the universal player had a winning strategy (or, in other
	words, the QBF was not satisfiable) then the strategy of \adam consists
	in following this strategy in choosing values for the variables and
	taking \eve out of clause gadgets if she ever enters one. If the play
	arrives at $\Phi$ we have that there is at least one clause gadget that
	was not visited by the play. We note there is an alternative strategy of
	\eve which, by choosing a different valuation of some variable, reaches
	this clause gadget and with the help of \adam achieves value of at least 
	$\lambda^{2nm-2}(\frac{C}{1-\lambda})$. Hence, by $(v)$,
	this strategy of \adam ensures regret of at least $r$. If \eve avoids
	reaching $\Phi$ then she can ensure a value of at most $0$, which means
	an even greater regret for her.
	\qed

	\paragraph*{Example assignment.}
	For completeness, we give one assignment of the positive rationals
	$\lambda$, $r$, $A$, $B$, and $C$ which satisfies the inequalities.
	It will be obvious the chosen values can be encoded
	into a polynomial number of bits w.r.t. $n$ and $m$.

	We can assume, w.l.o.g., that $2 \le 2m \le n$.  Intuitively, we want
	values such that $(i)$ $A < B < C$ and such that the discount factor
	$\lambda$ is close enough to $1$ so that going to the center of a clause
	gadget at the end of the value-choosing rounds, is preferable for \adam
	compared to doing some strange path between adjacent clauses---this is captured
	by item $(iii)$. A $\lambda$ which is close to $1$ also gives us item
	$(v)$ from $(i)$. In order to ensure \eve wins if she does visit the
	center of a clause gadget, we also would like to have $C - A < r
	\lambda^{-2} (1-\lambda)$, which would imply items $(ii)$ and $(iv)$
	from the inequality list. It is not hard to see that the following
	assignment satisfies all the inequalities:
	\begin{itemize}
		\item $\lambda \defeq 1 - \frac{1}{2^{n^3}}$,
		\item $A \defeq 2$,
		\item $B \defeq 3$,
		\item $C \defeq 4$, and
		\item $r \defeq 3(2^{n^6} - 1)$.
	\end{itemize}

\subsubsection{Proof of
	Theorem~\ref{thm:conp-hardness}}\label{sec:proof-conp-hardness}

\begin{figure}
\begin{center}
%\resizebox{0.4\textwidth}{!}{%
\begin{tikzpicture}
	\node[va] (A) at (0,1.5) {$v$};
	\node[ve] (B) at (2,1.5) {};
	\node[va] (C) at (4,1.5) {$t_1$};
	\node[va] (D) at (2,0) {};
	\node[va] (E) at (4,0) {$s_2$};
	
	\path
	(A) edge (B)
	(B) edge (C)
	(B) edge (D)
	(D) edge[loopleft] (D)
	(D) edge (E)
	;
\end{tikzpicture}
%}
\caption{Regret gadget for $2$-disjoint-paths reduction.}
\label{fig:2dp-gadget}
\end{center}
\end{figure}
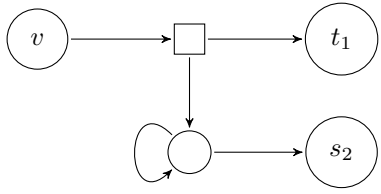

	The \textsc{$2$-disjoint-paths Problem} on directed
	graphs is known to be \NP-complete~\cite{eilam-tzoreff98}.
	We sketch how to translate a given instance of the
	\textsc{$2$-disjoint-paths Problem} into a weighted arena in which \eve
	can ensure regret value of $0$ if, and only if, the answer
	to the \textsc{$2$-disjoint-paths Problem} is negative.

	Consider a directed graph $G$ and distinct vertex pairs $(s_1,t_1)$ and
	$(s_2,t_2)$. W.l.o.g. we assume that for all $i \in \{1,2\}$:
	\begin{inparaenum}[$(i)$]
		\item $s_i \neq t_i$,
		\item $t_i$ is reachable from $s_i$, and
		\item $t_i$ is a sink (\ie has no outgoing edges).
	\end{inparaenum}
	in $G$. We now describe the changes we apply to $G$ in order to get the
	underlying graph structure of the weighted arena and then comment on the
	weight function. Let all vertices from $G$ be \adam vertices and $s_1$
	be the initial vertex. We replace all edges $(v, t_1)$ incident on $t_1$
	by a copy of the gadget shown in Figure~\ref{fig:2dp-gadget}. Next, we
	add self-loops on $t_1$ and $t_2$ with weights $A$ and $B$,
	respectively. Finally, the weights of all remaining edges are $0$.
	Our reduction works for any value of $A$ and $B$ such that
	\begin{enumerate}[$(i)$]
		\item $\lambda^{|V|} \frac{A}{1-\lambda} > r$, and
		\item $\lambda^{|V|} \frac{B}{1-\lambda} - \lambda
			\frac{A}{1-\lambda} > r$.
	\end{enumerate}
	For instance, consider $\alpha \defeq \frac{r+1}{\lambda^{|V|}}$. It is
	easy to verify that setting $A \defeq (1-\lambda)\alpha$ and $B\defeq
	(1-\lambda)\alpha^2$ satisfies the inequalities. Furthermore, $A$ and
	$B$ are rational numbers which can be represented using a polynomial
	number of bits w.r.t. $|V|$ and the size of the representation of both
	$\lambda$ and $r$.

	We claim that, in this new weighted arena, \eve can ensure a regret
	value of $0$ if in $G$ the vertex pairs $(s_1,t_1)$ and $(s_2,t_2)$
	cannot be joined by vertex-disjoint paths.  If, on the contrary, there
	are vertex-disjoint paths joining the pairs of vertices, then \adam can
	ensure a regret value strictly greater than $r$.  Indeed, we claim that
	the strategy that minimizes the regret of \eve is the strategy that, in
	states where \eve has a choice, tells her to go to
	$t_1$.
	
	First, let us prove that this strategy has regret $0$
	if, and only if, there are no two paths disjoint paths in the graph between
	the pairs of states $(s_1,t_1)$, $(s_2,t_2)$. Assume there are no
	disjoint paths, then if \adam chooses to always avoid $t_1$ then the
	regret is $0$. If $t_1$ is reached, then the choice of \eve ensures a
	value of at least $\lambda^{|V|} \frac{A}{1-\lambda}$.
	The only alternative strategy of \eve is to have chosen to go to
	$s_2$. As there are no disjoint paths, we know that either the path
	constructed from $s_2$ by \adam never reaches $t_2$, and then the value
	of the path is $0$ and the regret is $0$ for \eve or the path
	constructed from $s_2$ reaches $t_1$ again, and so the regret is also
	equal to $0$ since the discount factor ensures the value of this play is
	lower than the one realized by the current strategy of \eve.
	Now assume that there are disjoint paths, if \eve would
	have chosen to put the game in $s_2$ (instead of choosing $t_1$) then
	\adam has a strategy which allows \eve to reach $t_2$ and get a payoff
	of at least $\lambda^{|V|} \frac{B}{1-\lambda}$ while she achieves at
	most $\lambda \frac{A}{1 - \lambda}$. From $(i)$ we have that the regret
	in this case is greater than $r$.

	To conclude the proof, let us show that any other strategy of \eve has a
	regret greater than $0$. Indeed, if \eve decides to go to $s_2$ (instead
	of choosing to go to $t_1$) then \adam can choose to loop on $s_2$ and
	the payoff in this case is $0$. The regret of \eve is non-zero in this
	case since she could have achieved at least $\lambda^{|V|}
	\frac{A}{1-\lambda}$ by going to $t_1$. It follows from $(ii)$ that this
	ensures a regret value greater than $r$.
	\qed

\section{Missing Proofs From Section~\ref{sec:eloquent-adversary}}\label{app:eloquent-adversary}

\subsection{Proof of Theorem~\ref{thm:epsilon-gap-pspace}}
	We reduce the problem to determining the winner of a reachability game
	on an exponentially larger arena.  Although the arena is exponentially
	larger, all paths are only polynomial in length, so the winner can be
	determined in alternating polynomial time, or equivalently, polynomial
	space.  

	The idea of the construction is as follows.  Given a discounted-sum
	automaton $\mathcal{A}$, we determinize its transitions via a subset
	construction, to obtain a deterministic, multi-valued discounted-sum
	automaton $D_{\mathcal{A}}$.  Then we decide if Eve is able to simulate,
	within the regret bound, the $D_{\mathcal{A}}$ on $\mathcal{A}$ for all
	\emph{finite} words up to a length (polynomially) dependent on
	$\epsilon$.  If we simulate the automaton for a sufficient number of
	steps, then any significant gap between the automata will be
	unrecoverable regardless of future inputs, and we can give a
	satisfactory answer for the $\epsilon$-gap regret problem.

	More formally, given a discounted-sum automaton $\mathcal{A} = (Q,q_0,A,
	\delta, w)$, a regret value $r$ and a precision $\epsilon>0$, we
	construct a reachability game $G_\mathcal{A}^\epsilon(r)$ as follows.
	Let 
	\[
		N \defeq \left\lfloor
		\log_{\lambda}\left(\frac{\epsilon(1-\lambda)}{4W}\right)\right\rfloor +
		1,
	\]
	where $W$ is the maximum absolute value weight occurring in
	$\mathcal{A}$, so that $ \frac{\lambda^N \cdot W}{1-\lambda} <
	\frac{\epsilon}{4}$.  Let $P = \{\discfun{\lambda}(\pi) \st \pi\in
	Q^*\text{ is a finite run of $\mathcal{A}$ with }|\pi|\leq N\}$ denote
	the (finite) set of possible discounted payoffs of words of length at
	most $N$. Let $\mathcal{F}$ be the set of functions $f:Q \to
	\mathbb{R}\cup\{\bot\}$, and for $f \in \mathcal{F}$, let $\supp{f} =
	\{q \in Q \st f(q) \neq \bot\}$.  Intuitively, each $f \in \mathcal{F}$
	represents a weighted subset of $Q$ ($\supp{f}$ being the corresponding
	unweighted subset), where $f(q)$ for $q \in \supp{f}$ corresponds to the
	maximal weight over all (consistent) paths ending in $q$ (scaled by a
	power of $\lambda$). Given $f \in \mathcal{F}$ and $\alpha \in A$
	the $\alpha$-successor of $f$ is the function $f_\alpha$ defined as:
	\[
		f_\alpha(q') \defeq
		\begin{cases}
			\displaystyle
			\max_{\substack{q \in \supp{f}\\(q,\alpha,q') \in
			\delta }} \{\lambda^{-1}\cdot f(q) +
			w(q,\alpha,q') \} & \text{if this set is not empty}\\
			\bot & \text{otherwise.}
		\end{cases}
	\] 

	We define $\mathcal{F}_0 = \{f_0\}$ where $f_0(q_0) = 0$ and $f_0(q) =
	\bot$ for all $q \neq q_0$; and for all $n\geq 0$, we define
	$\mathcal{F}_{n+1} \defeq \{f_\alpha \st f \in \mathcal{F}\text{ and
	}\alpha \in A\}$.  For convenience, let $F = {\bigcup}_{i=0}^N
	\mathcal{F}_i$ (considered as a disjoint union).

	The game $G_{\mathcal{A}}^\epsilon(r) = (V,\VtcE,E,v_0,T)$ is defined as
	follows:
	\begin{itemize}
		\item $V = (Q \times F \times P) \cup (Q \times F \times P \times
			A)$;
		\item $\VtcE = (Q \times F  \times P \times A)$;
		\item $\big((q,f,c),(q,f,c,\alpha)\big) \in E$ for all $q \in
			Q$, $f \in F\setminus \mathcal{F}_N$, $c \in P$, and
			$\alpha \in A$;
		\item $\big((q,f,c,\alpha),(q',f',c')\big) \in E$ for all $q,q'
			\in Q$, $f \in F\setminus \mathcal{F}_N$, $c \in P$, and
			$\alpha \in A$ such that $(q,\alpha,q') \in
			\delta$, $f' = f_\alpha$, and $c'=c+\lambda \cdot
			w(q,\alpha,q')$;
		\item $v_0 = (q_0,f_0,0)$; and
		\item $(q,f,c) \in T$ if, and only if, $f \in \mathcal{F}_N$ and
			$\max_{s \in \supp{f}}\lambda^{N-1}\cdot f(s) \leq
			c+r+\frac{\epsilon}{2}$. 
	\end{itemize}

	We claim that determining the winner of $G_\mathcal{A}^\epsilon(r)$ yields a
	correct response for the $\epsilon$-gap promise problem.

	\begin{claim}\label{cla:claim1}
		Let $G_{\mathcal{A}}^\epsilon(r)$ be
		defined as above.  Then:
	\begin{itemize}
		\item If Eve wins $G_{\mathcal{A}}^\epsilon(r)$ then
			$\Regret{\mathcal{A}}{\StrE,\StrWordA} \le r+\epsilon$, and
		\item if Adam wins $G_{\mathcal{A}}^\epsilon(r)$ then
			$\Regret{\mathcal{A}}{\StrE,\StrWordA} > r$.
	\end{itemize}
	\end{claim}
	\begin{proof}[Proof of Claim~\ref{cla:claim1}]
		It is easy to see that a play of $G_{\mathcal{A}}^\epsilon(r)$
		results in Adam choosing a word $w \in A^*$ of length $N$,
		and Eve selecting a run, $\pi$, of $w$ on $\mathcal{A}$ by
		resolving non-determinism at each symbol.  Further, if the play
		terminates at $(q,f,c)$ then $c=\discfun{\lambda}(\pi)$ and, as
		$f$ contains the maximal weights of all paths (scaled by a power
		of $\lambda$), $\mathcal{A}(w) = \lambda^{N-1}(\max_{s \in
		\supp{f}} f(s))$. Since $|w|=N$ we have, for any
		infinite word $w' \in A^\omega$ and for any run,
		$\pi'$, of $\mathcal{A}$ on $w'$ from $q$, $\pi'$:
		\begin{eqnarray*}
			|\mathcal{A}(w\cdot w')-\mathcal{A}(w)| &\leq&
				\frac{\lambda^N \cdot W}{1-\lambda} <
				\frac{\epsilon}{4}, \text{ and}\\
			|\discfun{\lambda}(\pi\cdot \pi') -
				\discfun{\lambda}(\pi)|&\leq& \frac{\lambda^N \cdot
				W}{1-\lambda} < \frac{\epsilon}{4}.
		\end{eqnarray*}

		It follows that:
		\begin{equation}\label{eqn:approx}
			(\mathcal{A}(w)-\discfun{\lambda}(\pi)) -
			\frac{\epsilon}{2} <
			\mathcal{A}(w\cdot w') - \discfun{\lambda}(\pi\cdot\pi') <
			(\mathcal{A}(w)-\discfun{\lambda}(\pi)) + \frac{\epsilon}{2}.
		\end{equation}

		Now suppose Eve wins $G_\mathcal{A}^\epsilon(r)$.  Then, for
		every word $w$ with $|w|=N$, Eve has a strategy $\sigma$ that
		construct a run, $\pi$, on $\mathcal{A}$ such that
		$\mathcal{A}(w) \leq \discfun{\lambda}(\pi) + r +
		\frac{\epsilon}{2}$.  We extend this strategy to infinite words
		by playing arbitrarily after the first $N$ symbols.  It follows
		from Equation~\ref{eqn:approx} that for every infinite word
		$\hat{w}$, the resulting run, $\hat{\pi}$,
		\[
			\mathcal{A}(\hat{w}) - \discfun{\lambda}(\hat{\pi}) <
			(\mathcal{A}(w)-\discfun{\lambda}(\pi)) + \frac{\epsilon}{2}
			\leq r + \epsilon.
		\]
		Since $\regret{\sigma}{\StrE,\StrWordA}{\mathcal{A}} =
		\sup_{\hat{w} \in A^\omega}(\mathcal{A}(\hat{w}) -
		\discfun{\lambda}(\pi))$, we have
		$\Regret{\mathcal{A}}{\StrE,\StrWordA} \le r+\epsilon$.

		Conversely, suppose Adam wins $G_\mathcal{A}^\epsilon(r)$.  Then
		for any strategy of Eve, Adam can construct a word $w$, with
		$|w|=N$ such that the run, $\pi$, of $\mathcal{A}$ on $w$
		determined by Eve's strategy satisfies $\mathcal{A}(w) >
		\discfun{\lambda}(\pi) + r + \frac{\epsilon}{2}$.  Again, from
		Equation~\ref{eqn:approx} it follows that for any infinite word
		$\hat{w}$ with $w$ as its prefix and any
		consistent run $\pi'$,
		\[
			\mathcal{A}(\hat{w}) - \discfun{\lambda}(\hat{\pi}) >
			(\mathcal{A}(w)-\discfun{\lambda}(\pi)) -
			\frac{\epsilon}{2} > r.
		\]
		As this is valid for any strategy of Eve, we have
		$\Regret{\mathcal{A}}{\StrE,\StrWordA}>r$ as required.
	\end{proof}

	Now every path in $G_\mathcal{A}^\epsilon(r)$ has length at most $N$,
	and as the set of successors of a given state can be computed on-the-fly
	in polynomial time, the winner can be determined in alternating
	polynomial time.  Hence a solution to the $\epsilon$-gap promise problem is
	constructible in polynomial space.

\subsection{Proof of Theorem~\ref{thm:eloquent-pspace-hardness-epsilon}}

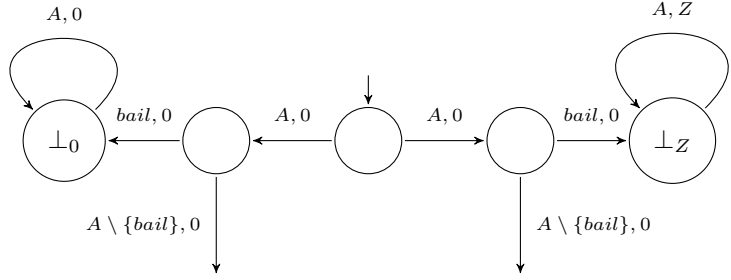
\begin{figure}
\begin{center}
\begin{tikzpicture}
\node[state] (sink0) at (0,2) {{$\bot_0$}};
\node[state] (linter) at (2,2) {};
\node[state,initial above] (vi) at (4,2) {};
\node[state] (rinter) at (6,2) {};
\node[state] (sink2) at (8,2) {{$\bot_Z$}};
\node (lbinter) at (2,0) {};
\node (rbinter) at (6,0) {};

\path
(sink0) edge[loop] node[el,swap] {$A,0$} (sink0)
(sink2) edge[loop] node[el,swap] {$A,Z$} (sink2)
(vi) edge node[el,swap] {$A,0$} (linter)
(linter) edge node[el,swap] {$bail, 0$} (sink0)
(vi) edge node[el] {$A,0$} (rinter)
(rinter) edge node[el] {$bail, 0$} (sink2)
(linter) edge node[el,swap] {$A\setminus \{bail\},0$} (lbinter)
(rinter) edge node[el] {$A\setminus \{bail\},0$} (rbinter)
;
\end{tikzpicture}
\caption{Initial gadget used in reduction from QBF.}
\label{fig:initial-gadget}
\end{center}
\end{figure}

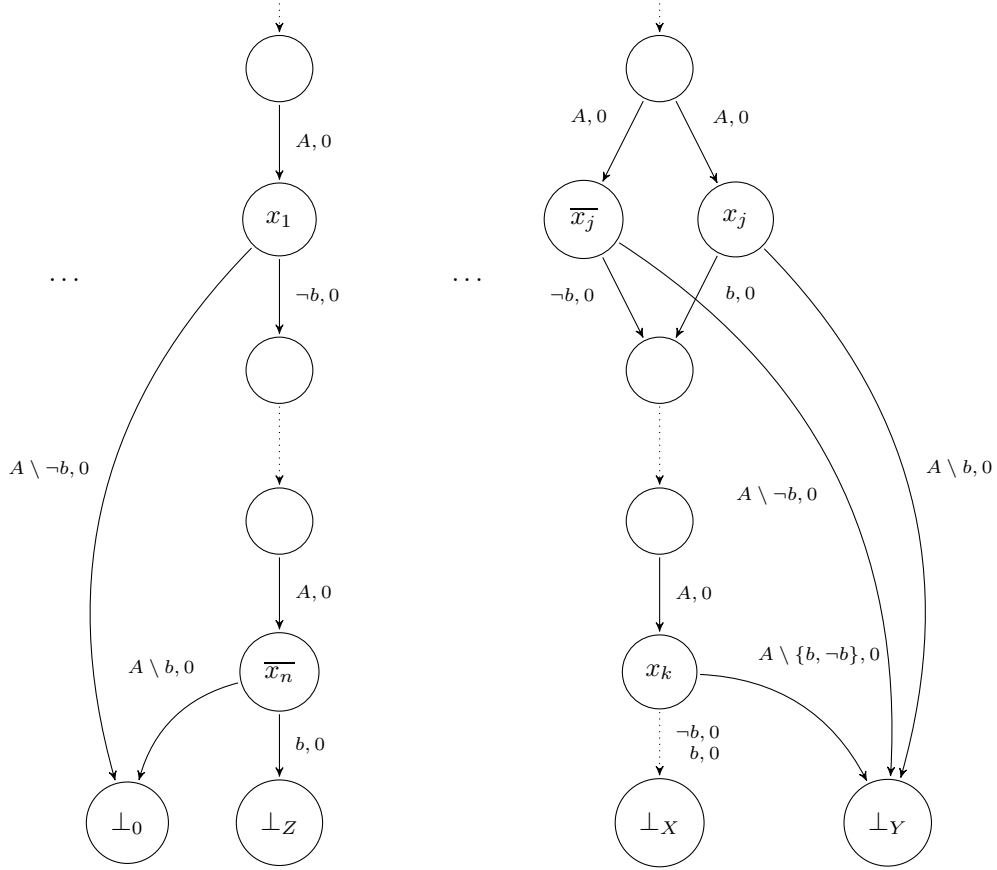
\begin{figure}
\begin{center}
\begin{tikzpicture}[every initial by arrow/.style={dotted}]
\node at (-2.8,7.2) {\dots};

\node[state] (sink0) at (-2,0) {$\bot_0$};
\node[state] (sinkZ) at (0,0) {$\bot_Z$};
\node[state] (xnleft) at (0,2) {$\overline{x_n}$};
\node[state] (axnleft) at (0,4) {};
\node[state] (bx1left) at (0,6) {};
\node[state] (x1left) at (0,8) {$x_1$};
\node[state,initial above] (ax1left) at (0,10) {};

\node at (2.5,7.2) {\dots};

\node[state] (sinkX) at (5,0) {$\bot_X$};
\node[state] (sinkY) at (8,0) {$\bot_Y$};
\node[state] (xnright) at (5,2) {$x_k$};
\node[state] (axnright) at (5,4) {};
\node[state] (bx1right) at (5,6) {};
\node[state] (notx1right) at (4,8) {$\overline{x_j}$};
\node[state] (x1right) at (6,8) {$x_j$};
\node[state,initial above] (ax1right) at (5,10) {};

\path
(ax1left) edge node[el]{$A,0$} (x1left)
(x1left) edge node[el]{$\lnot b,0$} (bx1left)
(x1left) edge[bend right] node[el,swap]{$A \setminus \lnot b, 0$} (sink0)
(bx1left) edge[dotted] (axnleft)
(axnleft) edge node[el]{$A,0$} (xnleft)
(xnleft) edge node[el]{$b,0$} (sinkZ)
(xnleft) edge[bend right,pos=0.2] node[el,swap]{$A \setminus b, 0$} (sink0)

(ax1right) edge node[el,swap]{$A,0$} (notx1right)
(ax1right) edge node[el]{$A,0$} (x1right)
(notx1right) edge node[el,swap,pos=0.2]{$\lnot b,0$} (bx1right)
(notx1right) edge[bend left] node[el,swap]{$A \setminus \lnot b, 0$} (sinkY)
(x1right) edge[pos=0.2] node[el]{$b,0$} (bx1right)
(x1right) edge[bend left] node[el]{$A \setminus b, 0$} (sinkY)
(bx1right) edge[dotted] (axnright)
(axnright) edge node[el]{$A,0$} (xnright)
(xnright) edge[dotted] node[el,align=right]{$\lnot b,0$\\$b,0$} (sinkX)
(xnright) edge[bend left,pos=0.2] node[el]{$A \setminus \{b, \lnot b\}, 0$} (sinkY)
;
\end{tikzpicture}
\caption{Left and right sub-arenas of the reduction from QBF. Clause $i$ shown
on the left; existential and universal gadgets for variables $x_j$ and $x_k$,
respectively, on the right.}
\label{fig:clause-gadgets}
\end{center}
\end{figure}

	Given an instance of the \textsc{QSAT Problem} -- a fully quantified
	boolean formula (QBF) -- we construct, in polynomial time, a weighted
	arena such that the answer to the regret threshold problem is positive
	if, and only if, the QBF is true. The main idea behind our reduction is
	to build an arena with two disconnected sub-graphs joined by an initial
	gadget in which we force \eve to go into a specific sub-arena.
	In order for her to ensure the regret is not too high she
	must now make sure all alternative plays in the other part of the arena
	do not achieve too high values. In the sub-arena where \eve finds
	herself, we will simulate the choice of values for the boolean variables
	from the QBF while in the other sub-arena these choices will affect which
	alternative paths can achieve high discounted-sum values based on the
	clauses of the QBF. We describe the reduction for $\le$. It will be clear
	how to extend the result to $<$.

	The \textsc{QSAT Problem} asks whether a given fully quantified boolean
	formula (QBF) is satisfiable. The problem is known to be
	\PSPACE-complete~\cite{gj79}. It is known the result holds even if the
	formula is assumed to be in conjunctive normal form with three literals
	per clause (also known as $3$-CNF). Therefore, w.l.o.g., we consider an
	instance of the \textsc{QSAT Problem} to be given in the following form:
	\[
		\exists x_0 \forall x_1 \exists x_2 \dots \Phi(x_0, x_1,
		\dots, x_n)
	\]
	where $\Phi$ is in $3$-CNF.

	We now give the details of the construction. Our reduction works for
	values of positive rationals $r$, $X$, $Y$, and $Z$ such that
	\begin{enumerate}[$(i)$]
		\item $\lambda^2 \frac{Z}{1 - \lambda} > r + \epsilon$,
		\item $\lambda^{2n} \frac{Z}{1-\lambda} - \lambda^{2n}
			\frac{X}{1-\lambda} > r + \epsilon$,
		\item $\lambda^{2n} \frac{Z}{1-\lambda} - \lambda^{2n}
			\frac{Y}{1-\lambda} \le r$,
		\item $\lambda^3 \frac{Y}{1-\lambda} - \lambda^{2n}
			\frac{X}{1-\lambda} \le r$.
	\end{enumerate}
	The alphabet of the new weighted arena is $A = \{bail, b, \lnot b\}$.

	\paragraph*{Example assignment.} In order to convince the reader
	that values which satisfy the above inequalities indeed exist
	for all possible valuations of $n$ and $\epsilon$ we give such
	a valuation. Let $f : \mathbb{Q} \to \mathbb{Q}$ be
	defined as $f(x) \defeq \frac{(1-\lambda)x}{\lambda^{2n}}$. Note that,
	w.l.o.g., we can assume that $n \ge 2$. Consider the valuation
	\begin{itemize}
		\item $r \defeq \lambda^{3-2n} (1+\epsilon)$,
		\item $Z \defeq f(r + \epsilon + 2)$,
		\item $X \defeq f(1)$,
		\item $Y \defeq f(2+\epsilon)$.
	\end{itemize}
	Clearly, inequalities $(i)$--$(iii)$ hold. Regarding $(iv)$, it will be
	useful to consider the equivalent inequality
	\[
		\lambda^{3 - 2n}Y - X \le \frac{r(1-\lambda)}{\lambda^{2n}}.
	\]
	We observe that the LHS is smaller than $\lambda^{3-2n}(Y - X)$.
	Furthermore the difference $Y - X$ is equivalent to $\frac{(1 +
	\epsilon)(1-\lambda)}{\lambda^{2n}}$. Finally, by choice of $r$ we have
	that the RHS is equivalent to
	\[
		\lambda^{3-2n}\left(\frac{(1+\epsilon)(1-\lambda)}{\lambda^{2n}} \right).
	\]
	Hence, $(iv)$ holds as well. Note that the chosen values can be encoded
	into a polynomial number of bits w.r.t. $\lambda$ and $n$ as well as the
	size of the representation of $\epsilon$.

	\paragraph*{Initial gadget.} The weighted arena we construct starts as is
	shown in Figure~\ref{fig:initial-gadget}. Here, \eve has a to make a
	choice: she can go left or right. If she goes left, then \adam can play
	$bail$ and force her into $\bot_0$ giving her a value of $0$ while an
	alternative play goes into $\bot_Z$ achieving a value of $\lambda^2
	\frac{Z}{1-\lambda}$. By $(i)$ we get that the regret of this strategy
	is greater than $r + \epsilon$.  Thus, we can assume that \eve will
	always play to the right.

	\paragraph*{Choosing values.} For each existentially quantified variable
	$x_i$ we will create a ``diamond gadget'' to allow \eve to choose a
	different state depending on the value she wants to assign to $x_i$.
	From the corresponding states, \adam will have to play $b$ or $\lnot b$,
	respectively, otherwise he allows her to get to $\bot_Y$. For
	universally quantified variables we have a $2$-transition path which
	allows \adam to choose $b$ or $\lnot b$ (in the second step). The right
	path shown in Figure~\ref{fig:clause-gadgets} depicts this construction.
	From $(iii)$ it follows that if \adam cheats at any point during this
	simulation of value choosing phase of the QSAT game, then the play
	reaches $\bot_Y$ and the regret is at most $r$. Hence, we can
	assume that \adam does not cheat and the play eventually reaches
	$\bot_X$.  Observe that the choice of values in this gadget is made as
	follows: at turn $2i$ after having entered the gadget, the value of
	$x_i$ is decided.

	\paragraph*{Clause gadgets.} For every clause from $\Phi$ we create a path
	in the new weighted arena such that every literal $\ell_i$ in the clause
	is synchronized with the turn at which the value of $x_i$ is decided in
	the value-choosing gadget. That is to say, there are $2i - 1$ states
	that must be visited before arriving at the state corresponding to
	$\ell_i$. At state $\ell_i$, if the value of $x_i$ corresponding to
	literal $\ell_i$ is chosen, the play deterministically goes to $\bot_0$.
	Otherwise, traversal of the clause-path continues.
		
	It should be clear that if the QBF is true, then \eve has a
	value-choosing strategy such that at least one literal from every clause
	holds.  That means that every alternative play in the left sub-arena of
	our construction has been forced into $\bot_0$ while \eve has ensured a
	discounted-sum value of $\lambda^{2n} \frac{X}{1-\lambda}$ by reaching
	$\bot_X$. From $(iv)$ it follows that \eve has ensured a regret of at
	most $r$. Conversely, if \adam has a value-choosing strategy
	in the QSAT problem so the QBF is show to be false, then he can use his
	strategy in the constructed arena so that some alternative path in the
	left sub-arena eventually reaches $\bot_Z$. In this case, from $(ii)$ we
	get that the regret value is greater than $r + \epsilon$, as expected.
	\qed

\subsection{Proof of Theorem~\ref{thm:eloquentadversary-zero}}
	\subparagraph{Membership.}
	Consider a fixed weighted automaton $\calA = (Q,q_I,A,\Delta,w)$
	and a discount factor $\lambda \in (0,1)$. Further, we suppose
	the regret of $\calA$ is $0$.
	
	Let us start by defining a
	set of values which, intuitively, represent lower bounds on the regret
	\eve can get by resolving the non-determinism of $\calA$ on the fly.
	First, let us introduce some additional notation. Define $\calA^q \defeq
	(Q,q,A,\Delta,w)$, \ie the automaton $\calA$ with new initial state $q$.
	For states $q,q' \in Q$, let $\mu(q,q') \defeq \sup\left(\{
	\calA^{q'}(x) - \calA^q(x) \st x \in A^\omega \} \cup
	\{0\}\right)$. We are now ready to describe our set of values:
	\[
		M \defeq \{ |w(p,\sigma,q') - w(p,\sigma,q) + \lambda \cdot
			\mu(q,q')| \st p \in Q \text{ and } q,q' \text{ are }
			\sigma \text{-successors of } p \}.
	\]

	Note that since $\calA$ is assumed to be total (\ie, every state-action
	pair has at least one successor) then $M$ cannot be empty.  Observe
	that, by definition, $M$ only contains non-negative values. 
	Since $\calA$ has regret $0$, then we know that for all
	$d \in (0,1)$, there is a strategy $\sigma_d$ of \eve such that
	$\regret{\sigma_d}{\calA}{\StrAllE, \StrWordA} = 0$. If $M \neq \{0\}$,
	we let $\epsilon < \lambda^{|Q|} \cdot \left(\min M \setminus
	\{0\}\right)$. Denote by $\tilde{Q}$ the set of states reachable from
	$q_I$ by reading some finite word $x$ of length at most $|Q|$,\ie $x \in
	A^{\le |Q|}$, according to $\sigma_\epsilon$. If $M = \{0\}$, let
	$\tilde{Q} = Q$. We now define a memoryless strategy $\sigma$ of \eve as
	follows: if $M = \{0\}$ then $\sigma$ is arbitrary, otherwise
	$\sigma(p,a) = q$ implies $q \in \tilde{Q}$. To conclude, we then show
	that $\sigma$ ensures regret $0$.\qed

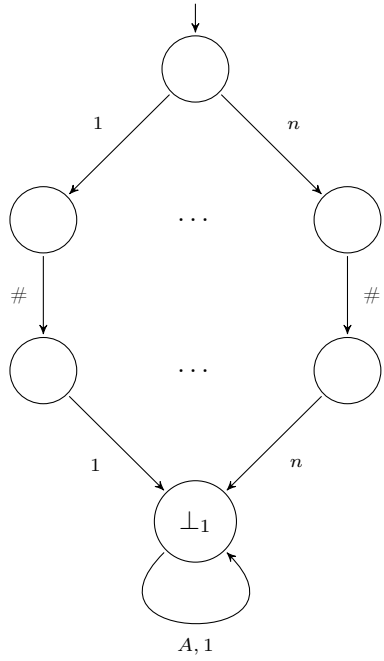
\begin{figure}
\begin{center}
\begin{tikzpicture}
\node[state,initial above](I) at (2,6) {};
\node[state](A) at (0,4) {};
\node(B) at (2,4) {\dots};
\node[state](C) at (4,4) {};
\node[state](D) at (0,2) {};
\node(E) at (2,2) {\dots};
\node[state](F) at (4,2) {};
\node[state](G) at (2,0) {$\bot_1$};

\path
(I) edge node[el,swap]{$1$} (A)
(I) edge node[el]{$n$} (C)
(A) edge node[el,swap]{\#} (D)
(C) edge node[el]{\#} (F)
(D) edge node[el,swap]{$1$} (G)
(F) edge node[el]{$n$} (G)
(G) edge[loopbelow] node[el,swap]{$A,1$} (G)
;

\end{tikzpicture}
\caption{Clause choosing gadget for the SAT reduction. There are as many paths
from top to bottom ($\bot_1$) as there are clauses ($n$).}
\label{fig:det-clause-chooser}
\end{center}
\end{figure}

\begin{figure}
\begin{center}
\begin{tikzpicture}
\node[state,initial above](I) at (3,6) {};
\node[state](A) at (1,4) {$x_1$};
\node[state](At) at (0,2) {$1_{true}$};
\node[state](Af) at (2,2) {$1_{false}$};
\node[state](C) at (5,4) {$x_2$};
\node[state](Ct) at (4,2) {$2_{true}$};
\node[state](Cf) at (6,2) {$2_{false}$};
\node[state](G) at (3,-1) {$\bot_1$};

\path
(I) edge node[el,swap]{$1,2,3$} (A)
(I) edge node[el]{$1,2,3$} (C)
(A) edge node[el,swap]{\#} (At)
(A) edge node[el]{\#} (Af)
(C) edge node[el,swap]{\#} (Ct)
(C) edge node[el]{\#} (Cf)
(At) edge node[el,swap]{$1$} (G)
(Af) edge node[pos=0.4,el,swap]{$2,3$} (G)
(Ct) edge node[pos=0.4,el]{$1,2$} (G)
(Cf) edge node[el]{$3$} (G)
(G) edge[loopbelow] node[el,swap]{$A,1$} (G)
;

\end{tikzpicture}
\caption{Value choosing gadget for the SAT reduction. Depicted is the
configuration for $(x_1 \lor x_2) \land (\lnot x_1 \lor x_2)
\land(\lnot x_1 \lor \lnot x_2)$.}
\label{fig:nondet-value-chooser}
\end{center}
\end{figure}

	\subparagraph{Hardness.}
	We give a reduction from the \textsc{SAT} problem, i.e. satisfiability
	of a CNF formula. The construction presented is based on a proof
	in~\cite{akl10}. The idea is simple: given boolean formula $\Phi$ in CNF
	we construct a weighted automaton $\Gamma_\Phi$ such that \eve can
	ensure regret value of $0$ with a positional strategy in $\Gamma_\Phi$
	if and only if $\Phi$ is satisfiable. Note that this restriction of \eve
	to positional strategies is no loss of generality. Indeed,
	we have shown that if the regret of a game against an eloquent adversary
	is $0$, then she has a positional strategy with regret $0$.

	Let us now fix a boolean formula $\Phi$ in CNF with $n$ clauses and $m$
	boolean variables $x_1,\ldots,x_m$. The weighted automaton
	$\Gamma_\Phi = (Q, q_I, A, \Delta, w)$ has alphabet $A = \{bail,\#\}
	\cup \{i \st 1 \le i \le n\}$.  $\Gamma_\Phi$ includes an initial gadget
	such as the one depicted in Figure~\ref{fig:initial-gadget}. Recall that
	this gadget forces \eve to play into the right sub-arena. As the left
	sub-arena of $\Gamma_\Phi$ we attach the gadget depicted in
	Figure~\ref{fig:det-clause-chooser}. All transitions shown have weight
	$1$ and all missing transitions in order for $\Gamma_\Phi$ to be
	complete lead to a state $\bot_0$ with a self-loop on every symbol from
	$A$ with weight $0$. Intuitively, as \eve must go to the right sub-arena
	then all alternative plays in the left sub-arena correspond to either
	\adam choosing a clause $i$ and spelling $i \# i$ to reach $\bot_1$ or
	reaching $\bot_0$ by playing any other sequence of symbols. The right
	sub-arena of the automaton is as shown in
	Figure~\ref{fig:nondet-value-chooser}, where all transitions shown have
	weight $1$ and all missing transitions go to $\bot_0$ again. Here, from
	$q_0$ we have transitions to state $x_j$ with symbol $i$ if the $i$-th
	clause contains variable $x_j$. For every state $x_j$ we have
	transitions to $j_{true}$ and $j_{false}$ with symbol $\#$. The idea is
	to allow \eve to choose the truth value of $x_j$. Finally, every state
	$j_{true}$ (or $j_{false}$) has a transition to $\bot_1$ with symbol $i$ if
	the literal $x_j$ (resp. $\lnot x_j$) appears in the $i$-th
	clause.

	The argument to show that \eve can ensure regret of $0$ if and only if
	$\Phi$ is satisfiable is straightforward. Assume the formula is indeed
	satisfiable. Assume, also, that \adam chooses $1 \le i \le n$
	and spells $i \# i$. Since $\Phi$ is satisfiable there is a choice of
	values for $x_1,\ldots,x_m$ such that for each clause
	there must be at least one literal in the $i$-th clause which makes the
	clause true. \eve transitions, in the right sub-arena from $q_0$ to the
	corresponding value and when \adam plays $\#$ she chooses the correct
	truth value for the variable. Thus, the play reaches $\bot_1$ and, as
	$W = 1$ in the left and right sub-arenas of $\Gamma_\Phi$,
	it follows that her regret is $0$. Indeed, her payoff will be
	$\lambda^2/(1-\lambda)$---recall the first two turns are spent in the
	initial gadget, where all transitions leading to both sub-arenas are
	$0$-weighted---which is the maximal payoff obtainable in either
	sub-arena.
	If \adam does not
	play as assumed then we know all plays in $\Gamma_\Phi$ reach $\bot_0$
	and again her regret is $0$. Note that this strategy can be realized
	with a positional strategy by assigning to each $x_j$ the choice of
	truth value and choosing from $q_0$ any valid transition for all $1 \le
	i \le n$.
	
	Conversely, if $\Phi$ is not satisfiable then for every valuation of
	variables $x_1,\ldots, x_m$ there is at least one clause which is not
	true. Given any positional strategy of \eve in $\Gamma_\Phi$ we can
	extract the corresponding valuation of the boolean variables. Now \adam
	chooses $1 \le i\le n$ such that the $i$-th clause is not satisfied by
	the assignment. The play will therefore end in $\bot_0$ while an
	alternative play in the left sub-arena will reach $\bot_1$. Hence the
	regret of \eve in the game is non-zero.
	\qed

\end{document}

%% file: macros.tex
\usepackage{mathtools}
\usepackage{amssymb}
\usepackage{cite}
\usepackage{tikz}
\usepackage{graphicx}
\usepackage{paralist}
\usepackage[basic]{complexity}
\usepackage{authblk}
\usepackage[T1]{fontenc}
\usepackage{enumerate}
\usepackage{xspace}
\usepackage{algorithm}
\usepackage{algorithmic}
\usepackage[obeyDraft]{todonotes}

\makeatletter
\let\phi\varphi
\let\epsilon\varepsilon
\makeatother

\newclass{\APSPACE}{APSPACE}
\renewclass{\AP}{APTIME}
\renewclass{\P}{PTIME}
\newclass{\AEXPTIME}{AEXPTIME}
\newclass{\ATIME}{ATIME}
\newclass{\ASPACE}{ASPACE}
\renewclass{\EXP}{EXPTIME}

\usetikzlibrary{positioning,fit,arrows,automata,calc,shapes,decorations.pathmorphing} 
\tikzset{->,>=stealth',shorten >=1pt,shorten <=1pt,auto,node distance=1cm,
every loop/.style={looseness=6},
initial text={},
every fit/.style={draw,densely dotted,rectangle},
el/.style={font=\scriptsize},
inner sep=2mm,
ve/.style={rectangle, draw},
va/.style={circle, draw},
loopright/.style={loop,looseness=6,out=-45, in=45},
loopleft/.style={loop,looseness=6,out=135, in=225},
loopabove/.style={loop,looseness=6,out=45, in=135},
loopbelow/.style={loop,looseness=6,out=-135, in=-45},
snake it/.style={decorate, decoration=snake},
}

\newcommand{\calA}{\mathcal{A}}

\newcommand{\defeq}{\vcentcolon=}%{\mathrel{\mathop:}=}

\renewcommand{\restriction}{\mathord{\downharpoonright}}
\newcommand{\arestriction}{\mathord{\downharpoonright}}
\providecommand\st{}
\renewcommand{\st}{\mathrel{:}}
\newcommand{\pow}{\mathcal{P}}
\newcommand{\eve}{Eve\xspace}
\newcommand{\adam}{Adam\xspace}

\newcommand{\ie}{\textit{i.e.}\xspace}
\newcommand{\eg}{\textit{e.g.}\xspace}

\newcommand{\discfun}[1]{\mathsf{DS}_{#1}}